\documentclass[letterpaper,12pt,titlepage,oneside,final]{book}
 
\newcommand{\href}[1]{#1} 

\usepackage{ifthen}
\newboolean{PrintVersion}
\setboolean{PrintVersion}{false} 

\usepackage{amsmath,amssymb,amsthm,amstext,bbm} 
\usepackage[pdftex]{graphicx} 

\usepackage[pdftex,pagebackref=false]{hyperref} 
\hypersetup{
    plainpages=false,       
    unicode=false,          
    pdftoolbar=true,        
    pdfmenubar=true,        
    pdffitwindow=false,     
    pdfstartview={FitB},    
    pdftitle={uWaterloo\ LaTeX\ Thesis\ Template},    
    pdfauthor={Jeff Hnybida},    
    pdfsubject={Physics},  
    pdfkeywords={Quantum Gravity} {Topological Field Theory} {Lattice Gauge Theory}, 
    pdfnewwindow=true,      
    colorlinks=true,        
    linkcolor=blue,         
    citecolor=green,        
    filecolor=magenta,      
    urlcolor=cyan           
}
\ifthenelse{\boolean{PrintVersion}}{   
\hypersetup{	
    citecolor=black,%
    filecolor=black,%
    linkcolor=black,%
    urlcolor=black}
}{} 

\let\origdoublepage\cleardoublepage
\newcommand{\clearemptydoublepage}{%
  \clearpage{\pagestyle{empty}\origdoublepage}}
\let\cleardoublepage\clearemptydoublepage

\newcommand{\C}{{\mathbb C}}
\newcommand{\N}{{\mathbb N}}
\newcommand{\R}{{\mathbb R}}

\newcommand{\bbT}{{\mathbb T}}
\newcommand{\one}{\mathbbm{1}}

\newcommand{\cG}{{\mathcal G}}

\newcommand{\cL}{{\mathcal L}}
\newcommand{\cH}{{\mathcal H}}

\newcommand{\cC}{{\mathcal C}}

\newcommand{\SU}{\mathrm{SU}}

\def\T{T^{\Gamma}}
\def\G{{\cal G}_{\Gamma}}

\newcommand{\be}{\begin{equation}}
\newcommand{\ee}{\end{equation}}
\newcommand{\beq}{\begin{eqnarray}}
\newcommand{\eeq}{\end{eqnarray}}
\newcommand{\bea}{\begin{eqnarray}}
\newcommand{\eea}{\end{eqnarray}}
\newcommand{\beal}{\begin{align}}
\newcommand{\eeal}{\end{align}}
\newcommand{\nn}{\nonumber}

\newcommand{\bra}{\langle}
\newcommand{\ket}{\rangle}
\newcommand{\la}{\langle}
\newcommand{\ra}{\rangle}

\newcommand{\tr}{{\mathrm Tr}}

\newcommand{\rd}{\mathrm{d}}

\newcommand{\bpm}{\begin{pmatrix}}
\newcommand{\epm}{\end{pmatrix}}
\newcommand{\bvm}{\begin{vmatrix}}
\newcommand{\evm}{\end{vmatrix}}

\def\nn{\nonumber}

\newcommand{\corner}{\,\raisebox{-0.5ex}{\resizebox{15pt}{!}{\mbox{\huge{$\lrcorner$}}}}\,}

\newtheorem{theorem}{Theorem}[section]
\newtheorem{lemma}[theorem]{Lemma}
\newtheorem{proposition}[theorem]{Proposition}
\newtheorem{corollary}[theorem]{Corollary}
\newtheorem{definition}[theorem]{Definition}

\newcommand{\sixj}[6]{
\left\{
\begin{array}{ccc}
#1 & #2 & #3 \\
#4 & #5 & #6
\end{array}
\right\}
}

\newcommand{\threej}[6]{
\left(
\begin{array}{ccc}
#1 & #2 & #3 \\
#4 & #5 & #6
\end{array}
\right)
}

\begin{document}


\pagestyle{empty}
\pagenumbering{roman}

\begin{titlepage}
        \begin{center}
        \vspace*{1.0cm}

        \Huge
        {\bf Generating Functionals for Spin Foam Amplitudes }

        \vspace*{1.0cm}

        \normalsize
        by \\

        \vspace*{1.0cm}

        \Large
        Jeff Hnybida \\

        \vspace*{3.0cm}

        \normalsize
        A thesis \\
        presented to the University of Waterloo \\ 
        in fulfillment of the \\
        thesis requirement for the degree of \\
        Doctor of Philosophy \\
        in \\
        Physics \\

        \vspace*{2.0cm}

        Waterloo, Ontario, Canada, 2014 \\

        \vspace*{1.0cm}

        \copyright\ Jeff Hnybida 2014 \\
        \end{center}
\end{titlepage}

\pagestyle{plain}
\setcounter{page}{2}

\cleardoublepage 


\begin{center}\textbf{Author's Declaration}\end{center}

  \noindent
I hereby declare that I am the sole author of this thesis. This is a true copy of the thesis, including any required final revisions, as accepted by my examiners.

  \bigskip
  
  \noindent
I understand that my thesis may be made electronically available to the public.

\cleardoublepage


\begin{center}\textbf{Abstract}\end{center}

Various approaches to Quantum Gravity such as Loop Quantum Gravity, Spin Foam Models and Tensor-Group Field theories use invariant tensors on a group, called intertwiners, as the basic building block of transition amplitudes.  For the group SU(2) the contraction of these intertwiners in the pattern of a graph produces what are called spin network amplitudes, which also have various applications and a long history.

We construct a generating functional for the exact evalutation of a coherent representation of these spin network amplitudes.  This generating functional is defined for arbitrary graphs and depends only on a pair of spinors for each edge.  The generating functional is a meromorphic polynomial in the spinor invariants which is determined by the cycle structure of the graph.

The expansion of the spin network generating function is given in terms of a newly recognized basis of SU(2) intertwiners consisting of the monomials of the holomorphic spinor invariants.  This basis is labelled by the degrees of the monomials and is thus discrete.  It is also overcomplete, but contains the precise amount of data to specify points in the classical space of closed polyhedra, and is in this sense coherent.  We call this new basis the discrete-coherent basis.  

We focus our study on the 4-valent basis, which is the first non-trivial dimension, and is also the case of interest for Quantum Gravity.  We find simple relations between the new basis, the orthonormal basis, and the coherent basis.  

The 4-simplex amplitude in the new basis depends on 20 spins and is referred to as the 20j symbol.  We show that by simply summing over five of the extra spins produces the 15j symbol of the orthonormal basis.  On the other hand, the 20j symbol is the exact evaluation of the coherent 4-simplex amplitude.

The asymptotic limit of the 20j symbol is found to give a generalization of the Regge action to Twisted Geometry.  By this we mean the five glued tetrahedra in the 4-simplex have different shapes when viewed from different frames.  The imposition of the matching of shapes is known to be related to the simplicity constraints in spin foam models.   

Finally we discuss the process of coarse graining moves at the level of the generating functionals and give a general prescription for arbitrary graphs.  A direct relation between the polynomial of cycles in the spin network generating functional and the high temperature loop expansion of the 2d Ising model is found.  

\cleardoublepage


\begin{center}\textbf{Acknowledgements}\end{center}

I would like to thank my supervisor Laurent Freidel for his inspiration and guidance.  I was fortunate to learn so much about physics, but also much more.

I am indebted to my collaborators Joseph Ben Geloun, John Klauder, Bianca Dittrich, Andrzej Banburski, Linqing Chen.  I would also like to thank the members of my committee for their support: Lee Smolin, Freddy Cachazo, Achim Kempf, Robb Mann, Florian Girelli, and Niayesh Afshordi.  Special thanks go to Mark Shegelski for his ongoing mentorship and friendship.

I would like to thank all the members of the Perimeter Institute, in particular Debbie Guenther, Joy Montgomery, George Arvola, Dan Lynch, Dawn Bombay, and to all the friends I've made throughout the years.  

Finally I couldn't have done it without my family and their never ending encouragement and support.  Thanks to Laurie Hnybida, Roxanne Heppner, and Marianne Humphries for coming all the way to Waterloo to see the defence.


\cleardoublepage


\begin{center}\textbf{Dedication}\end{center}
\vspace{150pt}
\begin{center}{ In memory of my father, Ted Hnybida.} \end{center}

\cleardoublepage

\renewcommand\contentsname{Table of Contents}
\tableofcontents
\cleardoublepage
\phantomsection


\addcontentsline{toc}{chapter}{List of Figures}
\listoffigures
\cleardoublepage
\phantomsection		


\pagenumbering{arabic}

\chapter{Introduction}

The work contained in this thesis was developed as part of the broader effort to quantize Einstein's theory of General Relativity.  While General Relativity can be treated successfully as a quantum field theory at low energy (small curvature) it is notoriously nonrenormalizable.  This likely implies that either new unknown physics is yet to be discovered, perhaps near the Planck scale, or that non-perturbative methods need to be considered.  To add to this challenge, there is currently very little experimental data with which to guide the theory.  For this reason it is precarious to stray too far from the conventional interpretations of Quantum Mechanics and General Relativity. 

So far a theory of Quantum Gravity has yet to be widely accepted.  There are currently several well established programs in pursuit of this goal.  These include simplicial path integral methods such as Quantum Regge Calculus \cite{Rocek:1982fr} and Causal Dynamical Triangulations \cite{Ambjorn:2010rx}.  Various approaches via canonical quantization include Loop Quantum Gravity \cite{Rovelli:2004tv}, Asymptotic Safety \cite{Benedetti:2009rx}, and String Theory \cite{Polchinski:1998rq}.  

The most popular approach to Quantum Gravity has always been String Theory partially owing to its much broader ambitions.  However, likely due to its nonconservatism the theory has been found to admit a large number of possible vacuum states, which casts doubt on its potential predictive power.  This ambiguity is known as the Landscape Problem and has yet to be reconciled.  String theory has had undeniable success and the Landscape problem is not necessarily insurmountable, but until a theory of Quantum Gravity is finally confirmed via experiment, it is most sensible to investigate all other approaches which might be more restricted in their assumptions.

The canonical quantization of Loop Quantum Gravity has had success in constructing a kinematical Hilbert space and deriving the spectra of geometrical operators \cite{Rovelli:1994ge}.  Intriguingly, the kinematical Hilbert space was found to be spanned by the spin network states of Penrose \cite{Rovelli:1995ac} representing the quantum states of space in much the way he had envisioned.  The dynamics of the theory was expected to be describe by the evolution of these spin networks states, but this still has yet to be realised and progress in the canonical picture has since stalled.  The other (and usually much simpler) path to dynamics is of course via the path integral.  Hence, interest in covariant quantum histories of spin networks began to subsequently emerge and was an early motivation for what are now called Spin Foam models.

Spin Foam models are in fact a synthesis of the various non-perturbative formulations of Quantum Gravity: they correspond classically to simplicial discretizations of spacetime, they compute transition amplitudes between spin networks on the boundary \`a la Loop Quantum Gravity, and their amplitudes are weighted by the (cosine of) the Regge action.  Furthermore, a partition function for these spin foam amplitudes can be derived somewhat rigorously from a {\it constrained} topological gauge theory known as BF theory.  The computation, coarse graining, and geometric characterisation of the amplitudes of BF theory will be the main focus of this thesis. 

BF theory is a topological field theory defined in all dimensions by the action
\be
  S_{BF}(B,\omega) \equiv \int_\mathcal{M} \: \bra B \wedge F(\omega) \ket \:\rd x^n 
\ee  
where the $n-2$ form $B$ acts as a Lagrange multiplier for the vanishing of the curvature 2-form $F = \rd \omega + \omega \wedge \omega$.  Here $\bra \cdot \ket$ is the Killing form on the Lie algebra of the gauge group.  See Chapter \ref{chapter_BF} for more details on BF theory.

In the rest of the introduction we give a quick review of pure 3d Quantum Gravity without cosmological constant, in particular the Ponzano-Regge model, and its relation to BF theory.  We then discuss the coherent representation of BF theory for which most of the results of this thesis pertain to.  Finally we discuss the path to 4d Quantum Gravity and end with a discussion of renormalization.

\subsubsection{3d Quantum Gravity}


In three dimensions General Relativity is topological and is precisely of the BF type.  Indeed, for Riemannian signature $\eta = \text{Diag}(1,1,1)$ we define a 1-form frame field $e$ in terms of the metric by
\be
  g_{\mu \nu} = \eta_{ij} e^{i}_\mu e^{j}_\nu \qquad\qquad e = e^{i}_\mu \rd x^\mu \tau_i
\ee   
where $\tau_j = (i/2) \sigma_j$ are the $\mathfrak{su}(2)$ generators.  The connection 1-form $\omega$  has curvature 
\be
  F(\omega) = \rd \omega + \omega \wedge \omega \qquad\qquad \omega = \omega^{i}_\mu \rd x^{\mu} \tau_i
\ee 
and the Einstein-Hilbert action is given by integrating the scalar curvature
\be
  S(e,\omega) = \frac{1}{\kappa} \int_\mathcal{M} \: \text{tr} \left( e \wedge F(\omega) \right)
\ee
The equations of motion enforce the compatibility of the metric and frame field and the flatness of the connection 
\be
  \rd_\omega e = 0 \qquad \qquad  F(\omega) = 0 
\ee
where $\rd_\omega$ denotes the covariant derivative.  With the additional condition of non-degeneracy of the frame field, this formulation is equivalent to the Einstein-Hilbert action for 3d General Relativity, without matter, and without a cosmological constant.  

Under local $SU(2)$ gauge transformations we have
\be \label{eqn_rot_gauge}
  e \mapsto g e g^{-1} \qquad \omega \mapsto g \omega g^{-1} + (\rd g)g^{-1}
\ee
which correspond to the local rotation symmetry.  In addition, there exist an extra set of gauge transformations corresponding to a translation of the frame field
\be \label{eqn_shift_gauge}
  e \mapsto e + \rd_\omega \phi \qquad \omega \mapsto \omega
\ee
which is due to the second Bianchi identity $\rd_\omega F = 0$.  

Combined with the equation of motion $\rd_\omega e = 0$ the translation gauge symmetry implies that all solutions are locally trivial $e \approx 0$ up to gauge transformations.  Hence the only degrees of freedom are global which is the reason this theory is referred to as a Topological Field Theory.  A combination of the rotational and translational gauge symmetries can be shown to correspond to the diffeomorphism symmetry of General Relativity on-shell \cite{Freidel:2002dw}.

Being topological, 3d gravity can be equivalently described in terms of a finite set of data on, for instance, a triangulation $\Delta$ homeomorphic to the manifold $\mathcal{M}$.  This data is usually taken to be the parallel transports between tetrahedra and the integrals of the frame field over edges of the triangulation, approximating the connection and frame field respectively.  This is a classical formulation of simplicial gravity for which the model of Ponzano and Regge \cite{ponzano1968semiclassical} is based.

Upon quantization the holonomies act by group multiplication while the frame fields act via the differential operators $\vec{J}$.  The frame field operators are interpreted as the quantization of edge vectors of $\Delta$ and their Casimir thus gives their norm (or length).  The spectrum of the Casimir is given by the unitary irreducible representations $j=1/2,1,3/2,...$ of SU(2) which shows that length in this system is quantized.

Finally, every triple of edge vectors meeting at a node must be invariant under the local rotational gauge transformations (\ref{eqn_rot_gauge}).  Furthermore, there is only one invariant rank three tensor on SU(2) up to normalization:  The Wigner 3j symbol (or Clebsch-Gordan coefficient).\footnote{In this thesis we refer to invariant tensors on SU(2) as intertwiners; for an explanation why see section \ref{section_CG_int}.}  

The 3j symbol thus has the interpretation as a quantum triangle and its three spins correspond to the lengths of its three edges, which close to form a triangle due to the SU(2) invariance.  Contracting four 3j symbols in the pattern of a tetrahedron gives the well-known 6j symbol which is the amplitude for each tetrahedron of $\Delta$.  

This leads to the Ponzano-Regge partition function of 3d gravity with zero cosmological constant
\be \label{eqn_PR_part}
  Z_{BF}^{\Delta} = \sum_{\{j_e\}} \prod_{\text{edges} \: e} (-1)^{2j_e} (2j_e+1) \prod_{\text{triangles} \: t} (-1)^{\sum_{e \cap t} j_e} \prod_{\text{tetrahedra} \: \tau} \{6j\}_{\tau}
\ee  
where $\{6j\}_{\tau}$ is the 6j symbol with the appropriate coloring by $j_e$ of the six edges of the tetrahedron $\tau$.  The signs and the factors $2j_e+1$ are necessary for topological invariance.  A derivation of the BF partition function will be given in the next chapter, explaining the origin of each of these factors.

If $\Delta$ has boundary then (\ref{eqn_PR_part}) computes the transition amplitude for the 2d geometry defined by the spins on the boundary.  This sum, however, does not always converge, partly due to the residual noncompact gauge symmetry (\ref{eqn_shift_gauge}) \cite{Freidel:2004vi} and partly due to topological (or potentially other) reasons \cite{Barrett:2008wh}.  In more generality these divergences can be shown to be attributed to degeneracies in the map between the continuous and discrete connections \cite{Bonzom:2010ar,Bonzom:2010zh,Bonzom:2011br}.

In the semiclassical limit the edge lengths are defined in terms of the spins by $l_e = j_e + 1/2$.  Taking all the spins to be uniformly large, i.e. $\lambda j_e$ where $\lambda \rightarrow \infty$, the 6j symbol scales like
\be \label{eqn_6j_asym}
  \{6j\} \sim \frac{\lambda^{-\frac{3}{2}}}{\sqrt{12 V(l_e)}} \cos\left( \lambda S(l_e) + \frac{1}{4}  \right)
\ee
where $S(l_e)$ is the Regge action.  

The Regge action was first proposed by Regge \cite{Regge:1961px} as an approximation to the Einstein-Hilbert action on a triangulation of a manifold.  It is defined in terms of the edge lengths $l_e$ and the dihedral angles $\Theta_e$ between the triangles sharing the edge $e$ by
\be \label{eqn_Regge_action_intro}
  S(l_e) = \sum_{e \subset \tau} l_e \Theta_e
\ee
As the triangulation is refined in a suitable way the discretized equations of motion converge to Einsteins equation $R_{\mu \nu} = 0$.  Note that the Regge action can be generalized to higher dimensions and Lorentzian signature.

It was the semiclassical limit (\ref{eqn_6j_asym}) that originally motivated Ponzano and Regge to take the 6j symbol as the quantum building block of 3d gravity and the connection with BF theory was discovered much later.  The derivation of the partition function (\ref{eqn_PR_part}) in terms of simplicial amplitudes from BF theory also generalizes to higher dimensions and Lorentzian signature.

\subsubsection{Coherent BF Theory}

The coherent intertwiners, introduced first by Girelli and Livine \cite{Girelli:2005ii} and then by Livine and Speziale \cite{Livine:2007vk}, are simply a coherent state representation of the space of invariant tensors on SU(2).  The power of the coherent representation cannot be overstated; the exact evaluations we compute in Chapter \ref{chapter_exact} are a result of a special exponentiating property of coherent states. 

Each SU(2) coherent state is labeled by a spinor $|z\ket \in \C^2$ where $|z]$ denotes its contragradient version.  We use a bra-ket notation for the spinors
\be
  |z\ket \equiv \bpm z^0 \\ z^1 \epm, \qquad |z] \equiv \bpm -\bar{z}^{1} \\ \bar{z}^0 \epm
\ee
such that given two spinors $z$ and $w$ the two invariants which can be formed by contracting with either epsilon or delta are denoted
\be
  [z|w\ket = z^0 w^1 - z^1w^0, \qquad \bra z | w \ket \equiv \bar{z}^{0} w^{0} + \bar{z}^1 w^1 
\ee
Since we will work heavily with polynomials of these invariants, this notation is more clear than the usual index notation.

The exponentiating property of the coherent states corresponds to the fact that the spin $j$ representation is simply the tensor product of $2j$ copies of the spinor $|z\ket^{\otimes 2j}$. 
A coherent rank $n$ tensor on SU(2) is therefore the tensor product of $n$ exponentiated spinors $\otimes_{i=1}^{n}|z_i\ket^{\otimes 2j_i}$ where the dimension of the $i$'th representation is $2j_i+1$.  To make the coherent tensor invariant we group average using the Haar measure 
\be \label{eqn_coh_amp}
  \|j_i,z_i\ket \equiv \int \rd g \: \bigotimes_{i=1}^{n} \frac{g|z_i\ket^{\otimes 2j_i}}{(2j_i)!} 
\ee
which is the definition of the Livine-Speziale intertwiner (up to normalization, which we choose to be $1/(2j_i)!$).  These states span the intertwiner space and are thus equally well suited to represent the BF theory parition function (\ref{eqn_PR_part}).

The coherent 6j symbol is constructed by contracting $4$ coherent intertwiners (\ref{eqn_coh_amp}) in the pattern of a tetrahedron.  Labeling each vertex by $i=1,...,4$ and edges by pairs $(ij)$ this amplitude depends on 6 spins $j_{ij} = j_{ji}$ and 12 spinors $|z^{i}_{j}\ket \neq |z^{j}_{i}\ket$ where the upper index denotes the vertex and the lower index the connected vertex.  Thus the coherent amplitude in 3d is given by
\be \label{eqn_coh_3S}
  \mathcal{A}_{3S}(j_{ij},z^{i}_{j}) \equiv \int_{SU(2)^4} \prod_{i=1}^{4} \rd g_i \prod_{1\leq i<j \leq 4} \frac{[z^{i}_{j}|g_{i}^{-1}g_{j}|z^{j}_{i}\ket^{2j_{ij}}}{(2j_{ij})!}
\ee

The asymptotics of the coherent amplitude have been studied extensively, however the actual evaluation of these amplitudes was not known.  While the asymptotic analysis is important to check the semi-classical limit, the exact evaluation could be useful to study recursion relations \cite{Bonzom:2009zd}, coarse graining moves, or to perform numerical calculations.

To obtain the exact evaluation we use a special property\footnote{See Lemma \ref{eqn_SU2_lemma}, which was first shown in \cite{Livine:2011gp}.} of the Haar measure on SU(2) to express the group integrals in (\ref{eqn_coh_3S}) as Gaussian integrals.  The price for converting group integrals into integrals over $\C^2$ are factors of $1/(J_i+1)!$ for each vertex where $J_i = \sum_j j_{ij}$.  Thus we define a generating functional as
\be \label{eqn_3S_gen}
  \mathcal{A}_{3S}(z^{i}_{j}) \equiv \sum_{\{j_{ij}\}} \mathcal{A}_{3S}(j_{ij},z^{i}_{j}) \prod_i (J_i+1)!
\ee

Remarkably, we are able to compute the Gaussian integrals in (\ref{eqn_3S_gen}), not just for the tetrahedral graph (\ref{eqn_coh_3S}) but for any {\it arbitrary graph}.  Performing the Gaussian integrals produces a determinant depending purely on the spinors.  Even more remarkably we find that the determinant can be evaluated in general and can be expressed in terms of loops of the spin network graph.  

For example, after integration and evaluating the determinant, the generating functional (\ref{eqn_3S_gen}) of the 3-simplex takes the form
\be \label{eqn_3S_loops}
  \mathcal{A}_{3S}(z^{i}_{j}) = \frac{1}{\left(1-A_{123} - A_{124} - A_{134} - A_{234} + A_{1234} - A_{1243} - A_{1324} \right)^2}
\ee
where $A_{12...p} = [z^{1}_{p}|z^{1}_{2}\ket[z^{2}_{1}|z^{2}_{3}\ket \cdots [z^{p}_{p-1}|z^{p}_{1}\ket$.  Each term in the sum is a cycle of the tetrahedron.  The signs in (\ref{eqn_3S_loops}) are determined by a convention defined in Theorem \ref{thm_amp} in terms of the orientations of the graph, which are implicit in the definition (\ref{eqn_coh_3S}).  The general result for an arbitrary graph is similar in that the sum contains a term for every set of cycles which do not share vertices or edges.  See Theorem \ref{thm_amp}. 

Expanding (\ref{eqn_3S_loops}) in a power series and comparing terms of the same homogeneity in (\ref{eqn_3S_gen}) determines a Racah formula for the evaluation of the group integrals in (\ref{eqn_coh_3S}).  This allows us to define a Racah formula for arbitrary graphs.  In the case of (\ref{eqn_coh_3S}) we find the Racah formula for the 6j symbol
\be \label{eqn_A3S_6j}
  \mathcal{A}_{3S}(j_{ij},z^{i}_{j}) = (-1)^{s} \{6j\} \prod_{a\neq i<j} \left( \frac{\prod_{i<j}[z^{a}_{i}|z^{a}_{j}\ket^{k^{a}_{ij}}}{\sqrt{(J_a+1)! \prod_{i<j} k^{a}_{ij}!}} \right)
\ee
where the integers $k^{a}_{ij} = k^{a}_{ji}$ are (by homogeneity) solutions of the equations $\sum_{j\neq i,a} k^{a}_{ij} = 2j_i$\footnote{For trivalent graphs there is a unique solution given by $k^{a}_{ij} = j_{ai} + j_{aj} - j_{ak}$ where $a,i,j,k$ are all distinct.} and
the sign $s = j_{12}+j_{13}$ comes from differing orientations compared with the conventional definition of the 6j symbol.

The expansion of the generating functional (\ref{eqn_3S_loops}) is expressed as polynomials in the fundamental holomorphic invariants 
\be
  [z_i|z_j\ket \equiv  
	\epsilon_{AB} z^{A}_i z^{B}_j = \bpm z_{i}^{0} & z_{i}^{1} \epm \bpm 0 & 1 \\ -1 & 0 \epm \bpm z_{j}^{0} \\ z_{j}^{1} \epm = z_{i}^{0}z_{j}^{1} - z_{j}^{0} \\ z_{i}^{1}
\ee
This suggests that given a set of $n$ spinors $\{z_i\}$ perhaps the most natural basis of intertwiners (invariant tensors on SU(2)) should be a monomial of invariants of the form

\be \label{eqn_k_basis}
  (z_i \| k_{ij} \ket \equiv \prod_{i<j} \frac{[z_{i}|z_{j}\ket^{k_{ij}}}{k_{ij}!}
\ee
with the conventional normalization.  In 3d this basis is proportional to the Wigner 3j symbol as it must, but in higher dimensions this gives a new basis of intertwiners.  In Chapter \ref{section_discrete_coherent} we study the algebraic and geometric properties of this basis in 4d, which is the first non-trivial case and the case of interest for Quantum Gravity.  

We find that while this basis is discrete, being labeled by the finite set of integers $\{k_{ij}\}$, it also coherent.  By this we mean that the states represent accurately the classical degrees of freedom, i.e. the space of polyhedra \cite{Conrady:2009px}.  In 4d there are six $\{k_{ij}\}$ which provide the correct amount of information to uniquely define a tetrahedron.  Compare this to the five spins of the orthonormal basis of the same space.  For this reason we refer to this basis as the discrete-coherent basis, or just the $k$-basis.

The $k$-basis (\ref{eqn_k_basis}) was first considered by Bargmann \cite{bargmann1962representations} also in the context of generating functionals.  There he also derived (\ref{eqn_3S_loops}) by other methods, but specifically for the 6j symbol.  Bargmann's work is built upon earlier work by Schwinger \cite{schwinger2001angular} who also considered generating functionals of the 6j and 9j symbols in his harmonic oscillator variables (see Section \ref{sec_su2}).  

Various other spin network generating functionals have also been developed since Bargmann's work.  These generating functionals differ in that they are either restricted to trivalent/planar graphs or they compute a slightly different spin network amplitude called the chromatic evaluation.
\footnote{In 1975 Labarthe \cite{Labarthe:1975yf} developed a set of Feynman rules for computing a 3nj-symbol generating functional for arbitrary trivalent graphs.  Then in 1998 Westbury found a closed formula for the generating functional of the chromatic evaluation on planar, trivalent graphs  \cite{westbury1998generating} and also shortly after by Schnetz \cite{schnetz1998generating}.  Finally, more recently Garoufalidis \cite{Garoufalidis:2009vi} proved the existence of an asymptotic limit of the chromatic evaluation while Costantino and Marche \cite{costantino2011generating} solved the asymptotic evaluation and also generalized to non-trivial holonomies.}

In Chapter \ref{section_discrete_coherent} we find many interesting relations between the $k$-basis, the orthonormal basis, and the coherent basis.  For example,  in four dimensions we find that while the orthonormal basis is labeled by one extra spin, the $k$-basis is labeled by two extra spins.  Summing over one of these extra spins produces an orthonormal state.  This feature also generalizes to higher dimensions.

In this way, amplitudes constructed from the $k$-basis are more fundamental than those constructed with the orthonormal basis. 

On the other hand, as mentioned above, the generating functionals allow us to evaluate the group integrals in (\ref{eqn_coh_3S}); the evaluation is given in the $k$-basis as in (\ref{eqn_A3S_6j}).  These evaluations are precisely the amplitudes of $k$-basis contractions, which are a new type of amplitude.  In other words, the $k$-basis amplitudes provide the exact evaluation of the group integrals of the coherent amplitudes.

As shown in Chapter \ref{chapter_semi} these amplitudes also give a generalization of Ponzano and Regge's asymptotic formula (\ref{eqn_Regge_action_intro}).\footnote{In 4d the amplitude of a 4-simplex in BF theory is called a 15j symbol.  Certain special cases of the large spin limit of the 15j symbol have recently been computed \cite{Bonzom:2011cy,Yu:2011is}.  However, the uniformly large spin limit analogous to (\ref{eqn_6j_asym}) has not been found until now \cite{Freidel:2013fia} which is a corollary to one of the main results of Chapter \ref{chapter_semi}, Theorem \ref{thm_twisted_action} .  Also, asymptotics of a coherent version of the 15j symbol have been computed recently \cite{Barrett:2009as} and were a major milestone in the spin foam program \cite{Conrady:2008mk,Barrett:2009gg} as this resulted in the recovery of the 4d Regge action in the semiclassical limit.  The exact evaluation of these coherent amplitudes was computed in \cite{Freidel:2012ji}, and is one of the main results of Chapter \ref{chapter_exact}.}  The asymptotic formula in four dimensions describes
 a discontinuous generalization of Regge calculus called Twisted Geometry \cite{Freidel:2010aq}.  Like Regge calculus, spacetime is discretized into a triangulation consisting of 4-simplicies each constructed out of five boundary tetrahedra.  However, in Twisted Geometry the shapes of these tetrahedra can be different when viewed from different frames.  It is only required that the areas of glued triangles match, but the shapes of the glued triangles could be wildly different.

Our action for the Twisted Geometry of a 4-simplex $\sigma$, having five boundary tetrahedra $\tau$ sharing ten triangles $t$ is of the form
\be
  S_{\mathbb{\tau}} = \sum_{t \in \sigma} A_t \Theta^{t} \qquad \Theta^{t} = \frac{1}{3} \sum_{\tau \not\supset t} \Theta^{t}_{\tau}
\ee
where $A_t$ is the area of triangle $t$ and $\Theta^{t}_{\tau}$ is the 4d dihedral angle between the two tetrahedra sharing triangle $t$ {\it when viewed from tetrahedron $\tau$}.  Thus the dihedral angles are not unique and the three $\Theta^{t}_{\tau}$ only agree when the shapes of the glued triangles are constrained to match.  Each $\Theta^{t}_{\tau}$ is a function of the boundary $\{k^{\tau}_{tt'}\}$ values

This possibility was investigated in the context of Spin Foam models by Dittrich and Ryan \cite{Dittrich:2010ey}.  This seems to suggest that this shape matching is more important in the transformation of BF theory into General Relativity than might have been previously thought.  In Section \ref{section_geo} we give a set of conditions on the boundary $k$ data which enforces this shape matching in the semi-classical limit and could be used to define a spin foam model.  See also the formulation of area-angle Regge calculus in terms of shape matching constraints \cite{Dittrich:2008va} which is also discussed in Section \ref{section_regge}.

\subsubsection{4d Quantum Gravity}

While General Relativity in four dimensions is not topological, it was discovered by Plebanski \cite{Plebanski:1977zz} that it could be formulated by a {\it constrained} four dimensional BF theory.  That is if $B$ is constrained to be of the form
\be
  B = \star(e \wedge e)
\ee  
for a real tetrad 1-form $e$ then the BF action becomes the Hilbert-Palatini action for General Relativity \cite{Peldan:1993hi}.  The crux of the spin foam program is to formulate a discretized version of these constraints can break the topological invariance of BF theory and give rise to the local degrees of freedom of gravity.  

The question of whether four dimensional quantum gravity could be formulated in an analogous way as the Ponzano-Regge model was first studied by Riesenberger \cite{Reisenberger:1997sk}, Barret and Crane \cite{Barrett:1997gw}, Baez \cite{Baez:1997zt}, Freidel and Krasnov \cite{Freidel:1998pt} and others.  The partition function (or state sum) of these models is defined on a dual cellular complex\footnote{We take the dual cellular complex $\Delta^\ast$ to be the dual of a triangulation $\Delta$ for simplicity, but more general cellular complexes can also be considered.  Further, we really only require the 2-skeleton of $\Delta^\ast$, i.e. the vertices, edges and faces.  This is a one-extra-dimensional generalization of a Feynman graph. } $\Delta^\ast$ and has the general form
\be
  Z^{\Delta^\ast}_{\text{spin foam}} = \sum_{j_f,i_e} \prod_{f} A_f(j_f) \prod_{e} A_e(j_f,i_e) \prod_{v} A_v(j_f,i_e)
\ee
where $v,e,f$ are the vertices, edges, faces of $\Delta^\ast$ and $j_f, i_e$ are combinatorial data.  The Ponzano-Regge model (\ref{eqn_PR_part}) is of this form where $A_v = \{6j\}$, $A_e = (-1)^{\sum_{f\cap e} j_f}$, and $A_f = (-1)^{2j_f}(2j_f+1)$.\footnote{Note that tetrahedra of $\Delta$ are dual to vertices of $\Delta^\ast$ and faces are dual to edges.}  Also the data $i_e$ is not necessary in this case since the Wigner 3j symbol is unique, but the set $\{i_e\}$ becomes non-trivial for four or more dimensions.  

The advantage of formulating GR as a constrained BF theory is that, instead of quantizing Plebanski's action, we can instead use the topological nature of BF theory to quantize the discretized BF action and impose the (discretized) constraints at the quantum level.  The first model of this type was proposed by Barret and Crane \cite{Barrett:1997gw}.  

While this is not a quantization of a constrained system in the sense of Dirac it is a quantization of the Gupta-Bleuler type which was realised by Livine and Speziale \cite{Livine:2007ya} and led to corrected versions of the Barret-Crane model by Engle, Livine, Pereira, Rovelli \cite{Engle:2007wy} and by Freidel, Krasnov \cite{Freidel:2007py}.

The discreteness of the boundary avoids the problem of defining the analogy of the well known Wheeler-Misner-Hawking sum over geometries \cite{Misner:1957wq}:  
\be \label{eqn_WMH}
  Z_{WMH}(g_{\mu \nu}|_{\partial M}) = \int_{\text{Metrics/Diff}} \rd \mu(g_{\mu \nu}) e^{i S_{EH}(g_{\mu \nu})}
\ee
where the measure over equivalence classes of metrics under diffeomorphisms is ill defined.  The discreteness of the Spin Foam formalism allows us to postpone this definition and construct (more or less) well-defined amplitudes.  The problem of summing over geometries is then recast into the question of ``what is to be done with the discretization?''

There are various points of view on how this should be handled.  On the one hand, the original idea from Loop Quantum Gravity was that the discrete structures are to represent quantum histories of spin networks and hence they should be summed like Feynman diagrams, in analogy with (\ref{eqn_WMH}).  

In fact, the Ponzano-Regge partition function (\ref{eqn_PR_part}) can be shown to be a Feynman amplitude of a non-local field theory over three copies of $SU(2)$ \cite{Boulatov:1992vp}:
\begin{align} \label{eqn_bulatov}
  S_{\text{Bulatov}}[\varphi] &= \frac{1}{2}\int_{\text{SU(2)}^3} \varphi(g_1,g_2,g_3) \varphi(g_3,g_2,g_1) \\
  & \hspace{60pt} + \frac{\lambda}{4!} \int_{\text{SU(2)}^6} \varphi(g_1,g_2,g_3) \varphi(g_3,g_4,g_3) \varphi(g_4,g_2,g_6) \varphi(g_6,g_5,g_1) \nonumber
\end{align}
where the field satisfies the reality condition $\overline{\varphi(g_1,g_2,g_3)} = \varphi(g_3,g_2,g_1)$ and integration is with respect to the Haar measure.  Gauge fixing the local rotation invariance amounts to a diagonal SU(2) invariance
\be
  \varphi(g_1 h, g_2 h, g_3 h) = \varphi(g_1,g_2,g_3) \qquad \forall h \in SU(2)
\ee
and so the fields can be expanded in invariant rank 3-tensors, i.e. the Wigner 3j symbol.  

The kinetic term is trivial since the 3j symbol is orthogonal while the interaction term contracts four 3j symbols into a 6j symbol.  Thus the Ponzano-Regge model (\ref{eqn_PR_part}) is reproduced for each Feynman graph.  Each Feynman amplitude, in turn corresponds to a dual 2-complex.  

One can see that the GFT sum is not only a sum over geometries, but also a sum over topologies.  One has to be careful though, since these 2-complexes are not all homeomorphic to manifolds, or even pseudo-manifolds \cite{Gurau:2010nd}.

Each 2-complex in the path integral (\ref{eqn_bulatov}) is then weighted simply by the inverse of the symmetry factor and powers of the coupling constants.   Such models can be defined in all dimensions for various group manifolds and are referred to as (Tensor) Group Field Theories \cite{Freidel:2005qe,Oriti:2006se,Oriti:2009wn}.  The problem of the continuum limit is then shifted to the renormalization of such models.  Much progress has been made in this direction \cite{Geloun:2010vj,BenGeloun:2011rc,Carrozza:2013wda}, especially with the proposal of an interesting new class of models called Coloured Group Field Theories \cite{Gurau:2009tw,Gurau:2011aq} for which the pathological pseudomanifolds are suppressed.

On the other hand, the discrete regularization of spin foam amplitudes can be viewed as part of a (background independent) Lattice Gauge theory and hence should be refined in a suitably defined sense \cite{Dittrich:2011zh,Dittrich:2012jq,Bahr:2012qj}.  While diffeomorphism symmetry is generically broken by the lattice regularization it is expected that it will be recovered in the limit of a large number of simplices.  For this reason there has been interest in studying the fixed points under coarse graining of, at first, simpler models.  These models are either dimensionally reduced or involve simpler gauge groups.

In Chapter \ref{chapter_coarse_graining} we study the behaviour of our spin network generating functional under general coarse graining moves.  We find a simple transformation of the coarse grained action in terms of lattice paths.  In section \ref{sec_ising} we show that for a square lattice, our generating functional expressed as sums over loops similar to (\ref{eqn_3S_loops}), gives precisely the partition function for the 2d Ising model.  Since the Ising model and its renormalization are very well understood this example could provide a toy model for which one could base a study of the more complicated spin foam renormalization. 








\chapter{Gravity from BF Theory}
\label{chapter_BF}


\section{BF Theory}

First let us recall the basic framework of gauge theory \cite{Oeckl:2005rh}.  

\begin{definition}
Let $P$ be a principal $G$-bundle over a smooth $n$ dimensional manifold $\mathcal{M}$.  A connection $\omega$ is defined to be a $\mathfrak{g}$-valued one-form on a local trivialization of $P$.  The curvature $F$ associated with a connection $\omega$ is defined to be
\be
  F \equiv \rd_\omega \omega = \rd \omega + [\omega \wedge \omega]
\ee
where $[\cdot,\cdot]$ is the Lie bracket and $\rd_\omega$ is the exterior covariant derivative.  A connection is said to be flat if $F=0$.
\end{definition}


A change of local trivialization is referred to as a gauge transformation and results in a transformation by $g \in G$ given by
\be \label{eqn_rot_symmetry}
  \omega \mapsto g^{-1} \omega g + g^{-1}\rd g, \qquad F \mapsto g^{-1}F g
\ee

The letter `F' in ``BF theory'' refers to the curvature while `B' refers to an additional $\mathfrak{g}$-valued $(n-2)$ form field which acts as a Lagrange multiplier enforcing flatness of the connection.  Under a gauge transformation $B \mapsto g^{-1} B g$ in a local trivialization.   The action for the theory is given by
\be \label{eqn_BF_action}
  S_{\text{BF}} \equiv \int_\mathcal{M} \: \bra B \wedge F(\omega) \ket
\ee 
where $\bra \cdot \ket$ is the Killing form on the Lie algebra.  A partition function can be formally defined by
\be \label{eqn_Z_BF}
  Z_{\text{BF}} \equiv \int DB \, D\omega \,  e^{i S_{\text{BF}}} = \int D\omega \, \delta(F(\omega))
\ee
having the equations of motion 
\be
  F =0, \qquad \rd_\omega B = 0
\ee
The first equation requires the connection $\omega$ to be flat while the second requires $B$ to be closed.  

Since the action (\ref{eqn_BF_action}) does not involve a metric it is said that BF theory is {\it topological}.  This means that the theory has no local degrees of freedom and hence only has topological degrees of freedom.  This is due to an additional gauge symmetry of the action (\ref{eqn_BF_action})
\be \label{eqn_shift_symmetry}
  \omega \mapsto \omega, \qquad B \mapsto B + \rd \phi
\ee
which follows from the second Bianchi identity $\rd_\omega F = 0$.  The field $\phi$ is any arbitrary $\mathfrak{g}$-valued $n-3$ form.  Now the local exactness of $B \approx \rd \psi$ implies that all solutions of the equation of motion are locally gauge equivalent to the trivial solution $B \approx \rd \psi - \rd \phi \approx 0$.

\section{Classical Actions for Gravity}

General Relativity describes the dynamics of the curvature of a Lorentzian connection on the tangent bundle of a four dimensional manifold.  The difference between General Relativity and a general principal bundle is the existence of a soldering one-form $e$ which relates the fibers of the bundle to the tangent space at each point.  Indeed, the existence of the so called tetrad form $e$ implies the existence of a metric which General Relativity requires.

For a Euclidean $\text{so}(4)$ or Lorentzian $\text{so}(3,1)$ connection the Lie algebra is labeled by a pair of antisymmetric internal indices $I,J = 0,1,2,3$.  Thus the spin connection and the tetrad are explicitly $\omega^{IJ}_{\mu} = - \omega^{JI}_{\mu}$ and $e^{I}_\mu$ where $\mu = 0,1,2,3$ is a spacetime index.  The tetrad defines the metric by
\be
   \eta_{IJ} \, e^{I}_{\mu} e^{J}_{\nu} = g_{\mu \nu} 
\ee
where $\eta$ is the flat spacetime metric.  We will suppress indices as much as possible though since this section is more of an overview than a rigorous derivation which can be found in \cite{Perez:2002vg}.

We will now recall how the existence of $e$ in the four dimensional BF action leads to the classical Palatini action for gravity.  By restricting the two form $B$ field in (\ref{eqn_BF_action}) to be of the form 
\be \label{eqn_simplicity_constraint}
  B = \star(e \wedge e)
\ee
where $e$ is a real one form field we obtain the Palatini action for General Relativity
\be \label{eqn_palatini}
  S_{\text{Palatini}} \equiv \int_\mathcal{M} \text{tr} \left( \star(e \wedge e) \wedge F(\omega) \right) 
\ee
Varying with respect to the tetrad produces the Einstein equations for (torsionful) curvature $F$.  However, the equation of motion obtained by varying with respect to the connection $\omega$ implies the zero torsion condition which thus restricts $\omega$ to be the Levi-Civita connection.  Hence the Palatini action produces the Einstein equations for the Levi-Civita connection at the classical level.  Note that as opposed to the Einstein-Hilbert action, the Palatini action treats the tetrad $e$ and connection $\omega$ as independent variables.

While a path integral quantization of the Palatini action has yet to be well-defined, and not for lack of trying, the BF action does have a well understood quantization.  The spin foam approach to quantum gravity is to impose the constraints (\ref{eqn_simplicity_constraint}) on the well-defined BF path integral {\it after} quantizing.  This approach was pioneered by Riesenberger \cite{Reisenberger:1997sk}, Barret, Crane \cite{Barrett:1997gw}, Baez \cite{Baez:1997zt}, Freidel, Krasnov \cite{Freidel:1998pt} and others.

\section{Discretized BF Theory}
\label{section_disc_BF}

Since BF theory is topological we can describe the continuous theory exactly by a discrete set of data, essentially encoding the topological information of the manifold.  This discrete data can be defined with respect to a triangulation of $\mathcal{M}$, which is just a simplicial complex $\Delta$ homeomorphic to $\mathcal{M}$.\footnote{We can also choose to specify data on much more general cell complexes.  We note that the results developed in this thesis are general enough to apply to these more general discrete structures, however for simplicity we will discuss only simplicial complexes. }

In four dimensions this simplicial complex consists of 4-simplices, tetrahedra, triangles, lines and points.  The dual complex $\Delta^\ast$ is constructed by placing a vertex at the center of each 4-simplex, connecting the vertices by edges, the edges form closed faces, etc.  Each $n$-simplex dual to a vertex in $\Delta^\ast$ approximates a flat neighbourhood of $\mathcal{M}$.  The parallel transport $g_e$ along an edge $e \in \Delta^\ast$ represents a change of local frame.



For each edge $e \in \Delta^\ast$ the parallel transport $g_e$ is related to the connection $\omega$ by the path ordered exponential
\be \label{eqn_parallel_transport}
  g_e \equiv \mathcal{P}\exp \left( \int_{\gamma_e} \omega \right)
\ee
where the path $\gamma_e : [0,1] \rightarrow \mathcal{M}$ parameterizes the edge $e$ in $\mathcal{M}$ which of course depends on the local trivialization.  This parallel transport is of course and element of the gauge group $G$, which we will take to be $\text{spin}(4)$.  If the path $\gamma_e$ is closed the parallel transport is referred to as a holonomy.  Under a gauge transformation
\be
  g_e \mapsto g^{-1}(\gamma_e(0)) \, g_e \, g(\gamma_e(1))
\ee
and so holonomies transform by conjugation.  Similarly, for each face $f$ of $\delta^\ast$ we can define a lie algebra element $X \in \text{so}(4)$ by integrating $B$ over $f$ as in
\be
  B_f \equiv \int_f \rd \sigma_f \, B 
\ee
where $\sigma_f$ is the area form on $f$.

The BF action in the discrete form is defined to be
\be \label{eqn_Z_BF_disc}
  Z^{\Delta^\ast}_{BF} 
  = \int \prod_{e\in \Delta^\ast} \rd g_e \, \prod_{f \in \Delta^\ast} \delta(G_f) 
\ee
in analogy with (\ref{eqn_Z_BF}).  Here $G_f \equiv \vec{\prod}_{e \in f} g_e$ is the holonomy around the face $f$.  The triviality of $G_f$ for each face in $\Delta^\ast$ is the discrete analog of vanishing curvature.

It is important to note that depending on the complex $\Delta^\ast$ there could be redundant delta functions on the RHS of (\ref{eqn_Z_BF_disc}).  For example take a tetrahedron and assume that the holonomy around three of the four faces is the identity.  Then the holonomy of the fourth face is automatically constrained to also be the identity.  This results in a factor $\delta(\one)$ which is divergent and implies that (\ref{eqn_Z_BF_disc}) is not always well defined.  

In analogy with loop divergences in Feynman graphs, these divergences intuitively (but not always) occur when there are closed two dimensional surfaces, or bubbles, in $\Delta^\ast$, and hence they are referred to as bubble divergences.  The degree of divergence can be determined for special complexes by counting these bubbles \cite{Freidel:2009hd}, or by analysing the cellular cohomology of a general cellular complex, see \cite{Bonzom:2010ar,Bonzom:2010zh,Bonzom:2011br}.

These bubble divergences are actually a result of the BF gauge symmetry (\ref{eqn_shift_symmetry}).  A gauge fixing procedure has been developed in the three dimensional case \cite{Freidel:2004vi} and can be extended to higher dimensional BF theory.  In fact it can be shown that the two gauge symmetries: local rotations (\ref{eqn_rot_symmetry}) and shift in $B$ (\ref{eqn_shift_symmetry}) are equivalent to diffeomorphism symmetry in three dimensions \cite{Freidel:2002dw}.

The crux of the spin foam program is that the simplicity constraints break this shift symmetry and thus give rise to non-topological models.  It is an open problem to then determine the residual gauge symmetry and its relation to diffeomorphisms in four dimensions.

\subsection{Vertex Amplitudes}

Let us first expand the BF partition function (\ref{eqn_Z_BF_disc}) into modes.  For a compact group, such as spin(4) this decomposition follows from the Peter-Weyl theorem \cite{faraut2008analysis}.  The Peter-Weyl theorem states that matrix elements of the unitary irreducible representations are dense in $L^2(G)$.  These representations, which we denote by $\rho$, are finite dimensional and countable.  Using the isomorphism $\text{spin}(4) \cong \text{SU}(2) \times \text{SU}(2)$ we will take $\rho = (j_L,j_R)$, but for now we will just write $\rho$.

The trace of a group element in a particular representation is called the character $\chi_\rho(g) \equiv \text{tr}_\rho(g)$.  For central functions $f(g) = f(hgh^{-1})$ for all $h \in G$ the characters $\chi_\rho(g) \equiv \text{tr}_\rho(g)$ form a basis.  The delta function is a central function so we can write
\be
  \delta(g) = \sum_{\rho} \text{dim}(\rho) \chi_{\rho}(g)
\ee
where $\text{dim}(\rho)$ is the dimension of the finite dimensional representation and it follows from the orthonormality of the characters with respect to the Haar measure.

The BF partition function (\ref{eqn_Z_BF_disc}) therefore becomes
\be \label{eqn_Z_BF_modes}
  Z^{\Delta^\ast}_{BF}  
  = \sum_{\rho_f} \int \prod_{e\in \Delta^\ast} \rd g_e \, \prod_{f \in \Delta^\ast} \text{dim}(\rho_f) \chi_{\rho_f}(G_f) 
\ee
where recall $G_f = g_{e_1} g_{e_2} \cdots g_{e_n}$ is the directed product of group elements belonging to the edges $e_1,...,e_n$ of the face $f$.  Thus each face of $\Delta^\ast$ is assigned a unitary irreducible representation $\rho_f$ and we sum over all such representations for each face.

If $\Delta$ is a four dimensional simplicial complex then each (dual) edge $e$ of $\Delta^\ast$ corresponds to a tetrahedron; hence it contains four faces and the four representations $\rho_f$ correspond to their areas.  

The group element $g_e$ on the dual edge $e$ in (\ref{eqn_Z_BF_modes}) acts on the space
\be \label{eqn_tensor_prod_space}
  V^{\rho_{1}} \otimes V^{\rho_{2}} \otimes V^{\rho_{3}} \otimes V^{\rho_{4}}  
\ee
where $\rho_1, \rho_2, \rho_3, \rho_4$ are the representations on the four faces of $e$ and $V^{\rho}$ is the finite dimensional representation space.  Therefore we are free to insert a resolution of identity of (\ref{eqn_tensor_prod_space}) into the partition function (\ref{eqn_Z_BF_modes}) on each edge.  The integration over the Haar measure $\rd g_e$ projects this resolution of identity on (\ref{eqn_tensor_prod_space}) to a resolution of identity on the invariant subspace
\be  \label{eqn_intertwiner_space}
  \cH_{\rho_1,\rho_2,\rho_3,\rho_4} \equiv \text{Inv}_G \left(V^{\rho_{1}} \otimes V^{\rho_{2}} \otimes V^{\rho_{3}} \otimes V^{\rho_{4}} \right)
\ee

Elements of (\ref{eqn_intertwiner_space}) are referred to as intertwiners since they intertwine the representations on (\ref{eqn_tensor_prod_space}) with the trivial representation.  Explicitly we can write
\be
  \int \rd g \, T^{\rho_1}(g) \otimes T^{\rho_2}(g) \otimes T^{\rho_3}(g) \otimes T^{\rho_4}(g) = \sum_{\iota} \|\iota \ket_{\rho_i} {}_{\rho_i}\bra \iota \|
\ee
where $\iota$ labels a basis of the finite dimensional intertwiner space (\ref{eqn_intertwiner_space}).  Note that we are free to choose any basis and just like the representation labels $\rho_f$ the intertwiner basis labels $\iota_e$ will also have a physical interpretation: Different bases will correspond to different physical quantities.  The study of various bases of the intertwiner space (\ref{eqn_intertwiner_space}) will be a main focus of this thesis.

The effect of inserting the intertwiner resolution of identity is that it splits all the dual edges in $\Delta^\ast$ in half, producing an invariant contraction at each vertex which is called a vertex amplitude.  The partition function is then a product of these vertex amplitudes 
\be \label{eqn_Z_BF_vertex_amps}
  Z^{\Delta^\ast}_{BF}  
  = \sum_{\rho_f} \sum_{\iota_e} \, \prod_{f \in \Delta^\ast} \text{dim}(\rho_f) \, \prod_{v \in \Delta^\ast} A_{v}(\rho_f, \iota_e)
\ee
and the vertex amplitude is defined by
\be
  A_{v}(\rho_f, \iota_e) \equiv \underset{e \in v}{\corner} \|\iota_{e} \ket_{\rho_f},
\ee
where we use $\corner$ to denote contraction of the representation arguments in an understood pattern, in this case a 4-simplex. 
We note that if $\Delta$ is not a simplicial complex then other patterns of contraction can be used to define various different vertex amplitudes so long as all the strands are contracted.

\section{Regge Calculus}
\label{section_regge}

Regge Calculus is a discrete approximation of General Relativity by a piece-wise flat simplicial manifold $\Delta$.  If the manifold is curved, then the curvature is concentrated on the co-dimension two simplicies.  For example imagine triangles meeting at a vertex in two dimensions.  If the triangles approximate a curved 2d surface then when they are laid flat the they will possess a deficit angle. See Figure \ref{fig_regge}.

\begin{figure} 
  \centering
    \includegraphics[width=0.25\textwidth]{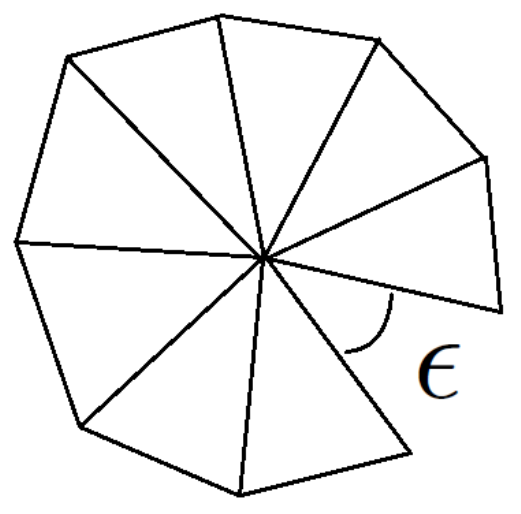}
   \caption{In 2d the curvature of a triangulation is concentrated at the vertices.  When the triangles surrounding a vertex are laid flat they produce a deficit angle $\epsilon$ which is a measure of the curvature at that point. }  \label{fig_regge}
\end{figure}

To be more precise the sum of the dihedral angles of co-dimension one simplices belonging to a co-dimension two simplex fails to be $2\pi$.  Let us specialize to the 4d Riemannian case in which the codimension 2, 1 and 0 simplices are triangles $t$, tetrahedra $\tau$ and 4-simplices $\sigma$.  The deficit angle at a triangle $t$ is then
\be \label{eqn_epsilon_dihedral}
  \epsilon_{t} = 2\pi - \sum_{\sigma \supset t} \Theta^{\sigma}_{t} 
\ee
where $\Theta^{\sigma}_{t}$ is the 4d dihedral angle between the pair of tetrahedra in $\sigma$ sharing the triangle $t$.  The Regge action for a four dimensional Riemannian manifold is defined by
\be
  S_{\text{Regge}} = \sum_{t} A_t \epsilon_t = 2\pi \sum_t A_t - \sum_\sigma \sum_{t \in \sigma} A_t \Theta^{\sigma}_{t} 
\ee 
where $A_t$ is the area of triangle $t$.  The continuum limit of this action coincides with the Einstein-Hilbert action \cite{Regge:1961px}.  


The Regge action is a second order action in a similar way that the Einstein-Hilbert action is second order, i.e. contains two derivatives of the frame field.  This is a result of imposing metric compatibility of the connection into the action.  The Palatini action (\ref{eqn_palatini}) on the other hand, is a first order action in that it has only first derivatives and treats the connection as an independent variable.  

The geometry of the triangulation is completely determined by the length variables $l_e$.  Therefore the areas $A_t$ and the dihedral angles $\Theta_t$ can both be considered as functions of $l_e$ and so the Regge action is also only a function of $l_e$.  Taking the areas and dihedral angles as independent variables gives a first order formulation of Regge Calculus as derived by Barrett \cite{Barrett:1994nn}.  This, however, requires imposing a constraint for each 4-simplex enforcing the vanishing of the Gram matrix \cite{Bahr:2009qd}.

Alternatively, and more relevant to our asymptotic results in Chapter \ref{chapter_semi}, one can use the 3d dihedral angles and the areas as configuration variables \cite{Dittrich:2008va}.  For a 4-simplex there are thirty 3d dihedral angles and 10 areas so we require thirty independent constraints.  These constraints can be taken to be the closure of the five tetrahedra, and the matching of the 2d interior angles of all of the triangles.

Label the tetrahedra of the 4-simplex by $\{i,j,k,l,m\} \in \{1,...,5\}$ so $(ij)$ is the triangle shared by $i$ and $j$ and $(ji)(ik)$ is the edge shared by the two triangles.  Let $\theta^{a}_{ij}$ be the 3d dihedral angle in tetrahedron $a$ between the triangles sharing $i$ and $j$ with the convention $\theta^{a}_{ii} = 0$.  Let $\alpha^{ai}_{jk}$ be the 2d interior angle, in tetrahedron $a$, in triangle $(ai)$, between edges $(ia)(aj)$ and $(ia)(ak)$.  Then $\alpha^{ai}_{jk}$ can be expressed in terms of $\theta^{a}_{ij}$ \cite{Dittrich:2008va}:
\be \label{eqn_2d_3d_dihedral}
  \cos \alpha^{ai}_{jk} = \frac{\cos \theta^{a}_{jk} - \cos \theta^{a}_{ij} \cos \theta^{a}_{ik}}{ \sin \theta^{a}_{ij} \sin \theta^{a}_{ik}}.
\ee

Defining the constraints:
\be
  \mathcal{C}^{ab}_{jk} \equiv \cos \alpha^{ab}_{jk} - \cos \alpha^{ba}_{kj} \qquad \mathcal{N}^{a}_{ij} \equiv A_{ai} - \sum_{(j \neq a,i)} A_{aj} \cos \theta^{a}_{ij}
\ee
The first ensures that the angles of glued triangles are of the same shape.  The second is simply the scalar product of the closure constraint.  There are 30 of the first and 20 of the second but together there are only 30 independent constraints \cite{Dittrich:2008va}.

This leads to a formulation of Regge calculus in terms of the areas and 3d dihedral angles
\be \label{eqn_Regge_3d_angles}
  S[A_{t},\theta^{\tau}_{t},\lambda^{\tau}_{t},\mu^{\sigma}_{ee'}] = \sum_{t} A_t \epsilon_t(\theta) 
  + \sum_{\tau} \sum_{t \in \tau} \lambda^{\tau}_{t} \mathcal{N}^{\tau}_{t} + \sum_{\sigma} \sum_{e,e' \in \sigma} \mu^{\sigma}_{ee'} \mathcal{C}^{\sigma}_{ee'}
\ee
where the deficit angles $\epsilon_t(\theta)$ are given in terms of the 4d dihedral angles (\ref{eqn_epsilon_dihedral}), which can be expressed in terms of the 3d dihedral angles by a formula similar to (\ref{eqn_2d_3d_dihedral}).

In \cite{Dittrich:2008va} it is postulated that the first two terms in the action correspond to the BF action and that the third term is the analog of Plebanski's constraints.  Hence the shape matching constraints $\mathcal{C}^{\sigma}_{ee'}$ with the closure constraint might be the proper discrete analog of the simplicity constraints.

This is relevant because in Chapter \ref{chapter_semi} we derive a semi-classical action for a 4-simplex BF amplitude in Theorem \ref{thm_twisted_action}.  We find an action possessing the first two terms in (\ref{eqn_Regge_3d_angles}) where the 4d dihedral angles are functions of the 3d dihedral angles and the constraints $\mathcal{C}^{\sigma}_{ee'}$ are not satisfied.  Since the shapes of the triangles do not match, this is action of a recently proposed generalization of BF theory called Twisted Geometry \cite{Freidel:2010aq,Freidel:2013bfa}.

In section (\ref{section_geo}) we discuss  the characterization of these geometricity constraints and we give a condition on the boundary data of the amplitude which ensures the vanishing of the shape matching constraints $\mathcal{C}^{\sigma}_{ee'}$ in the semi-classical limit.   It would be straightforward to implement these constraints into the BF partition function and this could be an interesting spin foam model to consider in future studies.

\chapter{BF Amplitudes}
\label{chapter_BF_amp}
\section{The Group SU(2)}
\label{sec_su2}

The group $\text{SU}(2)$ consists of unitary, $2\times2$ matrices with unit determinant and can be parameterized as follows  
\be \label{eqn_g_fund}
  g = \bpm \alpha & -\overline{\beta} \\ \beta & \alpha \epm, \qquad |\alpha|^2+|\beta|^2 = 1
\ee
The group operation is given by matrix multiplication for which there are a countable number of unitary irreducible representations labeled by half integers $j = 1/2, 1, 3/2,...$ referred to as spins.  The fundamental representation on $|z\ket \in \C^2$ is given by left multiplication $g|z\ket$ and is called the spinor representation.  We will use a bra-ket notation to denote a spinor and its contravariant conjugate by
\be
  |z\ket = \bpm z_0 \\ z_1 \epm, \qquad |z] = \bpm -\bar{z}_{1} \\ \bar{z}_0 \epm
\ee
We will sometimes refer to a spinor by simply $z$ and we will use $\check{z}$ to refer to $|z]$.
We should mention that this notation differs from other conventions for spinors which also use square and angle brackets.\footnote{We note that in the index notation $|z\ket = z^A$ where $A=0,1$ and $\bra z| = \delta_{A\bar{A}} \bar{z}^{\bar{A}}$, $|z] = \delta^{AB} \epsilon_{B\bar{A}} \bar{z}^{\bar{A}}$, $[z| = \epsilon_{AB} z^{A}$.  Thus the two invariants are $\bra z | w \ket = \delta_{\bar{A}B} \bar{z}^{\bar{A}} w^B$ and $[z|w\ket = \epsilon_{AB}z^{A}w^{B}$.} 

Let $V^j$ be the vector space of holomorphic polynomials on $\C^2$ which are homogeneous of degree $2j$.  
Then following action of $\text{SU}(2)$ on $P \in V^j$ 
\be \label{eqn_SU2_rep}
  T^{j}(g) P(z) = P(g^{-1}(z))
\ee
defines the $2j+1$ dimensional representation of spin $j$.  This representation is unitary with respect to the Hermitian inner product 
\be \label{barg_in_prod}
  \bra P | P' \ket = \int_{\C^2} \overline{P(z)} P'(z) \rd\mu(z), \qquad \rd\mu(z) = \pi^{-2} e^{-\bra z | z \ket}  \rd^{4}z
\ee 
for two functions $P,P' \in V^j$ and $\rd^{4}z$ is the Lebesgue measure on $\C^2$.  This is the well known Bargmann-Fock inner product \cite{bargmann1962representations,schwinger2001angular} which was introduced in the Loop Quantum Gravity context in \cite{Borja:2010rc, Livine:2011gp}.  Note that we will use a round bracket $(z|f\ket$ to denote the holomorphic representation of a state in this Hilbert space.

The standard orthonormal basis on $V^j$ with respect to this inner product is given by
\be \label{eqn_ortho_holo}
   e^{j}_{m}(z) \equiv (z|j m \ket = \frac{z^{j+m}_0 z^{j-m}_1}{\sqrt{(j+m)!(j-m)!}}, \qquad 
\ee
which is simply the holomorphic representaation.  Indeed, we can define the differential operators
\be  
  J_+ = z_0 \frac{\partial}{\partial z_1}, \quad J_- = z_1 \frac{\partial}{\partial z_0}, \quad J_3 = \frac{1}{2}\left( z_0 \frac{\partial}{\partial z_0} - z_1 \frac{\partial}{\partial z_1} \right)
\ee
and see that they satisfy the commutation relations
\be
  [J_3,J_\pm] = \pm J_\pm, \qquad [J_+,J_-] = 2J_3
\ee
We also have the linear Casimir operator 
\be  \label{eqn_E_casimir}
  E = \frac{1}{2}\left( z_0 \frac{\partial}{\partial z_0} + z_1 \frac{\partial}{\partial z_1} \right)
\ee
which commutes with the other operators $[E,J_3] = [E,J_\pm] = 0$ and is related to the quadratic Casimir by
\be
  J^2 \equiv \frac{1}{2}(J_+J_- + J_-J_+) + J_{3}^2 = E(E+1)
\ee
and the usual eigenvalue equations
\be
  J_3 |j m \ket = m |j m \ket, \qquad E |j m \ket = j |j m \ket
\ee
which can be found by acting on the basis (\ref{eqn_ortho_holo}). This representation is exactly the Schwinger representation \cite{schwinger2001angular} in terms of a pair of decoupled harmonic oscillators
\be \label{eqn_schwinger}
  [a,a^\dagger] = [b,b^\dagger] = 1, \qquad [a,b] = [a,b^\dagger] = [a^\dagger,b] = [a^\dagger,b^\dagger] = 0
\ee
by the quantization
\be  
  a^\dagger \mapsto z_0 \quad b^\dagger \mapsto z_1 \qquad a \mapsto \frac{\partial}{\partial z_0} \quad b \mapsto \frac{\partial}{\partial z_1}
\ee
This representation in terms of harmonic oscillators not only has a close connection with coherent states, but also illuminates a U(N) representation on the space of SU(2) intertwiners as first pointed out by Girelli and Livine \cite{Girelli:2005ii} and led to the so called U(N) formalism for coherent intertwiners \cite{Freidel:2009ck,Freidel:2010tt}.

\section{Orthonormal Intertwiners}

An intertwiner is a map which is invariant under the action of a group.  The classic example of an intertwining map is given by the Clebsch-Gordan coefficients of SU(2) which map the tensor product of two representations of to a direct sum of irreducible representations.  In this section we review the construction of the Clebsch-Gordan map in the spinor representation.  We emphasize the role of the existence of a holomorphic spinor invariant as the key to this decomposition.

We define the Wigner 3j symbol from the CG coefficients and with it we construct the canonical orthonormal basis intertwiners.  In the 4-valent case this leads to three possible bases $S$, $T$, and $U$ which are an allusion to the three Mandelstam channels.  Finally, we contract five 4-valent states into a 4-simplex amplitude called the 15j symbol which is the building block of the Ooguri model for spin(4) BF theory in 4d.

\subsection{The Clebsch-Gordan Intertwiner}
\label{section_CG_int}

Consider the tensor product of two representations $T^{j_1} \otimes T^{j_2}$ with the diagonal action on holomorphic polynomials
\be \label{eqn_diag_action}
  (T^{j_1}\otimes T^{j_2})(g) P(z_1,z_2) = P(g^{-1}z_1,g^{-1}z_2)
\ee
The canonical basis of $V^{j_1} \otimes V^{j_2}$ is given by $e^{j_1}_{m_1}(z_1) e^{j_2}_{m_2}(z_2)$.  However we can construct another basis due to the existence of the holomorphic invariant
\be \label{eqn_holo_inv_z12}
  [z_1|z_2 \ket = \alpha_1\beta_2 - \alpha_2\beta_1
\ee
Indeed, consider the set of holomorphic polynomials which are divisible by $[z_1|z_2 \ket$
\be \label{eqn_inv_CG_subspace}
  [z_1|z_2\ket^{j_1+j_2-j} Q^{j}(z_1,z_2)
\ee
where $Q^{j}$ is a polynomial homogeneous of degree $j_1-j_2+j$ in $z_1$ and $-j_1+j_2+j$ in $z_2$.  The subspaces spanned by the polynomials (\ref{eqn_inv_CG_subspace}) are invariant under the action (\ref{eqn_diag_action}) for each $j$ since (\ref{eqn_holo_inv_z12}) is invariant.  Moreover it is easy to see that polynomials (\ref{eqn_inv_CG_subspace}) of different $j$ are orthogonal which leads to the decomposition
\be \label{eqn_chlebsch}
  V^{j_1} \otimes V^{j_2} \cong \mathop{\oplus}_{j=|j_1-j_2|}^{j_1+j_2} V^j 
\ee 
which is the well known Clebsch-Gordan series.  It is also clear that each representation on the RHS of (\ref{eqn_chlebsch}) appears only once.  The factoring (\ref{eqn_inv_CG_subspace}) of invariants will be a key idea for the rest of this thesis.

Note that $j = j(j_1,j_2)$ since there are restrictions
\be \label{eqn_CG_conditions}
  |j_1-j_2| \leq j \leq j_1+j_2, \qquad j_1+j_2+j \in \N
\ee
which are known as the Clebsch-Gordan (or triangle) conditions.  They are equivalent to the existence of three positive integers $k_{12},k_{13},k_{23}$ such that
\be \label{eqn_CG_ks}
  k_{12} = j_1+j_2-j, \quad k_{13} = j_1+j_2-j \quad k_{23} = -j_1+j_2+j
\ee
We will later generalize (\ref{eqn_CG_ks}) to higher dimensions in the coming chapters.  Furthermore we will extend the key insight (\ref{eqn_inv_CG_subspace}) to characterize higher dimensional intertwiners in terms of the divisibility by the fundamental invariants.

It is now a straightforward, but tedious, task \cite{vilenkin1991representation} to construct the canonical orthonormal basis $e^{j(j_1,j_2)}_{m(m_1,m_2)}(z_1,z_2)$ on the RHS of (\ref{eqn_chlebsch}) from the highest weights of (\ref{eqn_inv_CG_subspace}).  It also follows by construction that
\be
  T^{j}(g) e^{j(j_1,j_2)}_{m(m_1,m_2)}(z_1,z_2) = e^{j(j_1,j_2)}_{m(m_1,m_2)}(g^{-1}z_1,g^{-1}z_2)
\ee

The Clebsch-Gordan coefficients are then defined to be the matrix elements of this change of basis
\be  \label{eqn_clebsch_basis_map}
  e^{j_1}_{m_1}(z_1) \, e^{j_2}_{m_2}(z_2) = \sum_{j=|j_1-j_2|}^{j_1+j_2}\sum_{m=-j}^{j} C^{j_1, j_2, j}_{m_1, m_2, m} \, e^{j(j_1,j_2)}_{m(m_1,m_2)}(z_1,z_2)
\ee  
and we can explicitly compute the Clebsch-Gordan coefficients by the scalar product (\ref{barg_in_prod}) as
\be \label{eqn_clebsch_inner_def}
  C^{j_1, j_2, j}_{m_1, m_2, m} = \bra e^{j(j_1,j_2)}_{m(m_1,m_2)} | e^{j_1}_{m_1} \, e^{j_2}_{m_2} \ket
\ee
Finally we note that the Clebsch-Gordan map
\be \label{eqn_CG_map}
  C^{j_1,j_2,j} : V^{j_1} \otimes V^{j_2} \rightarrow V^{j}
\ee
defined by the action (\ref{eqn_clebsch_basis_map}), is a linear isomorphism between two orthonormal bases of the same space and is hence unitary.  It also follows that 
\be
  C^{j_1,j_2,j} (T^{j_1}\otimes T^{j_2}) = T^{j} C^{j_1,j_2,j} 
\ee 
which shows that $C^{j_1,j_2,j}$ intertwines the two representations $T^{j_1} \otimes T^{j_2}$ and $T^{j}$.

This was an admittedly messy example of an intertwining map, but it introduced two key ideas (\ref{eqn_inv_CG_subspace}) and (\ref{eqn_CG_ks}) in a hopefully familiar context.  We will next look at the Wigner 3j symbol which is more symmetric and then generalize these intertwining maps from three to $n$ tensor products.

\subsubsection{The Wigner 3j Symbol}

In the previous section we constructed the Clebsch-Gordan coefficients and found that they defined an intertwining map.  Note that instead of the map (\ref{eqn_CG_map}) into $V^j$ we could have instead considered the invariant linear functionals
\be \label{eqn_3_valent_inter}
  V^{j_1} \otimes V^{j_2} \otimes (V^{j})^\ast \rightarrow \C 
\ee
where $(V^{j})^\ast$ is the canonical dual of $V^{j}$.  Moreover, since a space and its dual are isomorphic we lose no generality in considering the more symmetric set of maps
\be \label{eqn_3_inter_space}
  \cH_{j_1,j_2,j_3} \equiv \text{Inv}_{\text{SU}(2)}\left[V^{j_1} \otimes V^{j_2} \otimes V^{j_3} \right].
\ee
where $\text{Inv}_{\text{SU}(2)}$ indicates the SU(2) invariant subspace.  This is a Hilbert space with the inner product inherited from (\ref{barg_in_prod}) and each invariant vector intertwines the tensor product with the trivial representation.  

The Clebsch-Gordan map we just constructed intertwines three representations and is unique.  Indeed, since (\ref{eqn_CG_map}) is a unitary isomorphism and each representation in (\ref{eqn_chlebsch}) appears only once it follows that the Clebsch-Gordan map is the only invariant tensor in (\ref{eqn_3_valent_inter}) up to scaling.  Moreover, by unitarity of the map (\ref{eqn_CG_map}) the coefficients are also orthonormal
\be \label{eqn_CG_orthog}
  \sum_{m_1,m_2} C^{j_1, j_2, j}_{m_1, m_2, m} C^{j_1, j_2, j'}_{m_1, m_2, m'} = \delta_{j,j'} \delta_{m,m'}
\ee

Finally we can define a more symmetric version of the Clebsch-Gordan coefficient known as the Wigner 3j symbol which is defined simply by a rescaling and change of phase \cite{yutsis1962mathematical}
\be \label{eqn_3j_CG}
  (-1)^{j-m}\threej{j_1}{j_2}{j}{m_1}{m_2}{-m} \equiv \frac{(-1)^{j_1-j_2+j}}{\sqrt{2j+1}} C^{j_1, j_2, j}_{m_1, m_2, m} 
\ee

In summary, a 3-valent intertwiner is uniquely determined by its three spins which must satisfy the triangle conditions (\ref{eqn_CG_conditions}).  For more than three spins the space of invariant vectors has dimension greater than one and thus requires the specification of a basis.

\subsection{Edge Orientation, Vertex Ordering, and Contraction}
\label{sec_contraction}

The graphical representation of the $3j$ symbol is given by a trivalent node.  The three legs are labeled by the three spins, the ordering of which, affects the phase.  Indeed, while the $3j$ symbol is unique it possesses many symmetries \cite{varshalovich1988quantum}, the most important of which can be summarized by
\begin{itemize}
  \item even permutations of the columns are symmetric 
	\item odd permutations of the columns change the phase by the sign $(-1)^{j_1+j_2+j_3}$ 
	\item change of sign $m_i \mapsto -m_i$ $\forall i=1,2,3$ change the phase by the sign $(-1)^{j_1+j_2+j_3}$
\end{itemize}
The invariance under cyclic permutations of the spins $j_1, j_2, j_3$ implies that there are two possible orderings of the three spins in (\ref{eqn_3j_CG}) which are referred to as clockwise or counterclockwise with reference to the planar drawing of the trivalent node.  Therefore every amplitude defined by the contraction of trivalent intertwiners requires a label (usually plus/minus for ccw/cw) to specify the ordering which affects the overall phase.

Furthermore, we must assign an orientation of the edges to distinguish between $V^j$ and its dual $(V^j)^\ast$ in (\ref{eqn_3_inter_space}).  When contracting indicies we will always take one index in $V^j$ and the other in the dual $(V^j)^\ast$ so that there is a definite direction along the edge.  An edge outgoing from a node will indicate an index in $V^j$ while an incoming edge belongs to the dual $(V^j)^\ast$.  This will also affect the overall phase of an amplitude.      

Let us see how this affects the $3j$ symbol.  Observe that in the fundamental representation (\ref{eqn_g_fund}) the complex conjugate of a group element can be obtained  by
\be
 g \epsilon  = \epsilon \bar{g} \qquad \epsilon \equiv \bpm 0 & 1 \\  -1 & 0 \epm 
\ee
Further, since $\epsilon \in \text{SU}(2)$ it follows that
\be
  T^j(g) T^j(\epsilon) = T^j(\epsilon) (T^j(g))^\ast
\ee
where $(T^j(g))^\ast$ is the contragredient representation acting in the dual space $(V^{j})^\ast$.  Hence $T^j(\epsilon) : V^j \rightarrow (V^j)^\ast$ and a basis of $(V^{j})^\ast$ is given by
\be \label{eqn_dual_basis}
  e^{j}_{m}(\epsilon^{-1}z) = e^{j}_{m}\bpm -z_1 \\ z_0 \epm = (-1)^{j-m} e^{j}_{-m}(z)
\ee
Hence contraction of $V^j$ and $(V^j)^\ast$ can be achieved by multiplying by the sign $(-1)^{j-m}$, setting $m \mapsto -m$ on the dual leg and summing over $m$.  Observe that this is equivalent to representing one leg with $z$ and the dual leg with $\check{z}$ (recall $|\check{z}\ket \equiv |z]$) and integrating over $\rd \mu(z)$. 

In this way, if we denote the 3j symbol by $\|j_1, j_2, j_3\ket$ as in
\be \label{eqn_3j_state}
  (z_1,z_2,z_3\|j_1, j_2, j_3\ket \equiv \sum_{m_i} \threej{j_1}{j_2}{j_3}{m_1}{m_2}{m_3} \prod_i e^{j_i}_{m_i} (z_i)
\ee
the vertex ordering is explicit and the legs are either incoming or outgoing if $j_i$ is represented with $z_i$ or $\check{z}_i$.  The full contraction of two 3j symbols is then the scalar product
\begin{align}
  \bra j_1, j_2, j_3 \| j_3, j_2, j_1 \ket 
	&= \int \prod_{i=1}^{3} \rd \mu(z_i) \, (\check{z}_1,\check{z}_2,\check{z}_3\|j_1, j_2, j_3\ket (z_3,z_2,z_1\|j_3, j_2, j_1\ket \\
	&= \sum_{m_i} (-1)^{\sum_i j_i-m_i} \threej{j_1}{j_2}{j_3}{-m_1}{-m_2}{-m_3} \threej{j_3}{j_2}{j_1}{m_3}{m_2}{m_1} \nonumber \\
	&= \frac{1}{2j_3+1} \sum_{m_i}  C^{j_1, j_2, j_3}_{m_1, m_2, m_3} C^{j_1, j_2, j_3}_{m_1, m_2, m_3} = 1
\end{align}
where we used (\ref{eqn_CG_orthog}), $\sum_{m_3} = 2j_3+1$ and the basic symmetry properties of the 3j symbol. 

Although these edge orientations and vertex orderings are necessary to properly define the phases of amplitudes, these choices are a priori arbitrary and hence we will often omit them from diagrams.  The details of the contraction and orientations will be assumed implicit in the definition of the amplitude.  The determination of these signs will unfortunately be one of the most laborious aspects of later chapters, and will play a key role in connecting with the 2d Ising model in Section \ref{sec_ising}.

There are other conventions for spin network amplitudes such as Kauffman and Lins \cite{kauffman2013knots} which define away these signs.  On the contrary, we find that the inclusion of these signs provides a generality that allows for simpler, more powerful, expressions for spin network generating functionals than in these other conventions (see Chapter \ref{chapter_exact}).

Alternatively one could define a phase convention based on semiclassical considerations.  In \cite{Barrett:2009gg} a specific phase was postulated to give precisely the Regge action in the asymptotics of a coherent intertwiner contraction and was part of the definition of what they called a ``Regge state''.  The counterpart of this phase for a basis of intertwiners which we introduced in Chapter \ref{section_discrete_coherent} is computed explicitly in Theorem \ref{thm_twisted_action}.



\subsection{The Orthogonal Basis}

Let us now extend the intertwiner Hilbert space (\ref{eqn_3_inter_space}) to the $n$-valent case
\be \label{eqn_inter_space}
  \cH_{j_1,...,j_n} \equiv \text{Inv}_{\text{SU}(2)}\left[V^{j_1} \otimes \cdots \otimes V^{j_n} \right].
\ee

The vectors in this Hilbert space will be referred to as $n$-valent intertwiners since they are represented graphically by an $n$-valent node.  The legs are labeled by spins while the node has a label representing a state in (\ref{eqn_inter_space}).

When $n=2$ Shur's lemma requires $j_1 = j_2$ by irreducibility.  The first non-trivial case is $n=3$ for which we saw there was a unique intertwiner given by the Clebsch-Gordan or Wigner 3j symbol.  For $n=4$ the easiest way to construct a basis is to contract two 3-valent intertwiners together to create a 4-valent intertwiner.  

Besides the four spins on its external legs, this 4-valent intertwiner contains an extra spin on the adjoining link.  There are three ways to perform such a contraction corresponding to the three Mandelstam channels $S$, $T$, and $U$ as depicted in Figure \ref{fig_STU}.
\begin{figure} 
  \centering
    \includegraphics[width=0.8\textwidth]{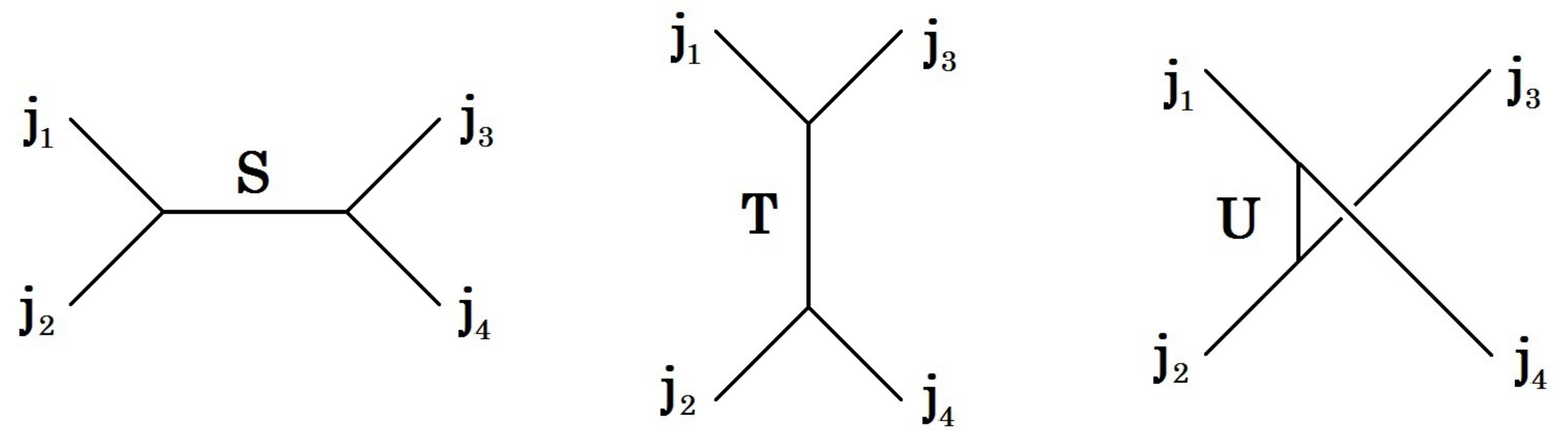}
    \caption{The three channels of a 4-valent vertex.}  \label{fig_STU}
\end{figure}

Contracting two 3j symbols with the procedure described below (\ref{eqn_dual_basis}) we can define the $S$ channels up to a constant factor
\begin{align} \label{eqn_S_3j}
  (z_i|S\ket \equiv N_S \sum_{m_i} \sum_{m=-S}^{S} (-1)^{S-m} \threej{j_1}{j_2}{S}{m_1}{m_2}{m} \threej{S}{j_3}{j_4}{-m}{m_3}{m_4} \prod_{i=1}^{4} e^{j_i}_{m_i}(z_i)
\end{align}
where the spin $S$ has the range
\be
\max\{|j_1-j_2|,|j_3-j_4|\} \leq S \leq \min\{j_1+j_2,j_3+j_4\}
\ee
Since the dual is on the vertex $(S,j_3,j_4)$ the orientation of the internal edge is from $(j_1,j_2,S)$ to $(S,j_3,j_4)$.  Similarly, the $T$ channel is defined by permuting $2 \rightarrow 3$ in (\ref{eqn_S_3j}) giving
\begin{align} \label{eqn_T_3j}
  (z_i|T\ket \equiv N_T \sum_{m_i} \sum_{m=-T}^{T} (-1)^{T-m} \threej{j_1}{j_3}{T}{m_1}{m_3}{m} \threej{T}{j_2}{j_4}{-m}{m_2}{m_4} \prod_{i=1}^{4} e^{j_i}_{m_i}(z_i)
\end{align}
and the $U$ channel by permuting $2 \rightarrow 4$ in (\ref{eqn_S_3j}).  

One can check that these states are eigenstates of the scalar product operators $\vec{J}_i \cdot \vec{J}_{j}$.  In fact this can be taken as the definition e definition 
\begin{align} \label{eqn_S_eigen}
  \vec{J}_1 \cdot \vec{J}_2 \, |S\ket &= \frac{1}{2}\left(S(S+1) - j_1(j_1+1) - j_2(j_2+1)\right) |S \ket, 
\end{align}
and similarly for $T$ and $U$.  The orthogonality of these states
\be
  \bra S' | S \ket = \frac{\delta_{S,S'}}{2S+1} N_{S}^2 , \quad \bra T' | T \ket =  \frac{\delta_{T,T'}}{2T+1} N_{T}^2 , \quad \bra U' | U \ket =  \frac{\delta_{U,U'}}{2U+1} N_{U}^2 
\ee 
is also easy to verify using (\ref{eqn_CG_orthog}).  Note that orthogonality in the external spins $j_i$ is implied.  

While we could choose the constants $N_{S,T,U}$ to normalize the scalar products we will instead choose
\be
  N_S = \Delta(j_1 j_2 S) \Delta(j_3 j_4 S), \quad N_T = \Delta(j_1 j_3 T) \Delta(j_2 j_4 T), \quad N_U = \Delta(j_1 j_4 U) \Delta(j_2 j_3 U)
\ee
where we define the triangle coefficients
\be \label{eqn_tri_coeff}
\Delta^{2}(j_{1}j_{2}j_{3}) \equiv \frac{(j_{1}+j_{2}+j_{3}+1)!}{(j_{1}-j_{2}+j_{3})! (j_{2}-j_{1}+j_{3})!(j_{1}+j_{2}-j_{3})!}.
\ee 
This normalization will allow us to relate both the $S$ and $T$ states in a natural way (See Theorem \ref{thm_sum_T}) to a new basis defined in section \ref{section_discrete_coherent} we call the discrete coherent basis.  

It can be shown that each set of states (\ref{eqn_S_3j}) spans the intertwiner space (\ref{eqn_inter_space}).  Thus we can express the resolution of identity as
\be \label{eqn_res_id_orth}
  \one_{j_1,...,j_4} = \sum_{S} \frac{|S\ket\bra S|}{\|S\|^2} = \sum_{T} \frac{|T\ket\bra T|}{\|T\|^2} = \sum_{U} \frac{|U\ket\bra U|}{\|U\|^2} 
\ee
where we define the normalization constants $\|S\|^2 = N_{S}^2/(2S+1)$.

This procedure can be extended to construct orthonormal bases for the higher valent intertwiner spaces.  The $n$-valent space will have possess (many) orthonormal bases labeled by $n-3$ extra spins constructed by contracting $n-2$ trivalent intertwiners together.  We will continue to focus on the 4-valent case though.

\subsection{The 15j Symbol and the Ooguri Model}

Let us now contract five 4-valent intertwiners from the bases (\ref{eqn_S_3j}) or (\ref{eqn_T_3j}) in the pattern of a 4-simplex.  This amplitude is called the 15j symbol since it depends on ten spins connecting the intertwiners and five spins on the vertices labeling a state from either the $S$, $T$, or $U$ bases.  

\begin{figure} 
  \centering
    \includegraphics[width=0.8\textwidth]{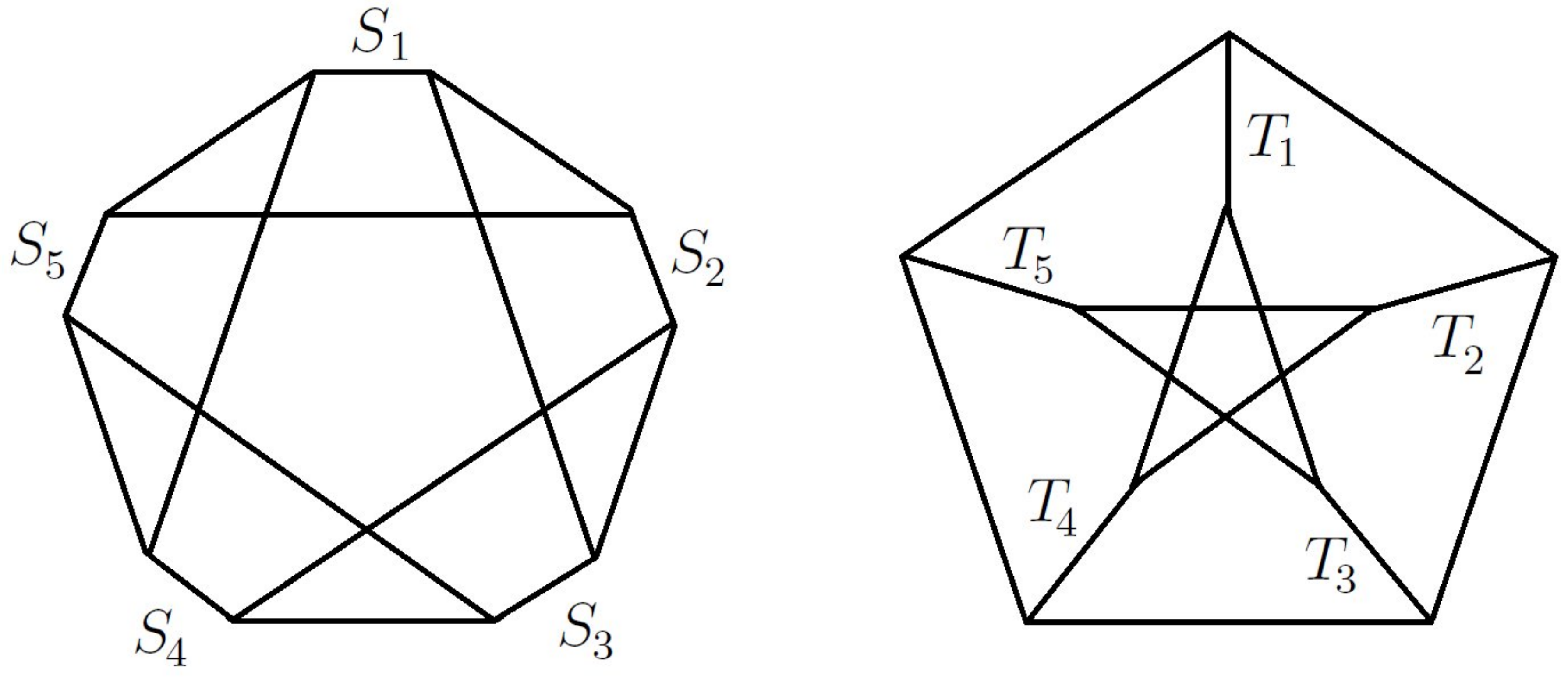}
    \caption{Two of the five kinds of 15j symbols constructed by contracting all $S$ or all $T$ channels.}  \label{fig_15j}
\end{figure}

Let us label the five nodes of the 4-simplex by 1,...,5 so the spins connecting the intertwiners are labeled $j_{ij}$ where $i,j \in \{1,...,5\}$.  Then contracting five $|S_i\ket$ channel states $i=1,...,5$ defines the 15j symbol of the first kind
\be \label{eqn_15j_cont}
  \{15j\}_{S_i} \equiv \underset{i}{\corner} \|S_i \ket
\ee
There are four other inequivalent, irreducible\footnote{By irreducible we mean it does not contain any cycles of length three.  This is because such an amplitude can be reduced to a product of a 12j symbol and a 6j symbol by inserting a trivial resolution of identity on 3-valent intertwiners factoring out the 3-cycle.  Recall the resolution of identity is trivial because the 3-valent intertwiner space is one-dimensional. } kinds of 15j symbols which can be constructed by contracting $S$, $T$, and $U$ channels \cite{yutsis1962mathematical}.

In the literature a normalized 15j symbol $\widehat{\{15j\}}_{S_i}$ is defined with respect to (\ref{eqn_15j_cont}) by simply dividing by the norm $\prod_{i} \|S_i\|$.  We note that in  the notation of \cite{yutsis1962mathematical} 
\begin{align} \label{eqn_yutsis_15j}
  &\widehat{\{15j\}}_{S_i} \equiv  \begin{Bmatrix} S_1 && j_{13} && S_3 && j_{35} && S_5 & \\
	                &  j_{12} && j_{23} && j_{34} && j_{45} && j_{15} \\
									j_{25} && S_2 && j_{24} && S_4 && j_{14} & \end{Bmatrix}  = \frac{\{15j\}_{S_i}}{\prod_i \|S_i\|}
\end{align}
Note that this definition (\ref{eqn_yutsis_15j}) of the 15j symbol comes with a specific overall phase determined by a conventional orientation of the edges and vertices.  Hence (\ref{eqn_yutsis_15j}) is up to a sign.

We will now express the BF partition function (\ref{eqn_Z_BF_vertex_amps}) for $G = \text{Spin}(4)$ in terms of 15j symbols.  This form of BF theory is known as the Ooguri model, but it was also studied by Crane, Yetter, 

Using the isomorphism $\text{Spin}(4) \cong \text{SU}(2) \times \text{SU}(2)$ a basis of $\text{Spin}(4)$ intertwiners is given by $\|S^{L} \ket \otimes \|S^{R} \ket$ where $L$ and $R$ denote the left and right copies of $\text{SU}(2)$.  Therefore if we insert two copies of the resolution of identity (\ref{eqn_res_id_orth}) into the BF partition function (\ref{eqn_Z_BF_vertex_amps}) we get
\be \label{eqn_ooguri_BF}
  Z^{\Delta^\ast}_{BF}  
  = \sum_{(j^{L}_{f},j^{R}_{f})} \sum_{(S^{L}_e,S^{R}_e)} \,(-1)^\chi \prod_{f \in \Delta^\ast} (2j^{L}_{f}+1)(2j^{R}_{f}+1) \, \prod_{e \in \Delta^\ast} \|S^{L}_e\|^{-2} \|S^{R}_e\|^{-2} \, \prod_{v \in \Delta^\ast} \{15j\}_{S^{L}_e} \{15j\}_{S^{R}_e}
\ee
where $\chi$ is a linear function of the $k^{L}_{ij}, k^{R}_{ij}$ relating the two sign conventions.  This model was first proposed by Ooguri \cite{Ooguri:1992eb} in a manner similar to the Bulatov model (\ref{eqn_bulatov}) but in four dimensions.

The use of the orthogonal basis in the Ooguri model means that it possesses the fewest number of parameters.  This is not necessarily an advantage though, since the geometry of each of the tetrahedra in the 4-simplex are underdetermined.  Indeed, for each node of the 15j symbol the four incident spins determine the four areas of the faces of the tetrahedron, while the intertwiner label $S$ represents the area of one of the the medial parallelograms \cite{Baez:1999tk}.   

\begin{figure} 
  \centering
    \includegraphics[width=0.7\textwidth]{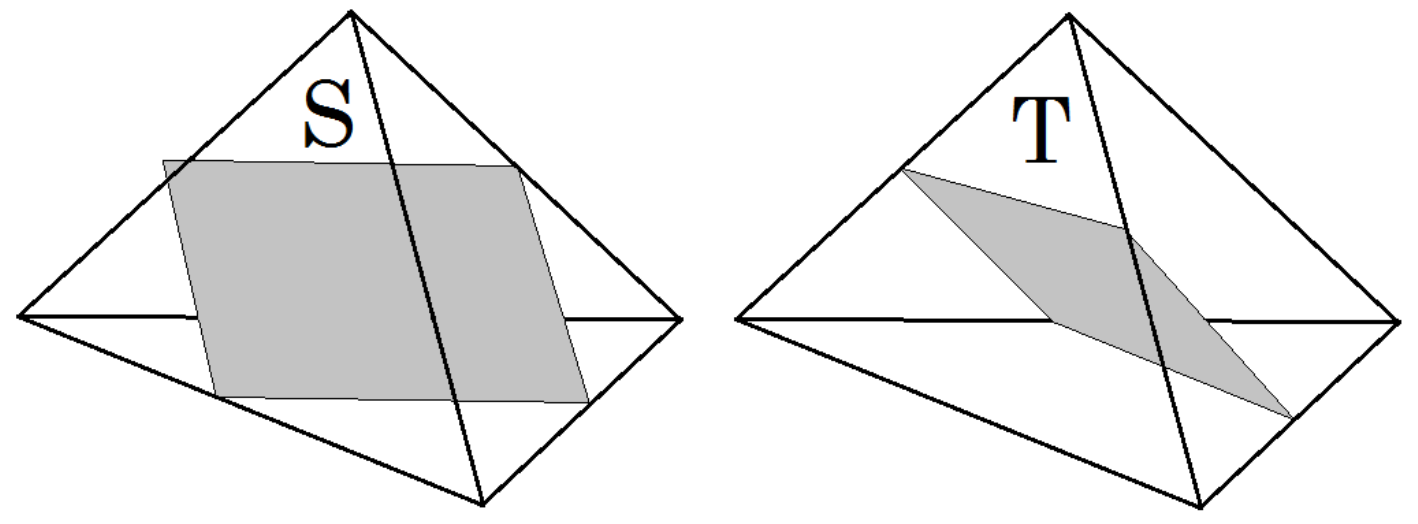}
    \caption{The spins $S$ and $T$ are the areas of the medial parallelograms of the tetrahedron. }  \label{fig_ST}
\end{figure}

To see this refer to figure \ref{fig_ST} and note that
\be
  |\vec{J}_{1} + \vec{J}_{2}|^2 = \vec{J}_{1} \cdot \vec{J}_{1} + \vec{J}_{2} \cdot \vec{J}_{2} + 2 \vec{J}_{1} \cdot \vec{J}_{2}
\ee
and use (\ref{eqn_S_eigen}).  However, to uniquely determine the geometry of a tetrahedron we require six pieces of information, such as the six edge lengths, or the four areas and two angles.  Therefore the geometry in each of the tetrahedra is fuzzy.

This is an issue when constructing spin foam models because discretized simplicity constraints are imposed on the intertwiner representation labels and are based on geometric arguments.  For this reason it is wise to instead consider a coherent representation of the intertwiners, which we review next.  In Chapter \ref{section_discrete_coherent} we introduce a new discrete-coherent basis which is labeled by two extra spins $S,T$ and hence defines the tetrahedral geometry uniquely.  Hence it is overcomplete but still finite dimensional.

\section{Coherent Intertwiners}
\label{section_coherent_intertwiners}

We now review the coherent intertwiner formalism introduced by Livine and Speziale \cite{Livine:2007vk}.  As opposed to the orthogonal basis which are labeled by spins, the coherent intertwiners are labeled by a continuous set of data: a spinor for each representation.  This means the coherent basis is overcomplete but it has the advantage that each state is peaked on points in phase space representing geometrical polyhedra.  

The other advantage of these states is that the amplitudes {\it exponentiate}.  This makes it easy for us to construct generating functionals, which we will use in Chapter \ref{chapter_exact} to compute these amplitudes exactly.

\subsection{SU(2) Coherent States}

Coherent states can be defined for arbitrary Lie groups \cite{Perelomov:1986tf}.  In the case of SU(2) this is particularly simple due to the Schwinger representation (\ref{eqn_schwinger}) in terms of a pair of independent harmonic oscillators.  


In the spin $j$ representation the coherent state $|j,z\ket$ is defined to be the holomorphic functional
\be \label{eqn_holo_func}
(w|j,z\ket \equiv \frac{[w|z\ket^{2j}}{(2j)!}.
\ee
These states possess the characteristic property that their scalar product with any spin $j$ state $|P\ket$ reproduces the functional $P(z)$, that is
\be \label{eqn_holo_rep}
\bra j, \check{z} | P\ket = P(z).
\ee
This follows from a  direct computation which shows that
$
(w|j,z\ket =\bra j,\check{w}|j,z\ket.
$
This property implies that we can identify the label $(z|$ of $(z|P\ket$ with the state $\bra j,z|$ when 
evaluated on a spin $j$ functional. In the following we will use interchangeably the notation $(z|=\bra j,\check{z}|$ for the labels.

From (\ref{eqn_SU2_rep}) this implies the coherence with respect to the SU(2) action
\be \label{eqn_CS_group_coherence}
  g | j, z \ket = |j, gz \ket \qquad \forall g \in SU(2)
\ee
where $g$ is in the fundamental representation.

These states resolve the identity on $V^j$ since
\be
  \int \rd \mu(z) \, \bra j, m | j, z \ket \bra j', z | j', m' \ket = \bra j, m | j', m' \ket = \delta_{j,j'} \delta_{m,m'}
\ee
which can be easily calculated using formula (\ref{eqn_gaussian_int_one_var}) in Appendix \ref{Gauss_int_app}.  Therefore the identity on $V^j(\C^2,\rd\mu)$ can be expressed as
\be
  \one_j = \int \rd \mu(z) | j, z \ket \bra j, z |
\ee
and so an arbitrary holomorphic function on the Bargmann-Fock space $L^2(\C^2,\rd \mu(z))$ can be expressed as
\be
  P(z) = \sum_{2j=0}^{\infty} \bra j, \check{z} | P\ket 
\ee
due to the analyticity.  This immediately implies the following formula for the overlap (and norm)
\be
  \bra j, w | j, z \ket = \frac{\bra w | z \ket^{2j}}{(2j)!}
\ee
which confirms (\ref{eqn_holo_rep}).

Besides the power of Gaussian integration, the coherent states also have nice semi-classical properties.  That is, they are peaked with minimal uncertainty around the expectations values
\be
   \frac{\bra j z | \vec{J} |j z \ket}{\bra j z | j z \ket} = j \frac{\vec{V}(z)}{|\vec{V}(z)|}
\ee
where the spinor corresponds to a classical 3-vector $\vec{V}(z) = \bra z | \vec{\sigma} | z \ket$ via the Penrose null-flag interpretation of the spinor
\be
  \|z \ket \bra z \| = \frac{1}{2} \left( \bra z | z \ket + \vec{V}(z) \cdot \vec{\sigma} \right)
\ee
The flag vector corresponding to the phase of the spinor also carries geometrical information \cite{Freidel:2010aq}.  In the case of intertwiners this will carry information about the extrinsic curvature.  In Theorem \ref{thm_twisted_action} we see how these phases arrange to define (generalized) 4d dihedral angles of a 4-simplex.

\subsection{Livine-Speziale Coherent Intertwiners}

Instead of using the Clebsch-Gordan coefficients to define invariant states, we can instead group average a basis of $V^{j_1} \otimes \cdots \otimes V^{j_n}$ to project into the intertwiner space (\ref{eqn_inter_space}).  This projector is sometimes known as the Haar projector.  Taking a coherent state basis labeled by $n$ spinors $z_{1},\cdots ,z_{n}$ the Livine-Speziale\footnote{Note that the normalization of these states is different from \cite{Livine:2007vk} to better suit the Bargmann scalar product (\ref{barg_in_prod}).} \cite{Livine:2007vk} coherent intertwiner is defined by group averaging the holomorphic functionals (\ref{eqn_holo_func}) as
\be \label{eqn_LS_coh_int}
 (w_{i} \|j_i,z_i\ket \equiv \int \rd g \prod_{i=1}^{n} \frac{ [w_{i}| g |z_i\ket^{ 2j_i}}{(2j_i)!}.
\ee

These states are coherent in the sense that their scalar product  reproduces the holomorphic functional, they are labeled by the continuous set of data $\{z_i\}$ and  they
resolve  the identity:
\be
  \bra j_{i}, w_{i} \|j_i,z_i\ket=(\check{w}_{i} \|j_i,z_i\ket,\qquad \one_{j_i} = \int \prod_i \rd\mu(z_i) \|j_i,z_i \ket \bra j_i,z_i\|. \label{coherent}
\ee
This is shown by using the identity $\int \rd\mu(w) \bra a|w\ket^{2j}\bra w|b\ket^{2j} = (2j)! \bra a | b \ket^{2j}$, which itself is proven by summing over $j$ and performing the Gaussian integration.  See Appendix \ref{Gauss_int_app}.

While the orthogonal basis can be defined by the eigenstates of the operators $\vec{J}_i \cdot \vec{J}_j$ these operators do not form a closed algebra \cite{Dupuis:2010iq}
\be
  \left[\vec{J}_1 \cdot \vec{J}_2, \vec{J}_1 \cdot \vec{J}_3 \right] = \vec{J}_1 \cdot \left(\vec{J}_3 \wedge \vec{J}_4 \right)
\ee

As mentioned above there exists a hidden U(N) structure to these coherent intertwiners which possess a set of observables which do form a closed algebra.  For each spinor $z_i$ let $a_i$ and $b_i$ be the pair of harmonic oscillators as in the Schwinger representation (\ref{eqn_schwinger}).  Then the differential operators are defined by
\be 
  J_{i}^{z} = \frac{1}{2} \left( a_{i}^\dagger a_i - b_{i}^\dagger b_i \right), \quad J^{+}_{i} = a_{i}^\dagger b_i, \quad J^{-}_{i} = a_i b_{i}^\dagger, \quad E_{i} = \frac{1}{2} \left( a_{i}^\dagger a_i + b_{i}^\dagger b_i \right)
\ee
where $E_i$ is the i'th representation and has the eigenvalue $2j_i$.

Just as we can take the ``square root'' of the quadratic Casimir by $J^2 = E(E+1)$, we can also remarkably take the square root of the scalar product operators $\vec{J}_i \cdot \vec{J}_j$ as
\be \label{eqn_scalar_op_sqrt}
  \vec{J}_i \cdot \vec{J}_j = \frac{1}{2} E_{ij} E_{ji} - \frac{1}{4} E_i E_j - \frac{1}{2} E_i
\ee
where the operators $E_{ij}$ are defined by
\be \label{eqn_Eij}
  E_{ij} = \frac{1}{2} \left( a_{i}^\dagger a_j + b_{i}^\dagger b_j \right)
\ee
These operators are precisely the generators of a $\mathfrak{u}(N)$ algebra satisfying the commutation relations
\be
  [E_{ij},E_{kl}] = \delta_{jk} E_{il} - \delta_{il}E_{kj}
\ee
We will use these operators to find the action of the scalar product operators in the new basis defined in the next chapter.

\subsection{The Coherent Amplitude and Spinor BF Theory}

The coherent vertex amplitude in four dimensions is constructed by contracting five coherent intertwiners (\ref{eqn_LS_coh_int}) in the pattern of a 4-simplex.  The contracted legs of the intertwiners must have the same spin $j_{ij} = j_{ji}$  where $i,j \in \{1,..,5\}$ label the vertices of a 4-simplex.  Each leg also carries two spinors $|z^{i}_{j}\ket \neq |z^{j}_{i}\ket$ where the upper index indicates the vertex the spinor belongs to and the lower index the vertex it connects to.

Performing the contraction the amplitude depends on 10 spins and 20 spinors
\be \label{eqn_coherent_4S}
  A_{4S}(j_{ij}, z^{i}_{j}) =\int_{\mathrm{SU(2)^{5}}}\prod_{i} \rd g_{i} \prod_{i<j} \frac{[z^{i}_{j}|g_{i}g_{j}^{-1}|z^{j}_{i}\ket^{2j_{ij}}}{(2j_{ij})!}.
\ee
This can be easily generalized to the $n$-simplex and also to arbitrary graphs.

The remarkable thing about this amplitude is that it exponentiates in the sum over the spins.  Furthermore, it can be written purely in terms of the spinor and group integrals.  As shown in \cite{Dupuis:2011fz} the BF theory parition function takes the form
\be
  Z^{\Delta^\ast}_{BF}  = \int \prod_{(e,f)} \rd \mu(z^{e}_{f}) \prod_f (\bra z^{e}_f|z^{e}_f\ket-1) \int \prod_{e} \rd g_e e^{\sum_{(e,f)} [z^{e}_{f}|g_e|z^{e}_{f}\ket}
\ee
where $e, f$ are edges and faces of $\Delta^\ast$.  Note that $z^{e'}_{f} = \check{z}^{e}_{f}$ if $f$ is directed from $e'$ to $e$ as with the usual rules of contraction, see Section \ref{sec_contraction}.  The dimension factors $2j_f+1$ for each face are taken care of by the insertion of the obervables $(\bra z^{e}_f|z^{e}_f\ket-1)$.

The asymptotics of the amplitude (\ref{eqn_coherent_4S}) have been studied extensively \cite{Conrady:2008mk,Barrett:2009gg} however the actual evaluation, i.e. the evaluation of the group integrals, has been left unsolved.  In Chapter \ref{chapter_exact} we construct a generating functional which computes these amplitudes exactly, and for arbitrary graphs.  

We find that this amplitude can be expressed as
\begin{align} 
  A_\Gamma(j_{ij},z^{i}_{j}) 
  &=  \sum_{k^{a}_{ij}} A_{\Gamma}(k_{ij}^{a})
\prod_{(a,i,j)}  \left(\frac{[z^{a}_{i}|z^{a}_{j}\ket^{k_{ij}^{a}}}{(J_{a}+1)!}\right). \nonumber
\end{align}
where $k^{a}_{ij} = k^{a}_{ji}$ is a set of six positive integers for each vertex $a$ such that $\sum_{j \neq a,i} k^{a}_{ij} = 2j_{ai}$.  This is described in detail in Section \ref{section_racah}.

For this reason it seems that the most natural basis of intertwiners is actually given by monomials in the holomorphic spinor invariants $[z^{a}_{i}|z^{a}_{j}\ket$ and is labeled by the degree $\{k^{a}_{ij}\}$ of these monomials.  We investigate this possibility in the next chapter.

\chapter{The Discrete Coherent Basis}
\label{section_discrete_coherent}

\section{A Basis of Monomial Invariants}

We will now show how to construct a new basis which is also coherent, resolves the identity, but is labeled by a discrete set of integers.  Since the product $[z|w\ket$ is holomorphic and SU(2) invariant it can be used to construct a  complete basis of the intertwiner space ${\cal H}_{n}\equiv \oplus_{j_{i}}{\cal H}_{j_{1}\cdots j_{n}}$  by
\be\label{C}
( z_{i} | k_{ij} \ket  \equiv \prod_{i<j}\frac{ [z_{i}|z_{j}\ket^{k_{ij}}}{k_{ij}!}.
\ee
This basis  is labeled by  $n(n-1)/2$ non-negative integers $[k]\equiv (k_{ij})_{i\neq j = 1,\cdots, n}$ with $k_{ij}=k_{ji}$.  
Note that we are free to choose a phase convention.\footnote{Later we will see that the asymptotic limit of the intertwiners will imply a canonical phase as noted in \cite{Barrett:2009gg}.}  The phase is affected by the implicit ordering chosen by the spinors $\{z_1,z_2,...,z_n\}$ and the convention of choosing pairs by $i<j$.

This basis was introduced by Schwinger \cite{schwinger2001angular} and also Bargmann \cite{bargmann1962representations} for studying generating functions of the 3nj-symbols and we have generalized it here to the $n$-valent case.  The $n$-valent states (\ref{C}) can also be used to construct generating functions for general graphs as was done in \cite{Bonzom:2012bn} and \cite{Freidel:2012ji}.  We will first review the 3-valent case and then go on to study the properties of the 4-valent case.

For a basis representing the subspace ${\cal H}_{j_{1}\cdots j_{n}}$ with fixed  spins $j_i$, we have $n$ homogeneity conditions which require the integers $[k]$ to satisfy
\be\label{kj}
\sum_{j\neq i} k_{ij} =2j_{i}.
\ee
The sum of spins at the vertex is defined by $J = \sum_i j_i = \sum_{i<j} k_{ij}$ and is required to be a positive integer.   
From the relation $[\check{w}|\check{z}\ket = \overline{[w|z\ket}$ we see that these states satisfy the reality condition 
\be\label{real}
\overline{( z_{i} | k_{ij} \ket} = ( \check{z}_{i} | k_{ij} \ket.
\ee
Furthermore, from the coherency property (\ref{coherent}) we can easily compute the overlap of these states with the coherent intertwiners:
\be
\bra j_{i}, \check{z}_{i} || k_{ij} \ket = ( z_{i} | k_{ij} \ket = \bra k_{ij} || j, z_{i} \ket .
\ee
where the last equality follows from 
the reality condition (\ref{real}) and the fact that $(-z_{i}|k_{ij}\ket = (z_{i}||k_{ij}\ket$.

In section \ref{section_relate_coh_discoh} we will use generating functionals to show that the scalar product of 
coherent intertwiners can be expressed in terms of the coefficients of the discrete basis as 
\be\label{fundrelation}
\bra j_{i}, \check{w}_{i} || j_{i}, z_{i}\ket = \sum_{[k]\in K_{j}} \frac{( {w}_{i} | k_{ij} \ket {( z_{i} | k_{ij} \ket} }{||[k]||^{2}},\quad \mathrm{with}\quad ||[k]||^{2} =\frac{ (J+1)!}{\prod_{i<j}k_{ij}!}.
\ee
where $K_{j}$ denotes all the $k_{ij}$ solution of (\ref{kj}).
This result in turn implies that 
\be \label{eqn_rel_coh_discoh}
|| j_{i}, z_{i}\ket = \sum_{[k]\in K_{j}} \frac{| k_{ij} \ket \bra k_{ij} \| j, z_{i} \ket }{||[k]||^{2}},
\ee
which expresses the coherent intertwiners in terms of the discrete basis.  We note that the scalar product (\ref{fundrelation}) was also studied in \cite{Bonzom:2012bn} where they use the notation $F_{ij} = [z_i|z_j\ket$.  In fact an explicit expression for the $n$-valent scalar product was derived there using the Pl\"ucker relations explicitly.

\subsection{3-valent Intertwiners}

In the case $n=3$ there is only one intertwiner.  Indeed, given $[k]=(k_{12},k_{23},k_{31})$ the  homogeneity restriction requires $2j_{1}=k_{12} + k_{13}$ which can be easily solved by
\be \label{eqn_3_k}
k_{12} = j_1 + j_2 - j_3,\qquad k_{13} = j_1 - j_2 + j_3, \qquad k_{23} = -j_1 + j_2 + j_3.
\ee
In this case the fact that homogeneous functions of different degree are orthogonal implies that $|{k_{12},k_{23},k_{31}}\ket$ form an orthogonal basis
\footnote{One can also arrive at this basis by considering the representation space of symmetrised spinors.  For details see appendix A of \cite{Rovelli:2004tv}.  The two approaches are essentially the same, however in the holomorphic representation we have the advantage of tools like generating functionals and Gaussian integration.} of (\ref{eqn_inter_space}).

Since there is only one holomorphic  function $( z_{i}|{[k]}\ket $ it must be proportional to the Wigner 3j symbol.  Furthermore from the relation (\ref{eqn_rel_coh_discoh}) and the resolution of identity we can read off the norm of these states
\be
  \bra k_{12},k_{23},k_{31}|k_{12},k_{23},k_{31}\ket = \|[k]\|^{2} = \Delta^2(j_1 j_2 j_3) 
\ee
where the triangle coefficients were defined in (\ref{eqn_tri_coeff}).  It can be shown \cite{bargmann1962representations}
\begin{align} \label{C3}
 ( z_{i}| k_{12},k_{23},k_{31} \ket
&= \Delta(j_1 j_2 j_3) \sum_{m_1 m_2 m_3} \threej{j_1}{j_2}{j_3}{m_1}{m_2}{m_3} e^{j_1}_{m_1}(z_1) e^{j_2}_{m_2}(z_2) e^{j_3}_{m_3}(z_3) \\
&= \Delta(j_1 j_2 j_3) ( z_{i}\| j_1,j_2,j_3 \ket \nonumber
\end{align}
where we defined the Wigner $3j$ states $\| j_1,j_2,j_3 \ket$ in (\ref{eqn_3j_state}).  Note that we could divide $| k_{12},k_{23},k_{31}\ket$ by $\Delta(j_1 j_2 j_3)$ to normalize this basis, but it will be simpler to instead work with these unnormalized states.

\subsection{Counting}

For $n>3$ there are more basis elements $|{k_{ij}}\ket$ than the dimension of the intertwined space so the basis is no longer orthogonal.  Indeed, since we have $n(n-1)/2$ $k_{ij}$'s satisfying $n$ relations (\ref{kj}) these intertwiners are labeled by 
$n(n-3)/2$ integers. But this is clearly more that the dimension of the Hilbert space of $n$-valent intertwiners, which is known to be labeled by $n-3$ integers, i.e. by contracting only 3-valent nodes.
This means that the basis given above is { \it overcomplete}.

Another way to understand this counting is to recall that the algebra of gauge invariant operators acting on ${\cal H}_{j_{1},\cdots, j_{n}}$ is given by $J_{ij} \equiv J_{i}\cdot J_{j}$ for $i\neq j$ where $J_{i}$ denotes the angular momentum operator action in the $i$ direction.
These operators satisfy the closure relation $\sum_{i} J_{i} =0$ and the action of $J_{i}^{2}$ is given by multiplication by $j_{i}(j_{i}+1)$.
These relations mean that we can express any instance of $J_{n}$ say, by a summation of operators depending on $J_{i}$ for $i<n$. Thus a good basis of operator is for instance $J_{ij}$ for $i\neq j $ and $i,j<n$.
There are $(n-1)(n-2)/2$ such operators. They satisfy one relation that stems from the closure relation which is\be
\sum_{i\neq j <n} J_{ij} = j_{n}(j_{n}+1) - \sum_{i<n}j_{i}(j_{i}+1).
\ee
This makes it clear that if we want to maximally represent these operators we need $ n(n-3)/2$ labels.
These operators do not commute, therefore these labels represent an overcomplete basis.  A maximal commuting subalgebra  is of dimension $n-3$.

For example, in the case $n=4$ the basis is  labeled by $2$ integers while we need only one, and
for $n=5$  it is  labeled by $5$ integers while we need only two.  Despite this overcompleteness we will be able to determine all of the necessary properties of these states and we will discover some interesting relations between the orthogonal bases on the one hand and coherent intertwiners on the other.

\section{The 4-valent case}
\label{4-val}

We now focus on the case $n=4$. A very convenient labeling of the basis $|{[k]}\ket$ is done in terms of three spins $S$, $T$, $U$ 
which refer to the three channels in which a 4-valent vertex can be split into two three valent ones.
The relationship between these labels and the $k$ labels is given by
\be\label{int1}
S\equiv j_{1}+j_{2} -k_{12},\quad T\equiv j_{1}+j_{3}  -k_{13},\quad U\equiv j_{1}+j_{4}-k_{14}.
\ee
where $S$, $T$, and $U$ are such that the $k_{ij}$ are non-negative integers.  The constraints in (\ref{kj}) imply that 
$ j_{1}+j_{2} -(j_{3}+j_{4}) = k_{12}-k_{34}$, thus we also have 
\be\label{int2}
S= j_{3}+j_{4}-k_{34},\quad T= j_{2}+j_{4}-k_{24},\quad U= j_{2}+j_{3}-k_{23}.
\ee
Summing over all $k_{ij}$ shows that $S$, $T$, and $U$ are not independent but satisfy the relation
 \be
 S+T+U=J.
 \ee 
where $J = j_1+j_2+j_3+j_4$.  We can therefore label the $4$-valent intertwiner basis by the four spins $j_i$ and two extra spins $S,T$ and we will henceforth denote by 
$k_{ij}(j_{i},S,T)$ the corresponding integers in (\ref{int1}, \ref{int2}). 
These integers cannot take arbitrary values, since  $k_{ij}$ are restricted by
$0\leq  k_{ij} \leq  \mathrm{max}(2j_{i},2j_{j})$, this restriction\footnote{It is given by
\bea
\mathrm{max}(|j_{1}-j_{2}|, |j_{3}-j_{4}|) \leq &S& \leq \mathrm{min}(j_{1}+j_{2}, j_{3}+j_{4}), \\
\mathrm{max}(|j_{1}-j_{3}|, |j_{3}-j_{4}|) \leq &T& \leq \mathrm{min}(j_{1}+j_{3}, j_{3}+j_{4}),\\
\mathrm{max}(j_{1}+j_{4}, j_{2}+j_{3}) \leq S &+&T \leq  J - \mathrm{max}(|j_{1}-j_{4}|, |j_{2}-j_{3}|).
\eea 
} is denoted by $(S,T)\in {\cal N}_{j_{i}}$.
In the case all spins are equal to $N/2$ this is simply $0\leq S,T\leq N$, $N\leq S+T\leq 2N$.

 We will  denote the corresponding basis by
  $|S,T\ket_{j_{i}}$ where
 \be \label{eqn_ST_notation}
 |S,T\ket_{j_{i}} \equiv  |[k](j_{i},S,T)\ket.
\ee
In the following we will omit the subscript $j_{i}$  and use the shorthand $|S,T\ket \equiv \left| S, T \right\ket_{j_{i}}$ for notational simplicity when the context is clear and the external spins are fixed.

\subsection{Overcompletness and Identity Decomposition}

As discussed above, the $|S,T\ket$ basis has one extra label and is thus overcomplete.  We will now investigate the nature of the relations among these states which is summarized by the following theorem:
\begin{theorem}
The $|S,T\ket$ states are  not linearly independent; all the relations among them are generated by the fundamental relation
\be \label{fund_rel}
 (k_{12}+1)(k_{34}+1) \left| S-1,T \right\ket - (k_{13}+1)(k_{24}+1) \left|  S,T-1 \right\ket +k_{14}k_{23} \left| S,T \right\ket =0 
\ee
where $k_{ij}$ stands for $k_{ij}(j_{i},S,T)$.
\end{theorem}

It turns out that the relation among the states is easily seen in the holomorphic representation.
It is well known that the gauge invariant quantities $[z_{i}|z_{j}\ket$ are not independent,
they satisfy the  Pl\"ucker relation:
\be \label{eqn_R_plucker}
R(z_i)\equiv [z_{1}|z_{2}\ket[z_{3}|z_{4}\ket - [z_{1}|z_{3}\ket[z_{2}|z_{4}\ket + [z_{1}|z_{4}\ket[z_{2}|z_{3}\ket =0.
\ee
In order to write the effect of this relation on the states $|S,T\ket_{j_{i}}$
let's compute first the effect of multiplication by one monomial
\bea
[z_{1}|z_{2}\ket[z_{3}|z_{4}\ket ( z_{i} | S,T\ket_{j_{i}-\frac12}
= ({k_{12}k_{34}})(j_{i},S,T ) ( z_{i} | S,T+1\ket_{j_{i}}
\eea
where we used that $k_{12}(j_{i}-\frac12,S,T)  +1 = k_{12}(j_{i},S,T )= k_{12}(j_{i},S,T +1)$,
while 
$k_{13}(j_{i}-\frac12,S,T) = k_{13}(j_{i},S,T)-1 = k_{13}(j_{i},S,T+1)$,
and $k_{14}(j_{i}-\frac12,S,T) = k_{14}(j_{i},S,T)+ 1 = k_{14}(j_{i},S,T+1)$.
Performing similar computations for the different monomials we find that the multiplication by the Pl\"ucker relation can be implemented in terms of an operator $\hat{R} : {\cal{H}}_{j_{i}-\frac12} \rightarrow {\cal{H}}_{j_{i}}$
whose image vanishes identically.  It is defined by
 $  R(z_{i}) ( z_{i} | S,T\ket_{j_{i}-\frac12} = ( z_{i}| \hat{R} | S,T\ket_{j_{i}-\frac12}$ where $\hat{R}$ is given by 
\be
\hat{R} | S,T\ket_{j_{i}-\frac12} = 
{k_{12}k_{34}} | S,T+1\ket_{j_{i}} 
- {k_{13}k_{24}} |S+1,T\ket_{j_{i}} 
+ (k_{14}+2)(k_{23} +2)  |S+1,T+1\ket_{j_{i}} 
\ee
here $k_{ij}$ denotes $ k_{ij}(j_{i}, S, T)$.
By shifting the parameters $S\to S-1$ and $T\to T-1$ and using that 
$ k_{12}(j_{i}, S -1, T-1)= k_{12} +1$ etc. we obtain the desired relation stated in the theorem.
 By taking powers of the operator $\hat{R}$ we can generate many more relations which we will discuss in a later section.
Despite the linear dependence of these states they admit a resolution of identity, consistent with a coherent state basis:  
\begin{theorem} \label{thm_completeness}
The resolution of identity on the space of 4-valent intertwiners has the simple form
\be \label{eqn_res_id}
  \one_{{\cal H}_{j_{i}}} = \sum_{S,T}  \frac{|S,T\ket \bra S,T|}{\|S,T\|_{j_{i}}^2}, \qquad \|S,T\|^2_{j_{i}} \equiv \frac{(J+1)!}{\prod_{i<j}k_{ij}!}.
\ee
\end{theorem}

\begin{proof}
To show this we introduce the following generating functional which depends holomorphically on $n$ spinors $|z_{i}\ket$ and $n(n-1)/2$ complex numbers $\tau_{ij} =-\tau_{ji}$
\be
 {\cal C}_{\tau_{ij}}(z_i)  \equiv  \sum_{[k] } 
 \prod_{i<j} \tau_{ij}^{k_{ij}} (z_i|k_{ij}\ket = e^{\sum_{i<j} \tau_{ij}[z_{i}|z_{j}\ket },
\ee
Here the sum is over all non-negative integers $[k]$ and not just those satisfying (\ref{kj}).  This functional was first considered by Schwinger \cite{schwinger2001angular}.  The scalar product between the generating functionals is
\be\label{CC22}
\left\bra {\cal C}_{\tau_{ij}} | {\cal C}_{\tau_{ij}} \right\ket =  \int \prod_{i}\rd\mu(z_{i})   \left|{\cal C}_{\tau_{ij}}(z_{i})\right|^{2}
 =   \int \prod_{i}\rd\mu(z_{i}) e^{\sum_{i<j} \tau_{ij}[z_{i}|z_{j}\ket + \bar{\tau}_{ij} \bra z_{j}|z_{i}]}.
\ee
This integral is Gaussian and can be shown to have the following exact evaluation (for more details see \cite{Freidel:2012ji})
\be \label{eqn_scalar_det}
\left\bra {\cal C}_{\tau_{ij}} | {\cal C}_{\tau_{ij}} \right\ket
= \frac{1}{\det(1 + T\overline{T})}
\ee
where $T = (\tau_{ij})$ and $\overline{T} = (\overline{\tau}_{ij})$ are $n \times n$ antisymmetric matrices.  This determinant can be evaluated explicitly and in the case $n=4$ it has the form 
\be\label{det4}
\det(1 + T\overline{T}) = \left(1-\sum_{i<j} |\tau_{ij}|^{2} +  |R(\tau)|^{2}\right)^{2}
\ee
where $R(\tau)= \tau_{12}\tau_{34} +  \tau_{13} \tau_{42}+ \tau_{14}\tau_{23}$.  Notice that $R(\tau)$ is the Pl\"ucker identity and when $\tau_{ij} =[z_{i}|z_{j}\ket$ it vanishes.  Now by expanding the LHS of (\ref{CC22}) in the notation of (\ref{eqn_ST_notation})
\bea \label{eqn_scalar_expanded} 
 \left\bra {\cal C}_{\tau_{ij}} | {\cal C}_{\tau_{ij}} \right\ket&= & \sum_{j_i,S,T} \sum_{j'_i,S',T'} \left\bra S,T \right|\left. S',T' \right\ket \prod_{i<j} \bar{\tau}_{ij}^{k_{ij}} \tau_{ij}^{k_{ij}'}
 \eea
we see that the generating functional contains information about the scalar products of the new intertwiner basis.  
We now have two different ways to evaluate the scalar product for the generating functional 
with $\tau_{ij} = [z_i|z_j\ket$.
On one hand we start from (\ref{eqn_scalar_expanded}) to get 
\bea
 \left\bra {\cal C}_{[z_i|z_j\ket} | {\cal C}_{[z_i|z_j\ket} \right\ket&= &  \sum_{j_i,S,T} \sum_{j'_i,S',T'}(J+1)!^{2} \frac{( z_i | S,T \ket \bra S,T | S',T' \ket \bra S', T' | z_i )}{\|S,T\|^2 \|S',T'\|^2}. 
\eea
here we used the definition of our states and normalization:
\be
 \prod_{i<j} [z_i|z_j\ket^{k_{ij}} = (J+1)!\frac{ ( z_i | S,T \ket}{\|S,T\|^2}.
\ee
On the other hand we can evaluate directly  the product by expanding (\ref{det4}) when $R(\tau)=0$.
This gives
\be
 \left\bra {\cal C}_{[z_i|z_j\ket} | {\cal C}_{[z_i|z_j\ket} \right\ket =
 \sum_{j_i,S,T}(J+1)!^{2}\, \frac{( z_i | S,T \ket \bra S, T | z_i )}{\|S,T\|^2}.
\ee
equating the two expressions gives the identity decomposition
\be
\sum_{S',T'} \frac{\bra S,T | S',T' \ket \bra S', T' | z_i )}{ \|S',T'\|^2} = \bra S, T | z_i ).
\ee
This completes the proof.
\end{proof}

We will show that despite the fact that they are discrete, the $|S,T\ket$ basis shares many of the same properties as the coherent intertwiners such as the correspondence with classical tetrahedra in the semi-classical limit. In addition the $|S,T\ket$ states also possess a simple relation with the orthogonal basis as  we will show in the next section.

\subsection{The Relation with the Orthogonal Basis}

In the previous sections we introduced a new and overcomplete basis of the space of 4-valent intertwiners which provided a simple decomposition of the identity.  On the other hand, the standard basis of 4-valent intertwiners is orthogonal, and is defined by the eigenstates of either of the invariant operators $J_{1}\cdot J_{2}$ or $J_{1}\cdot J_{3}$ or $J_{1}\cdot J_{4}$.  We will denote these orthogonal bases by $|S\ket$ and $|T\ket$ and $|U\ket$ respectively.  We would now like to investigate the action of the $S$ and $T$ channel operators $J_{1}\cdot J_{2}$ and $J_{1}\cdot J_{3}$ on $\left|S,T\right\ket$ as well as the relationship between the four  bases: $\left|S,T\right\ket$, $\left|S\right\ket$, $\left|T\right\ket$, $ | U\ket$. 

It is well known that, up to normalization,
the usual 4-valent intertwiner basis is obtained by the composition of two trivalent intertwiners.
For now we will focus on the $|S\ket$ states, which in the holomorphic representation, are defined to be
\bea \label{eqn_S_def}
\left( z_i | S \right\ket \equiv \int \rd\mu(z) {C}_{(j_{1},j_{2},S)}(z_{1},z_{2},\check{z}) {C}_{(S,j_{3},j_{4})}(z,z_{3},z_{4}),
\eea
where $|\check{z}\ket \equiv |z]$ and ${C}_{(j_{1},j_{2},S)}(z_{1},z_{2},\check{z}) = (z_{1},z_{2},\check{z}|k_{ij}(j_{1},j_{2},S)\ket$.  The expression (\ref{eqn_scalar_op_sqrt}) for the scalar product operators $J_{i}\cdot J_{j}$ repeated here for convenience
\be
  2J_{i}\cdot J_{j} = E_{ij}E_{ji}-\frac12 E_{ii}E_{jj} - E_{ii}.
\ee
can be written in the holomorphic representation using
\be
  E_{ij}\equiv z_{i}^{A}\partial_{z_{j}^{A}}.
\ee 

The operator $E_{ij}$ acts nontrivially only on a function of $z_{j}$ and its action amounts to replacing $z_{j}$ by $z_{i}$, i-e
 \be
E_{ij}\cdot [z_{j}|w\ket = [z_{i}|w\ket.
 \ee
Using this we can now compute the action of $J_1 \cdot J_2$ on $|S\ket$.  First note that the action of $E_{ii}$ on $|S\ket$ is given by $2j_{i}$ and the action of $E_{12}E_{21}$ is given by $(j_1-j_2+S)(-j_1+j_2+S+1)$.  Therefore the action of $J_1 \cdot J_2$ on $|S\ket$ is found to be
\bea \label{eqn_eigen_J_dot_J}
  J_{1}\cdot J_{2} \left|S\right\ket &=& \frac12\left(S(S+1) - j_{1}(j_{1}+1) - j_{2}(j_{2}+1)\right) \left|S\right\ket.
\eea

We are now in a position to discuss the physical interpretation of the spins $S$ and $T$.  From equation (\ref{eqn_eigen_J_dot_J}) we see that the operator $(J_1 +J_2)^2$ is diagonal in the $|S\ket$ basis with eigenvalue $S(S+1)$.  In \cite{Baez:1999tk} it is pointed out that if $A_1$ and $A_2$ are the classical area vectors of two faces of a tetrahedron then $|A_1 + A_2|^2$ is equal to four times the area of the medial parallelogram between the two faces.  The spins $T$ and $U$ would then be the areas of the other two medial parallelograms in the tetrahedron.  

This interpretation, however, does not hold for the $|S,T\ket$ states as we will see by computing the action $J_1 \cdot J_2$ on $|S,T\ket$.  We will find the true correspondence with the classical variables when we study the semi-classical limit.

\begin{theorem}
The action of $J_1 \cdot J_2$ on $|S,T\ket$ does not change the value of $S$ and it is given by
\begin{align} \label{Jact1}
2J_{1}\cdot J_{2} \left|S,T \right\ket = \left(S(S+1) - j_{1}(j_{1}+1) - j_{2}(j_{2}+1)\right) \left|S,T \right\ket \\
+\left( (k_{14}+1)(k_{23}+1)  \left|S,T-1 \right\ket -
k_{14}k_{23}  \left|S,T \right\ket  \right)  \nonumber \\
+\left( (k_{13}+1)(k_{24}+1) \left|S,T+1 \right\ket -
k_{13}k_{24}  \left|S,T \right\ket  \right). \nonumber
\end{align}
where $k_{ij}$ stands for $k_{ij}(j_i,S,T)$.  Similarly the action of $J_1 \cdot J_3$ does not change the value of $T$.
\end{theorem}
\begin{proof}
The action of $E_{ii}$ on $|S,T\ket$ is given by $2j_{i}$ while the action of $E_{12}E_{21}$ on $|S,T\ket$ is
\be
  \left( k_{13}(k_{23}+1) + k_{14}(k_{24}+1) \right)|S,T\ket + (k_{13}+1)(k_{24}+1)|S,T+1\ket + (k_{14}+1)(k_{23}+1)|S,T-1\ket. 
\ee
Now with this and the relation 
\be
  k_{13}k_{23} + k_{14}k_{24} = S^{2}- (j_{1}-j_{2})^{2} - k_{13}k_{24}-k_{14}k_{23}
\ee
we find the desired result.  The action of $J_1 \cdot J_3$ can be deduced from a permutation exchanging $1$ and $3$, under such a permutation $J_1 \cdot J_3 \to J_1 \cdot J_2$ and $ (-1)^{k_{23}} |S,T\ket \to  |T, S\ket$.  Similarly under an exchange of $1$ and $4$, $J_1 \cdot J_4 \to J_1 \cdot J_2$ and $ (-1)^{k_{23}+k_{34}} |S,T\ket \to  |U, T\ket$ .
\end{proof}

While the $S$ and $T$ spins don't share the interpretation of areas of parallelograms like in the orthogonal basis (since there are extra diagonal terms), it turns out that they are still closely related as we will now show.  First of all, notice that the coefficient of the first term in (\ref{Jact1}) is the same as the eigenvalue in (\ref{eqn_eigen_J_dot_J}).  Furthermore, if one sums over $T$ in (\ref{Jact1}) it can be seen that the last two terms cancel out because $k_{13}(j_{i},S,T-1)= k_{13}(j_{i},S,T)+1$,
$k_{14}(j_{i},S,T+1)= k_{14}(j_{i},S,T)+1$... and so on.  
Therefore $\sum_{T} \left|S,T \right\ket$ is {\it  proportional } to $\left| S \right\ket$.  What we will now show in the following theorem is that the proportionality constant is exactly one.
\begin{theorem} \label{thm_sum_T}
The orthogonal basis is obtained from the $\left|S,T \right\ket$ basis by summing over the $S$ or $T$ channels
\be
\left| S \right\ket =\sum_{T}  \left|S,T \right\ket, \hspace{12pt} \left| T \right\ket =\sum_{S} (-1)^{k_{23}}  \left| S,T \right\ket, \quad \left| U \right\ket =\sum_{S+T=J-U} (-1)^{k_{23}+k_{34}}  \left| S,T \right\ket.
\ee
\end{theorem}
\begin{proof}
Using the generating functionals in (\ref{defC}) in analogy with the definition (\ref{eqn_S_def}) of $|S\ket$ we can perform the following Gaussian integral 
\begin{align} \label{eqn_gen_fun_S_ST}
  & \int \rd\mu(z){\cal C}_{(\tau_{1},\tau_{2},\tau_{12})}(z_{1},z_{2},\check{z}) {\cal C}_{(\tau_{3},\tau_{4},\tau_{34})}(z,z_{3},z_{4})  \\
  &= e^{\tau_{12}[z_1|z_2\ket + \tau_{34}[z_3|z_4\ket} \int \rd\mu(z) e^{\tau_{1}[\check{z}|z_1\ket + \tau_{2}[\check{z}|z_2\ket} e^{\tau_{3}[z|z_3\ket + \tau_{4}[z|z_4\ket} 
  = e^{\sum_{i<j} \tau_{ij} [z_i|z_j\ket} 
  = {\cal C}_{(\tau_{ij})}(z_{i}), \nonumber
\end{align}
where $|\check z\ket = |z]$ and $\tau_{13}=\tau_{1}\tau_{3}$, $\tau_{14}=\tau_{1}\tau_{4}$, $\tau_{23}=\tau_{2}\tau_{3}$, $\tau_{24}=\tau_{2}\tau_{4}$.  Now let $k_{1} = j_{1}-j_{2}+S$, $k_{2} = j_{2}-j_{1}+S$, $k_{3} = j_{3}-j_{4}+S$, and $k_{4} = j_{4}-j_{3}+S$ as prescribed by (\ref{eqn_3_k}). Then looking at the coefficient of 
\be
\tau_{12}^{k_{12}}\tau_{1}^{k_{1}} \tau_{2}^{k_{2}}\tau_{3}^{k_{3}} \tau_{4}^{k_{4}}\tau_{34}^{k_{34}}
\ee
we get the conditions $k_{1} = k_{13}+k_{14}$, $k_{2} = k_{23}+k_{24}$, $k_{3} = k_{13}+k_{34}$, and $k_{4} = k_{14}+k_{24}$.  These conditions are trivially satisfied if the $k_{ij}$ are defined as in (\ref{int1}) and (\ref{int2}) which can be seen for instance by adding $k_{12}$ to the first condition.  Notice, however that the LHS of (\ref{eqn_gen_fun_S_ST}), when expanded, is a sum over $j_i$ and $S$ whereas the RHS is a sum over $j_i$, $S$, and $T$.  Thus we get the identity
\be
\int \rd\mu(z) {C}_{(j_{1},j_{2},S)}(z_{1},z_{2},\check{z}) {C}_{(S,j_{3},j_{4})}(z,z_{3},z_{4})
= \sum_{T} ( z_{i}| S,T\ket_{j_{i}}.
\ee
which implies $|S\ket = \sum_T |S,T\ket$.The other identities are obtained by permutation of indices.
\end{proof}
\begin{figure} 
  \centering
    \includegraphics[width=0.6\textwidth]{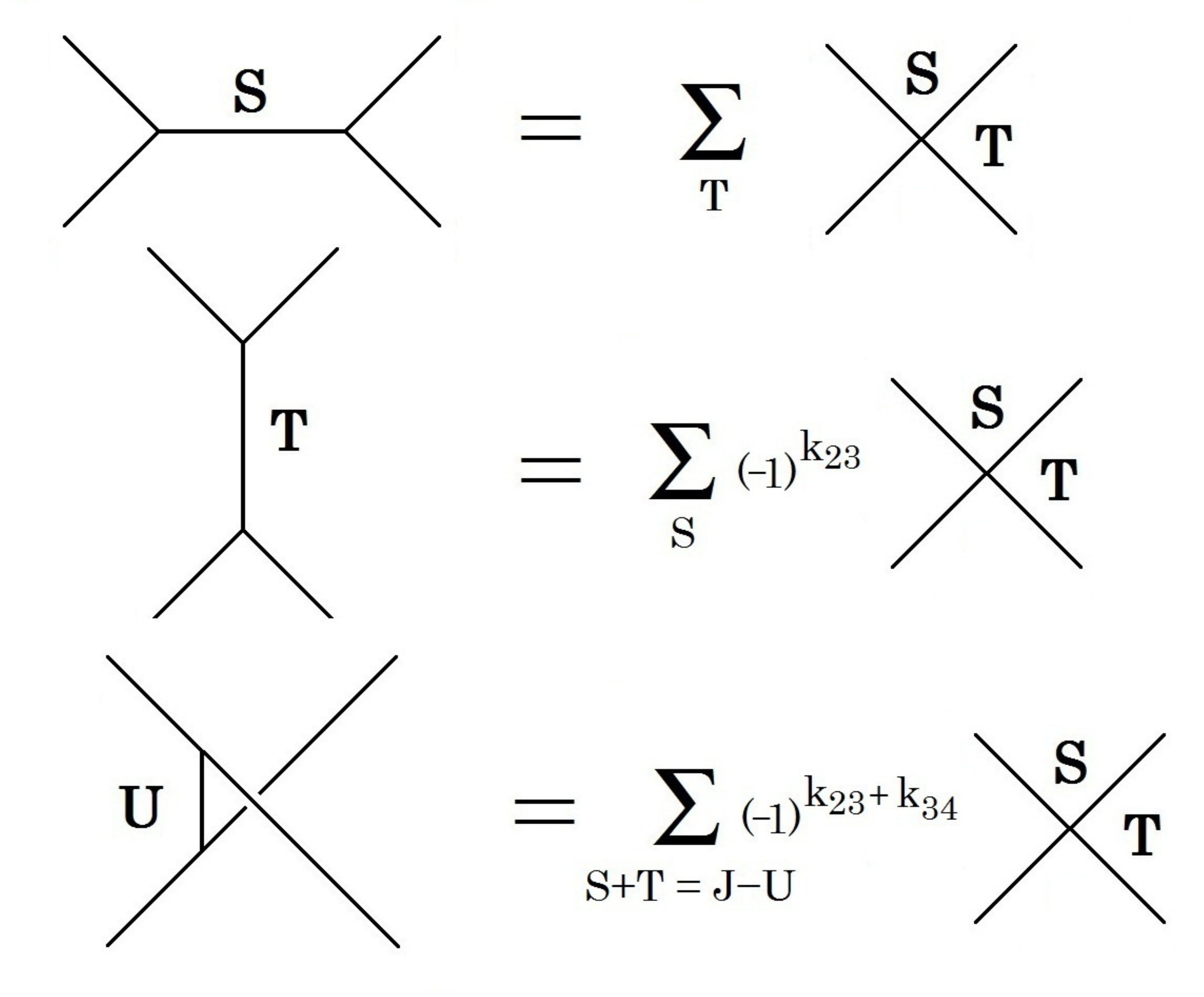}
    \caption{The graphical representation of theorem \ref{thm_sum_T} where a 4-valent vertex labeled by S and T is summed over T to produce the S channel decomposition.}  \label{fig_sumT}
\end{figure}
This last theorem shows that the $|S\ket$ and $|T\ket$ or $|U\ket$ bases are generated by the $|S,T\ket$ basis.  This is particularly useful for instance when describing spin-network amplitudes containing 4-valent nodes since a choice of $S$ or $T$ basis must be made at every such node.  The amplitude written in the $|S,T\ket$ basis however will generate all the different kinds of amplitudes by simply summing over the labellings.  For example the $15j$ symbol comes in five different kinds depending on the basis choice at the five nodes.  Thus a new symbol labeled by 20 spins, i.e. ten $j_i$, five $S$'s and five $T$'s, based on the $|S,T\ket$ basis would be a generator of these various symbols.  Moreover this $20j$ symbol would be the amplitude corresponding to the coherent 4-simplex.  We will define this new symbol shortly.

\subsection{Scalar Products}

In this section we will compute the scalar product in the $|S,T\ket$ basis and demonstrate the utility of Theorem \ref{thm_sum_T} by generating all the various other scalar products.  Let us first make a general remark about the form of the scalar product that follows from the  resolution of identity in (\ref{eqn_res_id}).
Let us split the scalar product into the naive product and the remainder:
\be
\left\bra S,T  \right.\left| S',T'  \right \ket = \|S,T\|^2 \delta_{S,S'} \delta_{T,T'} + O_{S,T}^{S',T'}
\ee
The resolution of the identity implies that 
\be
\sum_{S',T'} O_{S,T}^{S',T'} \frac{\bra S',T' |}{||S',T'||^{2}} =0= \sum_{S,T}\frac{|S,T\ket}{||S,T||^{2}} O_{S,T}^{S',T'}
\ee
This means that the reminder belongs to the algebra generated by the  fundamental relation in (\ref{fund_rel}).
These relations can be derived by considering the product of the operator $\hat{R}:H_{j_{i}-\frac12} \mapsto {\cal H}_{j_{i}}$ introduced previously: $R(z_i)^N \bra z_{i} |S,T\ket = \bra z_{i} | \hat{R}^{N}|S,T\ket$, where
 $R(z_i)$ is the Pl\"ucker relation given in (\ref{eqn_R_plucker}).
 Expanding $R(z_i)^N$ using the multinomial theorem we find
 \be\label{rel}
\left(R(z_i)\right)^N \prod_{i<j} [z_i|z_j\ket^{k_{ij}(j_i-N/2,s,t)} = \sum_{S,T} R^{(s,t)}_{(S,T)}(N) \prod_{i<j} [z_i|z_j\ket^{k_{ij}(j_i,S,T)} =0,
 \ee
where the summation coefficients are given by
\be\label{eqn_R_def}
R^{(s,t)}_{(S,T)}(N) = \frac{(-1)^{t-T+N}N!}{[s-S+N]![t-T+N]![S-s+T-t-N]!}
\ee
and the sum is over $S=s+N-a$, $T=t+N-b$ with $ a,b\geq 0$ and $ a+b \leq N$.  
From this result and the definition of the states $ \bra z_{i} |S,T\ket = ||S,T||_{j_{i}} / (J+1)!  \prod_{i<j} [z_i|z_j\ket^{k_{ij}(j_i,S,T)}$, we can write this relation  as
\be \label{eqn_R_action}
 \frac{ \hat{R}^{N}  | s,t\ket_{j_{i}-N/2}}{ ||s,t||_{j_{i}-N/2}^{2 }}  =  \frac{(J+1)!}{(J-2N +1)!} 
 \left(\sum_{S,T} R^{(s,t)}_{(S,T)}(N) \frac{| S,T \ket_{j_{i}}}{ ||S,T ||_{j_{i}}^{2 }}\right) =0.
\ee
The coefficients in the sum vanish if any of the arguments in the factorials is negative.  Note that for $N=1$ we recover the fundamental relation (\ref{fund_rel}).  

Now that we have determined the linear relations among the basis states we can deduce the exact form of the scalar product
\begin{proposition}\label{product1}
The scalar product is given by
\be\label{scalar}
\left\bra S,T  \right.\left| S',T'  \right \ket = \|S,T\|^2 \delta_{S,S'} \delta_{T,T'} +
\sum_{s,t,N} \frac{(-1)^N }{N!}\frac{(J-N+1)!}{\prod_{i<j} k_{ij}(j_i-N/2,s,t)!} 
{ R^{(s,t)}_{(S,T)}(N) R^{(s,t)}_{(S',T')}(N) }{}
\ee
\end{proposition}
The proof of this formula is given by equations (\ref{eqn_scalar_proj}) and (\ref{eqn_proj_kern}) together with (\ref{eqn_R_action}).

\subsection{Constraints Quantisation}
\label{Sec_const_quant}

We would like now to develop a deeper understanding of the construction just given of the scalar product.  We have seen that the complexity of the scalar product comes from the imposition of the constraints $\hat{R}=0$. This suggest that we should be able to understand the previous construction in terms of constraint quantization.
In order to do so, let's introduce the auxiliary  Hilbert space ${\cal \widehat{H}}_{j_{i}}$ with an orthogonal basis 
$|S,T)_{j_{i}}$ having $(S,T)\in {\cal N}_{j_{i}}$ and the scalar product
\be
(S',T'|S,T) =||S,T||_{j_{i}}^{2} \delta_{S,S'}\delta_{T,T'}.
\ee 
For this Hilbert space the decomposition of the identity takes the canonical form
\be
\one_{{\cal \widehat{H}}_{j_{i}}} = \sum_{S,T}  \frac{|S,T)( S,T|}{\|S,T\|_{j_{i}}^2}.
\ee
We  define the operator $\hat{R} : {\cal \widehat{H}}_{j_{i}-\frac12}\mapsto {\cal \widehat{H}}_{j_{i}}$ by
\be
\hat{R}|S,T)_{j_{i}-\frac12} \equiv 
{k_{12}k_{34}} | S,T+1)_{j_{i}} 
- {k_{13}k_{24}} |S+1,T)_{j_{i}} 
+ (k_{14}+2)(k_{23} +2)  |S+1,T+1)_{j_{i}} 
\ee
Its powers can be evaluated in terms of the coefficients introduced it the previous section, we find 
\be\label{matrixel}
{}_{j_{i}} (S,T| \hat{R}^{N}  | s,t)_{j_{i}-N/2} = ||s,t||_{j_{i}-N/2}^{2 }  \frac{(J+1)!}{(J-2N +1)!} 
  R^{(s,t)}_{(S,T)}(N).
\ee
The operator $\hat{R}$ is not hermitian, however the operator 
$$H \equiv \hat{R}^{\dagger} \hat{R} $$ is an hermitian operator, being positive its kernel coincides with the kernel of $\hat{R}$. 
The intertwiner Hilbert space is defined as the quotient of this auxiliary Hilbert space by the relation $H=0$. This means that $ {\cal {H}}_{j_{i}} = \mathrm{Im}\Pi_{j_{i}} $, where $\Pi^{2}_{j_{i}}=\Pi_{j_{i}}$ with $\Pi_{j_{i}} : {\cal \widehat{H}}_{j_{i}}\to {\cal \widehat{H}}_{j_{i}}$ the projector onto the 
kernel of $H$.  This means that the intertwined states are related to the auxiliary states as 
$$ |S,T\ket_{j_{i}} = \Pi_{j_{i}} |S,T)_{j_{i}}$$ and the physical scalar is given by the matrix element of the projector
\be \label{eqn_scalar_proj}
\left\bra S,T  \right.\left| S',T'  \right \ket = ( S,T | \Pi_{j_{i}} |  S',T' ).
\ee
From the results of the previous section this projector can be explicitly constructed.
\begin{lemma}\label{product2}
The projector onto the kernel of $H$ is explicitly given by
\be \label{eqn_proj_kern}
\Pi_{j_{i}} = 1 +\sum_{N=1}^{\mathrm{min}(2j_{i})}\frac{(-1)^{N}}{N!} \frac{(J-N+1)!(J-2N+1)!}{(J+1)!^{2}} \, \hat{R}^{N}(\hat{R}^{\dagger})^{N}.
\ee
\end{lemma}
The proof is contained in Appendix \ref{proj_kern_proof}.

\subsection{Overlap with the Orthogonal Basis}

Let us now show how theorem \ref{thm_sum_T} can be used to generate the various other scalar products.  To do so we need the following Lemma:
\begin{lemma} \label{Lemma_sum_T}
\be \label{eqn_R_delta}
  \sum_{T} R^{(s,t)}_{(S,T)}(N) = \delta_{s,S}. 
\ee
\end{lemma}
The proof is given in Appendix \ref{Lemma_sum_T_proof}.  This identity translates into the statement that $\hat{R}^{\dagger}$ acts diagonally on $|S)_{j_{i}}=\sum_{T}|S,T)_{j_{i}}$:
\be
(R^{\dagger})^{N} |S)_{j_{i}} = \frac{(J-2N +1)!}{(J+1)!} |S)_{j_{i}-N/2}.
\ee
Therefore summing over $T$ in (\ref{scalar}) yields
\bea \label{eqn_s_ST_overlap}
  \left\bra S  \right.\left|S',T' \right \ket &=&
  \sum_{N=0}^{\mathrm{min}(2j_{i})}  \alpha_{J,N} \sum_{t} R^{(S,t)}_{(S',T')}(N) ||S,t||_{j_{i}-N/2}^{2}
\eea
where it is convenient to define
\bea \label{eqn_alpha}
 \alpha_{J,N} &\equiv & 
 \frac{(-1)^N}{N!} \frac{(J-N+1)!}{(J-2N+1)!}.
\eea
By summing over the different labels in (\ref{scalar}) we can compute the remaining scalar products:
\bea \label{eqn_usual_sc_prods}
  \left\bra S\right|\left.S' \right\ket = \delta_{S,S'}\sum_{T,N} \alpha_{J,N} ||S,T||_{j_{i}-N/2}^{2}, \hspace{12pt}
   \left\bra T\right|\left.T' \right\ket = \delta_{T,T'} \sum_{S,N} \alpha_{J,N} ||S,T||_{j_{i}-N/2}^{2},\nonumber \\ \hspace{12pt}  \bra S | T \ket_{j_{i}} = (-1)^{k_{23}} \sum_N \alpha_{J,N} ||S,T||_{j_{i}-N/2}^{2} 
   =  \sum_N \alpha_{J,N} (S|T)_{j_{i}-N/2}
\eea
where here the sum over $N$ starts from zero.  Hence all of the scalar products between $|S\ket$, $|T\ket$, and $|S,T\ket$ bases are different summations over $\alpha_{J,N}$ and the canonical norms.  
%
We now show how to perform the summations in (\ref{eqn_usual_sc_prods}) to give the well known normalization factors for $|S\ket$ and $|T\ket$.  
\begin{proposition} \label{ST_orthog}
\be
  \left\bra S\right|\left.S' \right\ket  = \frac{\delta_{S,S'}}{2S+1}  \Delta^{2}(j_{1}j_{2}S) \Delta^{2}(j_{3}j_{4}S), \hspace{12pt} \left\bra T\right|\left.T' \right\ket  = \frac{\delta_{T,T'}}{2T+1}  \Delta^{2}(j_{1}j_{3}T) \Delta^{2}(j_{2}j_{4}T),
\ee
where the triangle coefficients were given in (\ref{eqn_tri_coeff}).
\end{proposition}
The proof is given in Appendix \ref{ST_orthog_proof}

Finally, it is easy to see that the overlap between the $|S\ket$ and $|T\ket$ bases is given by a 6j symbol.  That is, the third sum in (\ref{eqn_usual_sc_prods}) can be recognized as the Racah expansion of the 6j symbol by making a change of variable $m = J-N$.  Doing so we get
\bea
  \bra S | T \ket 
  &=& {(-1)^{J+k_{23}}}{\Delta(j_{1}j_{2}S)\Delta(j_{3}j_{4}S)\Delta(j_{1}j_{3}T)\Delta(j_{2}j_{4}T)} \sixj{j_1}{j_2}{S}{j_4}{j_3}{T}.
\eea
These relations with the orthogonal basis provide a consistency check, but they will also be useful later in connecting with the 15j symbol.  To do so we will next study the contraction of the $|S,T\ket$ states.

\subsection{The 20j symbol}

Let us now use all of the results obtained for the $|S,T\ket$ basis to compute a generalization of the $15j$ symbol, which will depend now on 20 spins: ten $j_e$ on the edges and five $S_v$, and five $T_v$ on the vertices.  The $20j$ symbol will be the amplitude corresponding to the coherent 4-simplex.  It is a generalization of the $15j$ symbol since by theorem \ref{thm_sum_T} we can sum over five of the extra spins and obtain one of the five variations of the $15j$ symbol.

First label the vertices of the 4-simplex by $i=1,..,5$ and an edge directed from $i$ to $j$ as in $z_{e} \equiv z^{i}_{j} \neq z^{j}_{i} \equiv z_{e^{-1}}$. Let $\{20j\}_{S_i,T_i} \equiv \underset{i}{\corner} |S_i, T_i \ket$ denote the unnormalized 20j symbol.
The relation (\ref{eqn_discrete_amp}) between coherent and discrete amplitudes reads
\be 
  A_{4S}(j_{ij},z^{i}_{j}) = \sum_{S_i,T_i} \{20j\}_{S_i,T_i} \prod_{i} \frac{( z^{i}_{j}  |  S_i, T_i\ket}{\|S_i,T_i\|^2}
\ee
\begin{figure} 
  \centering
    \includegraphics[width=0.5\textwidth]{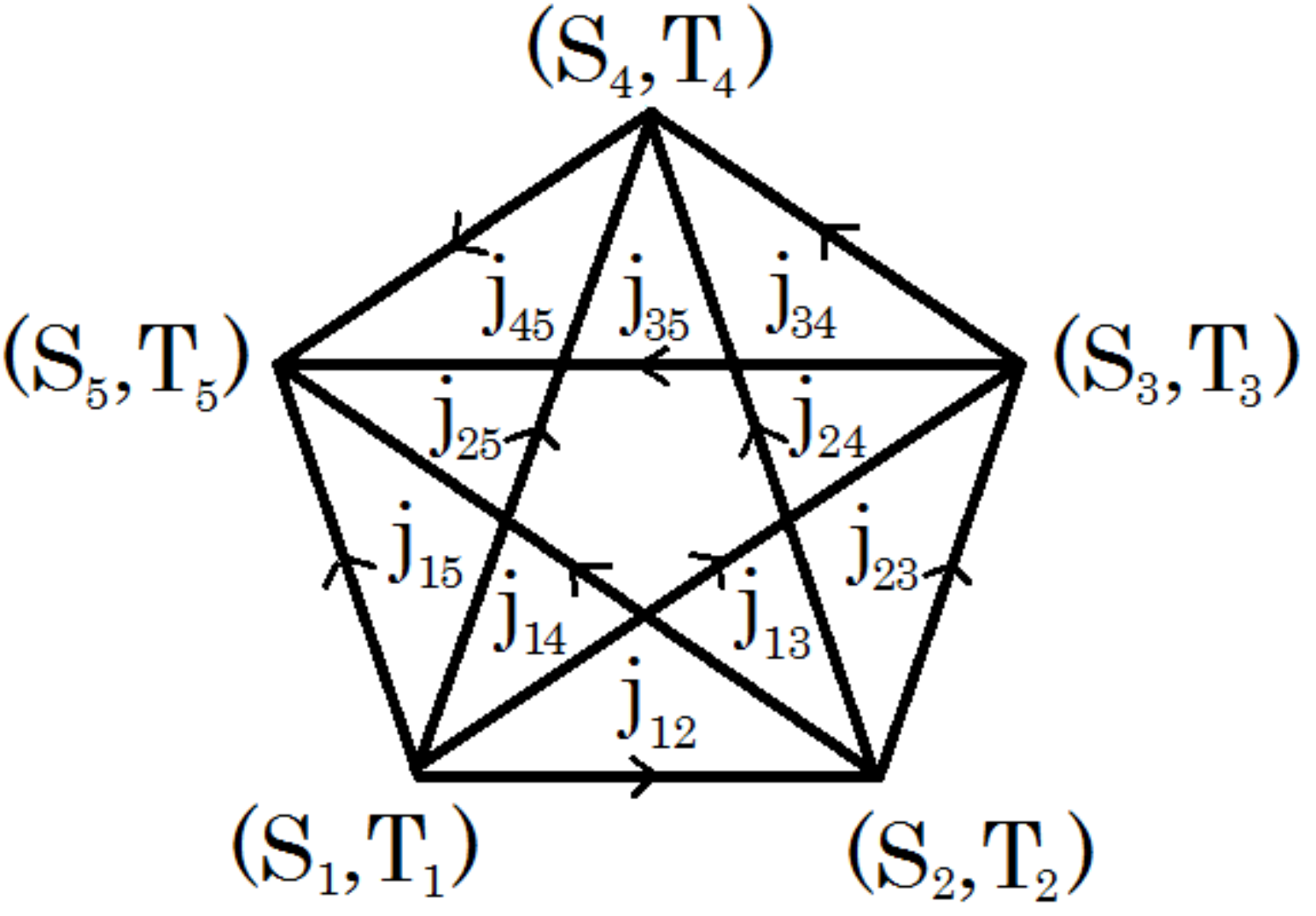}
    \caption{The graphical representation of the 20j symbol: The amplitude of the coherent 4-simplex.}  \label{fig_20j}
\end{figure}
We can express the $20j$ symbol explicitly in terms of the $15j$ symbol by inserting another five resolutions of identity $\one_{j_i} = \sum_{S'} |S'\ket\bra S'|/\|S'\|^2$ into the definition of the 20j symbol to get
\be \label{eqn_20j_15j}
  \{20j\}_{S_i,T_i} = \sum_{S_{i}^{'}} \{15j\}_{S'_i} \prod_{i} \frac{\bra S_{i}^{'} | S_i, T_i \ket}{\|S_{i}^{'}\|^2}
\ee
where $\{15j\}_{S'_i}$ is the unnormalized 15j symbol defined by $\{15j\}_{S'_i} \equiv \underset{i}{\corner} |S_{i}^{'} \ket$ and is equal to $\prod_i \|S_{i}^{'}\|$ times the conventional normalized 15j symbol up to a sign depending on the orientation of the edges.  Notice that by summing over $T_i$ in (\ref{eqn_20j_15j}) we obtain the unnormalized $15j$ symbol as expected.  Thus the five different kinds of $15j$ symbols are derived from the $20j$ by summing over the different channels.  For example the $15j$ with all $S$ channels is given by
\be
  \{15j\}_{S_i} = \sum_{T_i} \{20j\}_{S_i,T_i}, 
\ee
and the other kinds of 15j symbol are given similarly.

Let us now use theorem \ref{thm_sum_T} to rewrite the 20j symbol in a more symmetric form
\be
  \{20j\}_{S_i,T_i} = \sum_{S_{i}^{'},T_{i}^{'}} \{15j\}_{S'_i} \prod_{i} \frac{\bra S_{i}^{'},T_{i}^{'} | S_i, T_i \ket}{\|S_{i}^{'}\|^2}.
\ee
In this form it is easy to derive the asymptotics of the 20j symbol by those of the 15j since for large spins (see the next section) $\bra S,T | S', T' \ket \sim \delta_{S,S'} \delta_{T,T'} \|S,T\|^2$, and therefore 
\be
  \{20j\}_{S_i,T_i}\sim  \{15j\}_{S_i} \prod_i \frac{\|S_i,T_i\|^2}{\|S_i\|^2}.
\ee
This means that understanding the asymptotics of the $20j$ symbol will give us the asymptotics of the $15j$ symbol too.
There has been recent results on the asymptotics of spin networks evaluation \cite{Conrady:2008mk,Barrett:2009gg} but this progress concerns however the asymptotic evaluation of the coherent state amplitude $ A_{4S}(j_{ij},z^{i}_{j})$. 
The asymptotic evaluation of the non coherent $15j$ symbol is not known and as we are going to see in the next section our techniques allow us to unravel the asymptotics for the first time.

\chapter{Semi-classical Limit}
\label{chapter_semi}

It is now well-known and explained in great detail in \cite{Conrady:2009px,Freidel:2009nu} that the space of $4$-valent intertwiners can be uniquely labeled by 
oriented tetrahedra. In this section we will demonstrate this correspondence for the $|S,T\ket$ states.  In order to connect with the classical behaviour we would like to analyze the asymptotics of the scalar product of two such states in the limit where the spins $(j_{i},S,T)$ are all uniformly large.
We use the fact that this scalar product can itself be expressed as an integral 
\be
\la S,T|S',T'\ra =\frac1{\prod_{i<j}( k_{ij}! k_{ij}'!)} \int \prod_{i} \frac{\rd^2 z_{i}}{\pi^2}  e^{- S_{k}(z)}
\ee
  where the action is given by
  \be \label{eqn_action}
  S_{k} = \sum_{i} \la z_{i}| z_{i} \ra - \sum_{i<j}\left( k_{ij} \ln[z_{i}|z_{j}\ra + k'_{ij} \ln \la z_{i}|z_{j}] \right).
  \ee
The asymptotic evaluation of this scalar product is controlled by the stationary points\footnote{ If $k_{ij}= N K_{ij}$ and we  define $ z_{i}=\sqrt{N} x_{i}$ we see that this integral that we want to evaluate in the large $N$ limit takes the usual form $N^{2J} \int \prod_{i}\rd x_{i} e^{-N S_{K}(x)} $ .} of this action.
That is we look for solutions of 
\be
\sum_{j\neq i} \frac{k_{ij}}{[z_{i}|z_{j}\ra} [z_{i}| = \la z_{j} |,\qquad
\sum_{j\neq i} \frac{k_{ij}'}{\bra z_{j}|z_{i}]} |z_{i}] = | z_{j} \ra .\label{kz}
\ee
Now it is clear that if $k\neq k'$ there cannot be any real solution. This shows that this scalar product is exponentially suppressed unless $(S,T)=(S',T')$
\footnote{We could still evaluate the  integral asymptotically  when  $k\neq k'$ by looking for complex solutions.
In order to do so we use the fact that $[\check{z}_{i}|\check{z}_{j}\ra = [z_{i}|z_{j}\ra^{*}$.
We get an action holomorphic in $ |z_{i}\ket$ and $|\check{z}_{i}\ket$:
$$S_{k} =- \sum_{i} [ \check{z}_{i}| z_{i} \ra - \sum_{i<j}\left( k_{ij} \ln[z_{i}|z_{j}\ra + k'_{ij} \ln [\check{ z}_{i}|\check{z}_{j}\ra \right).$$
 The stationary equations are 
\be
\sum_{j\neq i} \frac{k_{ij}}{[z_{i}|z_{j}\ra} [z_{j}| = -[ \check{z}_{i} |,\qquad 
\sum_{j\neq i} \frac{k_{ij}'}{[\check{z}_{i}|\check{z}_{j}\ra} [\check{z}_{j}| =  [ z_{i} |,
\ee
In the case $k_{ij}\neq k'_{ij}$ we do not demand that $[{z}_{i}| = \la \check{z}_{i}|$ which corresponds to the real contour of integration.}.
Furthermore, if we contract this   equation with $ | z_{j}\ra$  we obtain the constraints
\be
2j_{i} = \sum_{j\neq i} k_{ij}= \la z_{i}|{z}_{i}\ra.
\ee
These equations are invariant under $\SU(2)$,  so $  g|z_{i}\ra$ is a solution if $|z_{i}\ra$ is and $ g \in \SU(2)$.
We also have an invariance of these equations under the rescaling $ |z_{i}\ra \to e^{i\alpha^{i}} |z_{i}\ra$.

Finally, by taking the conjugation of (\ref{kz}) $|\check z\ket = |z]$ and using the fact that $[\check{z}_{j}|\check{z}_{i}\ra = [z_{j}|z_{i}\ra^{*}$ we can show that this equation is also equivalent
to the conjugated equation 
\be
\sum_{j\neq i} \frac{k_{ij}}{[\check{z}_{j}|\check{z}_{i}\ra} [\check{z}_{j}| = \la \check{z}_{i} |.\label{conjkz}
\ee
This means that the $\mathbb{Z}_{2}$ transformation  $ |z_{i}\ra \to | \check{z}_{i}\ra =|z_{i} ]$ is also a symmetry of the equation of motion.  In summary this shows that the symmetry group of the solutions (\ref{kz}) is given by $\SU(2) \times \mathrm{U}(1)^{4}\times \mathbb{Z}_{2}$.

\section{Relation with Framed Tetrahedra}

What is remarkable about the solutions (\ref{kz}) is that they are in one to one correspondence with framed tetrahedra.
A framed tetrahedron in $\R^{3}$ is a tetrahedron together with a choice of frame on each face (i.e. a choice of a preferred direction tangential to the face).
The SU(2) invariance corresponds to rotations of the tetrahedron, while a rotation of the frame on face $i$ by an angle $\alpha^{i}$ corresponds to a rescaling of $|z_{i}\ra $ by $ e^{i\alpha^{i}/2}$.
The $\mathbb{Z}_{2}$ transformation corresponds to a global reflection exchanging inward and outward normals.

Indeed, suppose that we have a framed tetrahedron which is  such that the area and outward unit normal directions of the face $i$ are  denoted by $(A_{i},N_{i})$.
We also denote $F_{i}$ to be the unit vector in the face $i$ (i.e. $F_{i}\cdot N_{i}=0$) that provides the framing of the face $i$.
Then the fact that this data corresponds to a tetrahedron is implied by the closure constraints
\be
\sum_{i} A_{i} N_{i} =0.
\ee

Such a framed tetrahedron can be equivalently labeled in terms of four spinors $|z_{i}\ra$ which satisfy the closure relation 
\be
\sum_{i} |z_{i}\ket \bra z_{i}| = \frac{A}2 1
\ee
where $A=\sum_{i}A$ is  the total area of the  tetrahedra.
This data is 
related to the data $(A_{i},N_{i},F_{i})$ as follows: First $\bra z_{i}|z_{i}\ket = A_{i}$ and second
\be \label{eqn_z_N_F}
|z_{i}\ra\la z_{i}| -|z_{i}][z_{i}| = A_{i} N_{i} \cdot \sigma ,\qquad  |z_{i}\ra[z_{i}| = i\frac{A_{i}}{2}\left(F_{i} + i N_{i}\times F_{i}\right) \cdot \sigma
\ee
where $\sigma = (\sigma_{1},\sigma_{2},\sigma_{3})$ are the Pauli matrices and $\times$ denotes the cross product.
The first equation determines $|z_{i}\ra$ up to a phase while the second equation determines the phase up to an overall sign. Thus $\pm |z_{i}\ra$ is uniquely determined by the framed tetrahedron.

\begin{theorem} \label{thm_closure}
The solutions of (\ref{kz}) are in one to one correspondence with framed tetrahedra, with face areas $2j_{i}$ and total area $2J$.
 The discrete parameters related to the spinors are given by
\be \label{eqn_k_on_shell}
\bra z_{i}|z_{i}\ket = 2j_{i}, \quad J k_{ij} \equiv   |[z_{j}|z_{i}\ket|^{2}.
\ee
%
\end{theorem}
\begin{proof}
Lets suppose that $|z_{i}\ket$ is a solution of (\ref{kz}).  Then
\bea
\sum_{i} |z_{i}\ra \la z_{i} | = 
\sum_i \sum_{j\neq i } \frac{k_{ij}}{[z_{j}|z_{i}\ra} |z_{i}\ra [z_{j}| 
= \sum_{i < j } \frac{k_{ij}}{[z_{j}|z_{i}\ra} (|z_{i}\ra [z_{j}|  -|z_{j}\ra [z_{i}|)
= \sum_{i < j } {k_{ij}} 1 = J 1.
\eea
which implies that $|z_{i}\ket$ satisfy the closure constraint and hence constitute a framed tetrahedron.
The area of the faces of this tetrahedron are given by $A_{i} =\sum_{j\neq i} k_{ij}$.

Let us now suppose that $|z_{i}\ket$ is a solution of the closure constraints and lets define $k_{ij} \equiv \frac2{A} [z_{j}|z_{i}\ket\bra z_{i}|z_{j}]$.
Then by construction we have 
\be
\sum_{j\neq i} \frac{k_{ij} }{[z_{j}|z_{i}\ket} [z_{j}| = \frac2{A} \sum_{j\neq i}\bra z_{i}|z_{j}] [z_{j}| = \bra z_{i}|.
\ee
which shows that $|z_{i}\ket$ is a solution of (\ref{kz}). 

Finally let us suppose that given $|z_{i}\ket$ we have another set $k'_{ij}$ which is a solution of (\ref{kz}).
 This would imply that 
$\sum_{j\neq i} {\Delta_{ij}} [z_{j}| =0$ with $\Delta_{ij}\equiv(k_{ij}-k_{ij}')/[z_{j}|z_{i}\ket$.
The sum contains three terms, and by contracting it with $|z_{k}]$, with $k\neq i,j$, we obtain  two relations.
The consistency of these relations implies that $\Delta_{ij}=0$. 
\end{proof}

This shows that the $|S,T\ket$ states enjoy the same geometrical properties as the coherent intertwiners.  Namely they are peaked on states representing closed bounded tetrahedra.  We note that the proof of theorem \ref{thm_closure} also holds for general $n$-valent intertwiners by simply extending the range of indices from 4 to $n$.  A similar analysis of stationary points of intertwiner generating functionals is given in \cite{Bonzom:2012bn}.

\section{Geometrical Interpretation}
It is interesting to make explicit the geometrical interpretation of the data encoded in the spinor variables.
We have seen that the norms of the spinors $2j_{i}=\bra z_{i}|z_{i}\ket$ are the areas of the faces. 

This can be made explicit by writing these spinors in terms of the geometrical data:
As we have seen $A_{i}$ denotes the area of the face $i$ and 
we denote by $\theta_{ij} \in [0,\pi]$ the dihedral angle between the normals $N_{i}$ and $N_{j}$.
The extra data necessary is the angle $\alpha^{i}_{j}$ in the face $i$ between   the oriented edge $(ij)$ and the reference vector $F_{i}$.
This data is represented in figure \ref{fig_tri} and is related to the spinors in the following lemma.

\begin{figure} 
  \centering
    \includegraphics[width=0.7\textwidth]{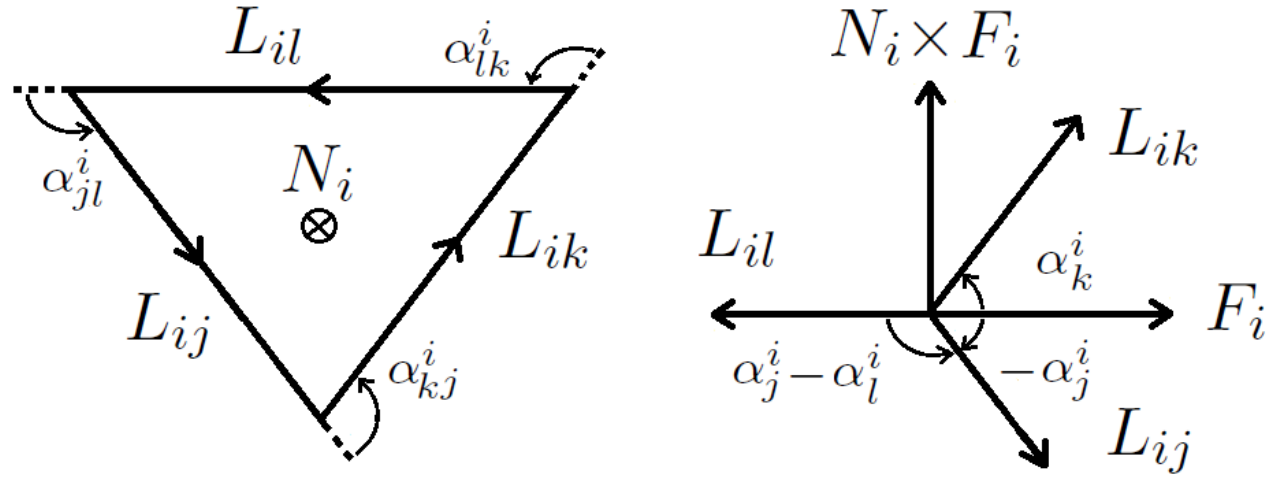}
    \caption{The geometrical data on the face $i$ of a framed tetrahedron.}  \label{fig_tri}
\end{figure}

\begin{lemma} \label{lemma_geo}
The angle between the edge $(ij)$ and the edge $(ik)$ is $\alpha^{i}_{jk}$ where
$\alpha^{i}_{jk} = \alpha^{i}_{j}-\alpha^{i}_{k}$.
The expression of the spinor products in terms of the geometrical data is  given by:
\bea
[z_{i}|z_{j}\ra &=& \epsilon_{ij} \sqrt{A_{i}A_{j}} \sin \frac{\theta_{ij}}{2} e^{i (\alpha^{i}_{j} + \alpha^{j}_{i})/2},\\
{[}z_{i}|z_{j}] &=& \sqrt{A_{i}A_{j}} \cos \frac{\theta_{ij}}{2} e^{i (\alpha^{i}_{j} - \alpha^{j}_{i})/2},
\eea
where $\epsilon_{ij} = +1$ if $(ij)$ is positively oriented and $-1$ otherwise.  The area of face $i$ is $A_{i} = 2j_{i}$ and
$\theta_{ij}\in[0,\pi]$ is the 3d external dihedral angle for which we choose the convention $\theta_{ii}=0$.  Finally, $\alpha^{i}_{j}$ is the angle in the face $i$ between the edge $(ij)$ and the reference vector $F_{i}$.
\end{lemma}
\begin{proof}
From the definitions (\ref{eqn_z_N_F}) we have
\be
|z_{i}\ket\bra z_{i}|=  \frac{A_{i}}{2}(1 + N_{i}\cdot \sigma ),\qquad 
|z_{i}][ z_{i}|= \frac{A_{i}}{2}(1 - N_{i}\cdot \sigma ).
\ee
and so the scalar product between two normals is
\be \label{eqn_N_dot}
  A_i A_j N_i \cdot N_j = |\bra z_i | z_j \ket|^2 - |\bra z_i | z_j]|^2 = A_i A_j \cos \theta_{ij}.
\ee
Defining the edge vectors by $L_{ij}\equiv A_{i}A_{j}(N_{i}\times N_{j})$ then gives
\be\label{Lij}
 L_{ij} \cdot \sigma = 2i \Big( |z_{i}\ra\la z_{i} | z_{j}][z_{j}| - |z_{j}][z_{j}|z_{i}\ra \la z_{i}| \Big).
\ee
Now by taking the trace of the square of (\ref{Lij}) we obtain 
\be \label{eqn_N_cross}
  A_i A_j |N_i \times N_j | =  2|[z_{i}|z_{j}\ket \bra z_j | z_i \ket |=  A_i A_j \sin \theta_{ij}.
\ee
Equations (\ref{eqn_N_dot}) and (\ref{eqn_N_cross}) determine the magnitudes of the spinor products.  

Lets now look at the scalar product between the edge vectors $L_{ij}$ 
and the complex vector $F_{i} + i N_{i}\times F_{i}$
\bea
 L_{ij}\cdot (F_{i} + i N_{i}\times F_{i})& =&
\frac{2}{A_{i}}\tr\left( \left\{|z_{i}\ra\la z_{i} | z_{j}][z_{j}| - |z_{j}][z_{j}|z_{i}\ra \la z_{i}| \right\} 
 |z_{i}\ra [ z_{i}| \right)  \label{eqn_L_F_N} \\
&= & \frac{2 }{A_{i}}  [ z_{i}|z_{j}][  z_{i} | z_{j}\ra \la z_{i}|z_{i}\ra   \nn \\
&=& \epsilon_{ij} A_{i}A_{j}\sin {\theta_{ij}} e^{i\alpha^{i}_{j} } = \epsilon_{ij} |L_{ij}| e^{i\alpha^{i}_{j}} \nn
\eea
where $|L_{ij}| = A_i A_j \sin \theta_{ij}$.  This shows that $L_{ij} \cdot F_i = |L_{ij}| \cos(\alpha^{i}_{j})$ and so $\alpha^{i}_{j}$ is indeed the angle  between the edge $(ij)$ and the frame vector on face $i$.  The sign of $L_{ij}\cdot (N_{i}\times F_{i})$ determines the orientation of this angle with respect to the 2d basis $\{F_i,\hat{F}_{i}\}$ where $\hat{F}_i \equiv N_{i}\times F_{i}$.

We can also show that the angle $\alpha^{i}_{jk} \equiv \alpha^{i}_{j} - \alpha^{i}_{k}$ is the angle between the edge vectors $L_{ij}$ and $L_{ik}$ at the vertex $i$.
Using (\ref{eqn_L_F_N}) we can construct the following quantity
\be
  \epsilon_{ij} \epsilon_{ki} |L_{ij}| \cdot |L_{ik}| e^{i(\alpha^{i}_{j} - \alpha^{i}_{k})} = [L_{ij} \cdot (F_{i} + i N_{i} \times F_{i})] \cdot [L_{ik} \cdot (F_{i} - i N_{i} \times F_{i})].
\ee
Now in the components of the 2d basis $L^{(1)}_{ij} = L_{ij} \cdot F_{i}$ and $L^{(2)}_{ij} = L_{ij} \cdot \hat{F}_{i}$ we have
\be
  \epsilon_{ij} \epsilon_{ki} |L_{ij}| \cdot |L_{ik}| e^{i(\alpha^{i}_{j} - \alpha^{i}_{k})} 
  = (L^{(1)}_{ij} + i L^{(2)}_{ij})(L^{(1)}_{ik} - i L^{(2)}_{ik}).
\ee
The real part is equal to $L_{ij} \cdot L_{ik}$.  
Therefore
\be \label{eqn_alpha_i_jk}
  L_{ij} \cdot L_{ik} = \epsilon_{ij} \epsilon_{ki} |L_{ij}| \cdot |L_{ik}| \cos \alpha^{i}_{jk}.
\ee
which shows that $\alpha^{i}_{jk}$ is the angle between edges $(ij)$ and $(ik)$.
\end{proof}

Using (\ref{eqn_alpha_i_jk}) and the definition $L_{ij} = A_i A_j (N_i \times N_j)$ we can relate the angles $\alpha^{i}_{jk}$ to the 3d dihedral angles by
\be \label{eqn_alpha_theta}
  \epsilon_{ij} \epsilon_{ki} \cos \alpha^{i}_{jk} = \frac{\cos \theta_{jk} - \cos \theta_{ij} \cos \theta_{ik}}{ \sin \theta_{ij} \sin \theta_{ik}}.
\ee
This is the spherical law of cosines relating the edges $(\theta_{ij},\theta_{jk},\theta_{ki})$ and angles 
$(\alpha_{ji}^{k},\alpha_{kj}^{i},\alpha_{ik}^{j})$ of a spherical triangle.  This relation with spherical geometry is captured by the so called three terms relations which we discuss next.

\subsection{Geometry of 3-terms Relations}

The relationships between the the 3d dihedral angles and the internal angles between edges is expressed via the three term relations 
(Fierz identity) satisfied by a set of spinors.
There are two such types of relations.
First there are the relations arising at a given vertex of the tetrahedra which imply that the angles $(\theta_{ij},\theta_{jk},\theta_{ki})$ and 
$(\alpha_{ji}^{k},\alpha_{kj}^{i},\alpha_{ik}^{j})$ are respectively the edge lengths and angles of a spherical triangle.
We can write two such relations\footnote{Note that $|z_{j}\ra \la z_{j} | + |z_{j}][ z_{j} | = \la z_{j} |z_{j}\ra \one$.}, 
\bea
\la z_{i}|z_{j}\ra \la z_{j} |z_{k}\ra + \la z_{i}|z_{j}][ z_{j} |z_{k}\ra &=& \la z_{i}|z_{k}\ra \la z_{j} |z_{j}\ra \\
\la z_{i}|z_{j}\ra \la z_{j} |z_{k}] + \la z_{i}|z_{j}][ z_{j} |z_{k} ] &=& \la z_{i}|z_{k}] \la z_{j} |z_{j}\ra 
\eea
which translate into
\bea
c_{ij} c_{jk} \label{eqn_3_term_1}
+ s_{ij} s_{jk} e^{i\alpha^{j}_{ki} }
&=& c_{ik} e^{i(\alpha^{i}_{jk}+\alpha^{j}_{ik} + \alpha^{k}_{ij})/2},\\
c_{ij} s_{jk} 
+  s_{ij} c_{jk} e^{i\alpha^{j}_{ki} }
&=&  s_{ik} e^{i( \alpha^{i}_{jk}-\alpha^{j}_{ik} - \alpha^{k}_{ij})/2}, \label{eqn_3_term_2}
\eea
where $c_{ij} \equiv \cos  \frac{\theta_{ij}}2$ and $ s_{ij}\equiv \epsilon_{ij}\sin  \frac{\theta_{ij}}2$.  Taking the difference of squares of equations $|(\ref{eqn_3_term_1})|^2 - |(\ref{eqn_3_term_2})|^2$ produces equation (\ref{eqn_alpha_theta}).

We also have a relation that  genuinely depends on the tetrahedral geometry and involves the four spinors;
it follows from the Pl\"ucker relation that we have already made extensive use of
\be
[z_{1}|z_{2}\ket[z_{3}|z_{4}\ket + [z_{1}|z_{3}\ket[z_{4}|z_{2}\ket + [z_{1}|z_{4}\ket[z_{2}|z_{3}\ket =0
\ee
and it reads 
\be
s_{12}s_{34} e^{i (\alpha_{12} +\alpha_{34})/2} + s_{13}s_{24} e^{i (\alpha_{13}+\alpha_{24})/2} + s_{14}s_{23}e^{i(\alpha_{14}+\alpha_{23})/2} =0
\ee
where we have defined 
\be \label{eqn_alpha_geo}
\alpha_{ij}\equiv \frac12 \sum_{k\neq i,j} (\alpha^{i}_{jk} +\alpha^{j}_{ik}).
\ee
It can be checked that these angles sum up to $0$: $\sum_{i\neq j } \alpha_{ij}=0$.  Now what needs to be appreciated is the non trivial fact that the angles 
$$
\Phi_{S} = \alpha_{12} +\alpha_{34}
,\quad
\Phi_{T} =\alpha_{13}+\alpha_{24},\quad
\Phi_{U} = \alpha_{14}+\alpha_{23}
$$
determine completely the geometry of the tetrahedron once we know the face areas $A_{i}=2j_{i}$.
This means that $(j_{i}, \alpha_{S},\alpha_{T})$ determine the value of all the 3d dihedral angles $\theta_{ij}$ and internal angles $\alpha_{ij}^{k}$.
This non-trivial fact follows from the analysis performed in \cite{Freidel:2009nu}.

\section{Asymptotic Evaluation of the 20j Symbol}

We will now take an indepth look at the asymptotic evaluation of the normalized  $20j$ symbol.
This object depends on the choice of an orientation of the edges, and we denote by $\epsilon_{ij}$ a sign which $+1$ if the edge $[ij]$ is positively oriented from $i$ to $j$ and $-1$ otherwise. 
This normalized $20j$ symbol is defined as a contraction of the normalised intertwiner $ |S,T\ra$ times the normalisations  and it is 
 expressed as 
\be \label{normalized_20j}
\widehat{\{20j\}}_{S_a,T_a}\equiv \frac{\{20j\}_{S_a,T_a}}{\prod_a \|S_a,T_a\|} = \frac{I(k_{ij})}{\sqrt{\prod_{a}(J_{a}+1)! \prod_{a\neq i<j}k^{a}_{ij}! }} ,
\ee
where $I(k_{ij})$ is an integral over $20$ spinors $|z_{i}^{j}\ra $. The contour of integration is a real contour where $ |z_{i}^{j}\ra$ is related to the  conjugate $|z_{j}^{i}]$ by the reality condition.
\be \label{eqn_reality_cond}
|z_{i}^{j}\ra = \epsilon_{ij}|z_{j}^{i} ].
\ee
This condition implies that the normals of glued faces are related by $N^{j}_{i} = - N^{i}_{j}$ and that the frame vectors match $F^{i}_{j} = F^{j}_{i}$.  

The integral  is given by
\be
I(k_{ij}) =\int \prod_{i\neq j} \frac{\rd^{2} z_{j}^{i} }{\pi^2}  e^{S(z^{i}_{j}) },\quad \mathrm{with}\quad S \equiv \sum_{i<j}  [ z^{i}_{j} | z_{i}^{j}\ket + \sum_{a} \sum_{i<j} k_{ij}^{a} \ln [z_{i}^{a}|z_{j}^{a}\ra
\ee
There are four spinors $|z_{i}\ra$ associated with a framed tetrahedron.
The stationary points of this equation are given by solutions of 
\be \label{system_of_equations}
\sum_{j\neq a,i} \frac{ k_{ij}^{a} }{[z_{j}^{a}|z_{i}^{a}\ra} [z_{j}^{a}| = - [z^{i}_{a}| = \epsilon_{ai} \la z_{i}^{a}|,
\ee
and according to the previous section the solution of these equations are given by
oriented framed tetrahedra.
The relationship between $k_{ij}^{a}$ and the spinors depends on the choice of graph orientation,
it is given by
\be \label{eqn_k_z}
k_{ij}^{a} = \frac1{J^{a}} |[\hat{z}_{i}^{a}|\hat{z}_{j}^{a}\ra|^{2}, \quad 2j_{ai}=  \la \hat{z}^{a}_{i}|\hat{z}^{a}_{i}\ra,\quad J^{a} = \sum_{i\neq a} j^{a}_{i} 
\ee
where $|\hat{z}^{a}_{i}\ra =|z^{a}_{i}\ra$ if the edge is oriented from $a$ to $i$ and 
$|\hat{z}^{a}_{i}\ra = |z^{a}_{i}]$ if the edge is oriented from $i$ to $a$.
This determines the norm of the spinor scalar products in terms of $k_{ij}^{a}$. 

The phases of these products are denoted $\alpha^{ab}_{i}$ and they denote the angle in the face $b$ of the tetrahedron $a$, between the edge $(bi)$ and the reference frame vector in the face $b$ of tetrahedra $a$.
As shown in lemma \ref{lemma_geo}, they are related to the spinor products by 
\be
[\hat{z}_{i}^{a}|\hat{z}_{j}^{a}\ra =\sqrt{J^{a} k_{ij}^{a}} e^{i (\alpha_{j}^{ai}+ \alpha_{i}^{aj})/2}.
\ee

Thus the on-shell evaluation of the action is
\be
  S_{\mathrm{onshell}} = -\sum_{a} J^{a} +\frac12 \sum_{a\neq i<j} k_{ij}^{a}\ln (J^{a} k_{ij}^{a} ) + \frac{i}{2}\sum_{a\neq i<j} k_{ij}^{a}  (\alpha_{i}^{aj} + \alpha_{j}^{ai})
\ee
The real part can be rewritten as
  \be
\mathrm{Re}(S_{\mathrm{onshell}}) = \frac12\left(\sum_{a} J_{a}\ln J_{a}-J^{a} + \sum_{a\neq i<j} ( k_{ij}^{a}\ln  k_{ij}^{a}    - k_{ij}^{a}) \right)
\ee
which is easily recognized as the dominant\footnote{up to a term given by $  \frac14\ln(\prod_{a}( 2\pi J_{a}^{3}\prod_{i<j} (2\pi  k_{ij}^{a} ) )).$ } term in the Stirling  expansion of  
\be
\ln\sqrt{\prod_{a}(J_a+1)! \prod_{a\neq i<j}k_{ij}^{a}!}
\ee
This cancels the factor in (\ref{normalized_20j}).  Let us now focus on the imaginary part:
\be\label{Ims}
{\mathrm{Im}}(S_{\mathrm{onshell}}) =  \frac12  \sum_{a\neq i<j} k_{ij}^{a} (\alpha_{i}^{aj} + \alpha_{j}^{ai}).
\ee
  First,  recall that the system of equations (\ref{system_of_equations}) possesses a gauge symmetry,
\be
\alpha^{ai}_{j} \to \alpha^{ai}_{j} +\theta^{ai}
\ee
where $\theta^{ai}= -\theta^{ia}$.
This  corresponds to the rotation of the  frame vector in the face $(ai)$ by an angle $\theta^{ai}$.
The action is invariant under these gauge transformations.  Indeed under $|z_{i}^{a}\ra \to e^{i\theta^{ai}} |z_{i}^{a}\ra $ the variation of the on-shell action is 
\bea
2\Delta S_{\mathrm{onshell}} &=& i \sum_{a\neq i<j} k_{ij}^{a} (\theta^{ai} + \theta^{aj}) =
i\sum_{(a,i,j)} k_{ij}^{a} \theta^{ai} = i\sum_{(a,i)} \left(\sum_{j\neq (a,i)}k_{ij}^{a}\right) \theta^{ai} \nn \\
&=& 2i \sum_{(a,i)} j_{ai} \theta^{ai} = i\sum_{(a,i)} j_{ai} ( \theta^{ai}+\theta^{ia}) =0
\eea
Here we have denoted by $(a,i,j)$ or $(a,i)$ a set of indices all distinct from each other. 
Therefore the on-shell action can be determined entirely in terms of gauge invariant angles.  The question is which combinations appear.

There are two types of gauge invariant data:
The first type characterizes the intrinsic geometry of each tetrahedron and depends only on the data associated with one tetrahedron.
These correspond  to the angles in a given tetrahedron $a$ between edges $(ij)$ and $(ik)$, and are given by
\be
\alpha^{ai}_{jk} \equiv \alpha^{ai}_{j}-\alpha^{ai}_{k}.
\ee
We already have seen in (\ref{eqn_alpha_i_jk}) that these angles are angles between the edges $(ij)$ and $(ik)$ at the tetrahedron $a$.

The second type of gauge invariant angles encode the extrinsic geometry of the gluing of the five tetrahedra. It depends on two tetrahedra and  involves the sum\footnote{Since the faces $(ab)$ and $(ba)$ have opposite orientations this is really a differences of angles when we take the orientation into account.} of angles between two tetrahedra
\be
\xi^{ab}_{i} \equiv \alpha^{ab}_{i} +\alpha^{ba}_{i}.
\ee

In order to understand the geometrical meaning of these angles let us first remark that when the shapes of the triangles $(ab)$ and $(ba)$ match then the angles between the edges of the triangles when viewed from $a$ or $b$ coincide.  Hence 
 \be \label{eqn_shape_match}
 \alpha^{ab}_{ij} =\alpha^{ba}_{ji}.
 \ee
This condition of shape matching therefore implies that
$
  0 = \alpha^{ab}_{ij} - \alpha^{ba}_{ji} = \xi^{ab}_{i} - \xi^{ab}_{j}
$
and so $\xi^{ab}_{i}$ is {\it independent} of  $i$.  This is the condition which will allow us to interpret $\xi^{ab}_i$ as the 4d dihedral angle between tetrahedra $a$ and $b$.

When the face matching condition is not satisfied, the geometry is twisted in the sense of \cite{Freidel:2010aq} 
and $\xi^{ab}_{i}$ represent a generalization of dihedral angles to twisted geometry.
Moreover, the on-shell action will therefore represent a generalization of the Regge action to twisted geometry.  

\begin{figure} 
  \centering
    \includegraphics[width=0.5\textwidth]{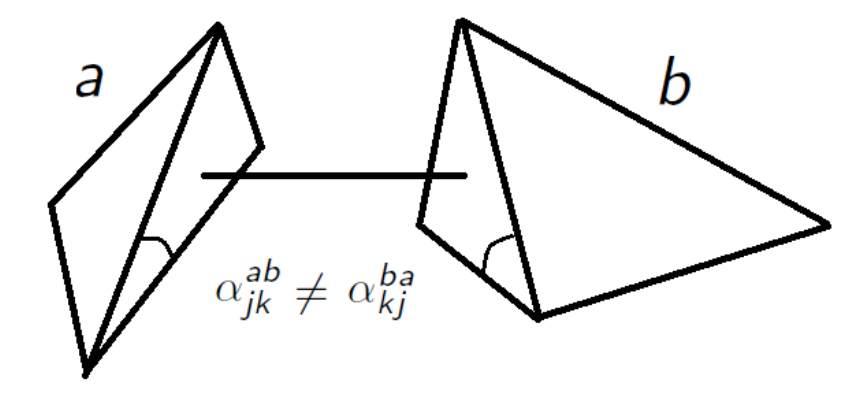}
    \caption{The gluing of tetrahedra in Twisted Geometry requires the areas of glued triangles to match, but the shapes can be different. }  \label{fig_twisted}
\end{figure}


Let us now express the on-shell action in terms of this data.

\begin{theorem} \label{thm_twisted_action}
  The generalization of the Regge action to twisted geometry is given by
\be\label{action}
S_\bbT =   \sum_{i<j} j_{ij} \xi^{ij} +  \sum_{a\neq i<j} k_{ij}^{a} \alpha_{ij}^{a}.
\ee
where 
\be\label{defang}
   \xi^{ij} \equiv \frac13  \sum_{k\neq (i,j)} \xi^{ij}_{k}, \quad   \alpha^{a}_{ij} \equiv  \frac16 \sum_{b\neq (i,j,a)} (\alpha^{ai}_{jb} +\alpha^{aj}_{ib}).
\ee 
\end{theorem} 
\begin{proof}
Lets first recall the expression (\ref{Ims}) for the imaginary part of 
 the on-shell action 
 \be \label{eqn_2I_action}
2{\mathrm{Im}}(S_{\mathrm{onshell}}) \equiv 2 I  = \sum_{a\neq i<j} k_{ij}^{a}  (\alpha_{j}^{ai} + \alpha_{i}^{aj})
=\sum_{(a,i,j)} k_{ij}^{a}  \alpha_{j}^{ai}.
\ee
where we denote by $(i,j)$, $(a,i,j)$ a set of index distinct from each other.
We now evaluate the sum using the symmetries $ k^{a}_{ij}= k^{a}_{ji}$ , $j_{ij}=j_{ji}$ and the relation 
$ \sum_{j} k^{a}_{ij} = 2 j_{ai}$.
\begin{align} \label{eqn_action_proof}
  \sum_{(a,i,j)} k^{a}_{ij} \alpha^{a}_{ij} 
  &= \frac16 \sum_{(a,i,j,b)} k^{a}_{ij} (\alpha^{ai}_{j} + \alpha^{aj}_{i} - \alpha^{ai}_{b} - \alpha^{aj}_{b} )
  =  \frac13 \sum_{(a,i,j,b)} k^{a}_{ij} (\alpha^{ai}_{j}  - \alpha^{ai}_{b}),  \\
  &= \frac13 \sum_{(a,i,j)} k^{a}_{ij} \sum_{b\neq (a,i,j)}\left(\alpha^{ai}_{j} - \alpha^{ai}_{b}\right)  
  = 2I - \sum_{(a,i)} (\sum_{j\neq(a,i)} k^{a}_{ij}) ( \frac13\sum_{b\neq (a,i)} \alpha^{ai}_{b}),
 \nn  \\
 &= 2I - 2\sum_{(a,i)} j_{ai}  ( \frac13\sum_{b\neq (a,i)} \alpha^{ai}_{b})
 = 2I - \sum_{(a,i)} j_{ai} \xi^{ai} \nn.
\end{align}
 as required.
\end{proof}
Here $\xi^{ij}$ measures the extrinsic curvature of the face $(ij)$ inside the 4-simplex.  It is a generalisation of the 
dihedral angle in the case of twisted geometry.  The angle $\alpha^{a}_{ij}$ is a geometrical angle\footnote{See equation (\ref{eqn_alpha_geo}).} associated with the edge $(ij)$ inside the tetrahedron $a$. 

The first term is a generalization of the Regge action while the second term defines a canonical phase for the intertwiners.   This agrees with the analysis of \cite{Barrett:2009gg} in which the asymptotics of the 4-simplex amplitude in the coherent intertwiner basis, i.e. (\ref{amplitude}), was computed.  This amplitude depends on the boundary data defined by the spinors and spins and the critical points depend on their associated geometry.  For the case where this data represents a non-degenerate, geometrical 4-simplex and the intertwiners are given a canonical phase then it is found that the asymptotic evaluation contains the cosine of the Regge action.  The cases of non-geometric and non-degenerate boundary data are also considered there. 

In the analysis given here the boundary data is defined by the integers $k^{a}_{ij}$ and it is found that the asymptotic evaluation also contains the Regge action when the geometricity conditions (\ref{eqn_shape_match}) are enforced and the phase given by the second term in (\ref{action}) is taken.  The action given here is however of a more general form which also includes non-geometric configurations.  In \cite{Barrett:2009gg} it is also shown that the critical points imply the existence of a 4-simplex embedded in $\R^4$ which is what we show next. 

\subsection{Geometricity and 4d Dihedral angles}
\label{section_geo}

In this section we will discuss the connection between the $\xi^{ij}$ angles and the 4d dihedral angles of a 4-simplex when shape matching is imposed.  To do so we first derive relations between the angles $\xi$ and $\theta$ from the 3-term relations.  Indeed, using the reality condition (\ref{eqn_reality_cond}) and $|z^{a}_{b}\ket\bra z^{a}_{b}| + |z^{a}_{b}][ z^{a}_{b}| = A_{ab}\one$ we have 
\bea
  \left[z^{a}_{i}|z^{a}_{b} \right] \bra z^{b}_{a} | z^{b}_{i}] - [z^{a}_{i}|z^{a}_{b}\ket [ z^{b}_{a} | z^{b}_{i}] = \epsilon_{ab}\epsilon_{ai}\epsilon_{bi} A_{ab} \bra z^{i}_{a}|z^{i}_{b}\ket \\
  \left[z^{a}_{i}|z^{a}_{b} \right] \bra z^{b}_{a} | z^{b}_{i}\ket - [z^{a}_{i}|z^{a}_{b}\ket [ z^{b}_{a} | z^{b}_{i}\ket = \epsilon_{ab}\epsilon_{ai}\epsilon_{ib} A_{ab} \bra z^{i}_{a}|z^{i}_{b}] 
\eea
which are given explicitly by
\bea
  c^{a}_{ib}s^{b}_{ai} - s^{a}_{ib} c^{b}_{ai} e^{i \xi^{ab}_{i}} &=& \epsilon_{ab}\epsilon_{ai}\epsilon_{bi} c^{i}_{ab} e^{i(\xi^{ib}_{a}+\xi^{ab}_i-\xi^{ai}_{b})/2}, \label{eqn_4d_3_term_1}\\
  c^{a}_{ib}c^{b}_{ai} - s^{a}_{ib} s^{b}_{ai} e^{i \xi^{ab}_{i}} &=& \epsilon_{ab}\epsilon_{ai}\epsilon_{ib} s^{i}_{ab} e^{i(-\xi^{ib}_{a}+\xi^{ab}_b-\xi^{ai}_{b})/2}, \label{eqn_4d_3_term_2}
\eea
where $c^{a}_{ij} \equiv \cos  \frac{\theta^{a}_{ij}}2$ and $ s^{a}_{ij}\equiv \epsilon_{ij}\sin  \frac{\theta^{a}_{ij}}2$.  Taking the difference of squares of equations $|(\ref{eqn_4d_3_term_1})|^2 - |(\ref{eqn_4d_3_term_2})|^2$ we get
\be \label{eqn_4d_dihedral}
 -\cos \theta^{a}_{ib}\cos \theta^{b}_{ai} - \epsilon_{ib} \epsilon_{ai} \sin \theta^{a}_{ib} \sin \theta^{b}_{ai} \cos \xi^{ab}_{i} = \cos \theta^{i}_{ab}.
\ee

Another way to derive this relationship is to use the relations $N^{a}_{b} = -N^{b}_{a}$ and $F^{a}_{b} = F^{b}_{a}$ and an argument similar to the one leading to (\ref{eqn_alpha_i_jk}) to show that 
\be
  L^{a}_{bi} \cdot L^{b}_{ai} = \epsilon_{ib} \epsilon_{ai} |L^{a}_{bi}| \cdot |L^{b}_{ai}| \cos \xi^{ab}_{i}.
\ee
Then using the definition $L^{a}_{bi} = A_{ab} A_{ai} N^{a}_{b} \times N^{a}_{i}$ one arrives at $(\ref{eqn_4d_dihedral})$.

In the twisted picture we have three different $\xi^{ab}_{i}$ for $i\neq a,b$ and $\xi^{ab}$ is their average.  In order for the tetrahedra to glue together into a geometrical 4-simplex we must impose the shape matching conditions (\ref{eqn_shape_match}).  We already noted that when these conditions are satisfied $\xi^{ab}_{i}$ is independent of $i$.  Then as shown in \cite{Bahr:2009qd} equation (\ref{eqn_4d_dihedral}) is the relationship between the 3d and 4d dihedral angles of a  4-simplex\footnote{Note that our convention $\theta^{a}_{ii}=0$ differs from the other convention $\theta^{a}_{ii}=\pi$.}.  

Let us now construct a condition on the boundary $k^{a}_{ij}$ such that the matching constraints are satisfied in the semiclassical limit.

Finally, we note that all the gauge invariant angles are entirely determined by the values of $k_{ij}^{a}$.
First the 3d dihedral angles are determined by the $k_{ij}^{a}$ via (\ref{eqn_k_z})
\be \label{eqn_sc_k}
\left(\sin \frac{\theta_{ij}^{a}}{2}\right)^{2} = \frac{J^{a} k_{ij}^{a}}{4 j_{i}^{a}j_{j}^{a}} 
\ee
and then $\alpha^{ai}_{jk}$ and $\xi^{ab}_{i}$ are related to $\theta^{a}_{ij}$ by (\ref{eqn_alpha_theta}) and (\ref{eqn_4d_dihedral}) respectively.  Furthermore, these relations give an interpretation of $k^{a}_{ij}$ in terms of spherical geometry.  





\chapter{Exact Evaluations}
\label{chapter_exact}

In this chapter we will start by building a generating functional which computes the group integrals in the coherent amplitude (\ref{eqn_coherent_4S}).  The definition of the generally coherent amplitude is not limited to the 4-simplex, it can be extended to arbitrary graphs, and this poses no obstacle to its evaluation. 

In Theorem \ref{thm_gauss} it is shown how the group integrals in the generating functional can be expressed as a multi-dimensional Gaussian integral.  The evaluation of this Gaussian integration is the determinant of a matrix $1+X$ depending only on the spinors.

The matrix $1+X$ is defined as a block matrix of outer products of the spinors.  This matrix has a special property described in Definition \ref{def_scalar_loop}.  We call the class of matrices satisfying this property ``scalar loop matrices''.  We prove that for such matrices, defined in terms of the non-commutative blocks, the determinant is related to a quasi-determinant, we call the Loop Determinant.  

The final result is given in Theorem \ref{thm_amp} where the coherent generating functional is expressed as a perfect square of a polynomial in the holomorphic spinor invariants $[z^{v}_{e}|z^{v}_{e'}\ket$ following the pattern of simple loops of the graph which {\it  do not share vertices or edges}.  We give an illustration of this generating functional for the dipole graph as well as the 3 and 4 simplices.

In section \ref{sec_k_gen_func} we go on to define a generating functional for the $k$-basis amplitudes, also for arbitrary graphs.  This generating functional is defined in terms of complex numbers $\tau^{v}_{ee'}$ which keep track of the data $k^{v}_{ee'}$.  We find a formula of the same form as the coherent generating functional where the $\tau^{v}_{ee'}$ now follow the pattern of simple loops of the graph which {\it do not share edges but can share vertices.}

We show that the $k$-basis generating functional reduces to the coherent generating functional when we set $\tau^{v}_{ee'} = [z^{v}_{e}|z^{v}_{e'}\ket$, i.e. when the $\tau^{v}_{ee'}$ satisfy the Pl\"ucker relations.

Finally we show how to use these generating functionals to derive Racah formulae for arbitrary graphs.  In particular the Racah formula for the 20j symbol is derived and an explicit parameterization is given in Appendix \ref{20j_symbol}.

\section{Evaluating the Coherent Amplitude}

A general spin network can be defined by a directed graph $\Gamma$ in which the edges are labeled by spins $j_e$
and vertices are labeled by intertwiners. The spin network amplitude is obtained by contracting the intertwiners along the  edges of $\Gamma$. 
Depending on the intertwiner basis we get different amplitudes.  



Coherent intertwiners are labeled by a spinor on each edge, therefore we assign two spinors $z_{e}, z_{e^{-1}}$ to each oriented edge $e$ of $\Gamma$, one for $e$ and one for the reverse oriented edge $e^{-1}$.  We define also define the sum of spins at a vertex by
$$J_{v} \equiv\sum_{e:s_{e}=v} j_{e} + \sum_{e: t_{e}=v} j_{e}$$
where $s_{e}$ (resp. $t_{e}$) is the starting (resp. terminal vertex) of the edge $e$.

The contraction of coherent intertwiners then produces an amplitude depending on $j_{e}$ and holomorphically on all $z_{e}$
\be \label{amplitude}
  A_\Gamma(j_e,z_e) \equiv \underset{v\in V_\Gamma}{\corner} {\|j_e,z_e \ket} = \int \prod_{v\in V_\Gamma} \rd g_v \prod_{e \in E_\Gamma} \frac{ [z_e|g_{s_e}g_{t_e}^{-1}|z_{e^{-1}}\ket^{2j_e}}{(2j_{e})!}.
\ee 
where we define $E_{\Gamma}$ to be the set 
of  edges of $\Gamma$ and  $V_{\Gamma}$  the set of vertices. 

Our main goal is to compute the group integrals in this expression.  To do this we express the Haar integrals over SU(2) as integrals over $\C^2$ using the following lemma, which was first shown in \cite{Livine:2011gp}:
\begin{lemma} \label{eqn_SU2_lemma}
Let $f \in L^2(SU(2))$ be homogeneous of degree $2J$, i.e. $f(\lambda g) = \lambda^{2J} f(g)$.  Given a spinor by $|z\ket$ define $g(z) = (|0\ket\bra 0| + |0][0|)g(z) = |0\ket\bra z| + |0][z|$ where $|0\ket = (1,0)^t$.  Then
\be
  \int_{\C^2} \rd\mu(z) f(g(z)) = \Gamma(J+2) \int_{\text{SU(2)}} \rd g \, f(g)
\ee
\end{lemma}
\begin{proof}
We can relate the inner product (\ref{barg_in_prod}) to the standard $L^2(\text{SU(2)})$ inner product by parameterizing the spinor as
\be
  |z\ket = \bpm r \cos\theta e^{i\phi} \\ r \sin\theta e^{i\psi} \epm
\ee
where $r \in (0,\infty)$, $\theta \in [0,\pi/2)$, $\phi \in [0,2\pi)$, $\psi \in [0,2\pi)$.  The Lebesgue measure in these coordinates is $\rd^4 z = r^3 \sin \theta \cos \theta \rd r \, \rd\phi \, \rd\theta \, \rd\psi$.  Now using the homogeneity property $f(g(z)) = r^{2J} f(\widetilde{g}(z))$ we have
\be
  \int_{\C^2} \rd\mu(z) f(g(z)) = \int_{0}^{\infty} \rd r \, r^{3+2j} e^{-r^2} \int_{0}^{\pi/2} \rd\theta \sin\theta \cos\theta \int_{0}^{2\pi} \rd\phi \int_{0}^{2\pi} \rd\psi f(\widetilde{g}(z))
\ee
where $\widetilde{g}(z) \in \text{SU(2)}$.  Performing the integral over $r$ we get
\be
  \int \rd r \, r^{3+2J} e^{-r^2} = \frac{1}{2} \Gamma(J+2)
\ee
and so
\be
  \int_{\C^2} \rd\mu(z) f(g(z)) = \Gamma(J+2) \int_{\text{SU(2)}} \rd g \, f(g)
\ee
where $\rd g$ is the normalized Haar measure on SU(2).  In our case $J$ is an integer so $\Gamma(J+2) = (J+1)!$.
\end{proof}
We can now use this lemma to convert the group integrals in (\ref{def3}) to Gaussian integrals.  This motivates the following generating functional depending purely on the spinors
\bea\label{def2}
{\cal A}_{\Gamma}(z_{e})
&\equiv& \sum_{j_{e}} \prod_{v\in V_{\Gamma}}(J_{v}+1)! A_{\Gamma}(j_{e},z_{e}), \\
&=& \sum_{j_{e}} \prod_{v\in V_{\Gamma}}(J_{v}+1)! \int \prod_{i\in V_{\Gamma}} \rd g_{i} \prod_{e \in E_{\Gamma}} \frac{[z_{e}|g_{s_e}g_{t_e}^{-1}|z_{e^{-1}}\ket^{2j_{e}}}{(2j_{e})!}. \label{def3}
\eea
We now show how this generating functional can be expressed as a Gaussian integral, the evaluation of which is an inverse determinant.
\begin{theorem} \label{thm_gauss}
The fully coherent amplitude (\ref{def2}) can be evaluated as a Gaussian integral
\bea \label{eqn_coherent_gauss_int}
{\cal A}_{\Gamma}(z_{e}) 
&=& \int_{\C^{2|V_{\Gamma}|}} \prod_{i \in V_{\Gamma}}\rd\mu(\alpha_{i}) 
\exp\left(-{\sum_{i,j \in V_{\Gamma}}\bra \alpha_{i} | X_{ij} |\alpha_{j} \ket}\right) = \frac{1}{\det(1+X(z_{e}))}
\eea
where $\rd\mu(\alpha) \equiv e^{-\bra\alpha |\alpha \ket } \rd^{4} \alpha / \pi^{2}$; 
and $X_{ij}$ is a 2 by 2 matrix which vanishes if there is no edge between $i$ and $j$.
If $(ij)=e$ is an edge of $\Gamma$,  $X_{ij}$ is given by
\be \label{eqn_matrix_X}
X_{ij} = \sum_{e|s_{e}=i, t_{e}=j} |z_{e}\ket [z_{e^{-1}}| - \sum_{e|t_{e}=i, s_{e}=j} |z_{e^{-1}}\ket [{z}_{e}|.
\ee
\end{theorem}
\begin{proof}
Define the spinors  $|\alpha_{i}\ket \equiv g_{i}|0\ket$ where $|0\ket = (1 \: 0)^T$.
Using the decomposition of the identity $\one=|0\ket\bra0| + |0][0|$,  we can express the group product
as 
\bea
g_{i} g_{j}^{-1} & = & g_{i}(|0\ket\bra0| + |0][0|) g_{j}^{-1} = |\alpha_{i}\ket\bra \alpha_{j}| + |\alpha_{i}][\alpha_{j}|
\eea  

Therefore by Lemma \ref{eqn_SU2_lemma} ${\cal A}_{\Gamma}(z_{e}) $ can be written as 
\bea
{\cal A}_{\Gamma}(z_{e}) 
&=& \int_{\C^{2|V_{\Gamma}|}} \prod_{i \in V_{\Gamma}}\rd\mu(\alpha_{i}) 
\exp\left({\sum_{e \in E_{\Gamma}} [z_{e}| \left( |\alpha_{s(e)}\ket\bra \alpha_{t(e)}| + |\alpha_{s(e)}][\alpha_{t(e)}| \right) |z_{e^{-1}}\ket} \right) 
\eea
where $\rd\mu(\alpha) \equiv e^{-\bra\alpha |\alpha \ket } \rd^{4} \alpha / \pi^{2}$. 

Using the relation $[\alpha |w\ket [z|\beta] = - \bra \beta | z\ket [w|\alpha\ket$ we can write the integrand as 
in (\ref{eqn_coherent_gauss_int}) where the 2 by 2 matrix $X_{ij}$ is given by
\be 
X_{ij} = \sum_{e|s_{e}=i, t_{e}=j} |z_{e}\ket [z_{e^{-1}}| - \sum_{e|t_{e}=i, s_{e}=j} |z_{e^{-1}}\ket [{z}_{e}| 
\ee
and $X_{ij}$ vanishes if there is no edge between $i$ and $j$.  We can now define a matrix $X(z_e)$ of size $2n \times 2n$ by the $2 \times 2$ blocks.  The covariance matrix of the Gaussian is then $1 + X$ where the identity matrix comes from the measure and the contraction is with respect to the $2n$-vector composed of the $n$ spinors $|\alpha_{i} \ket$ stacked on top of eachother.  

Now, if the matrix $1+X$ were Hermitian then the Gaussian integral could be evaluated as the determinant of the inverse matrix by diagonalisation with a unitary transformation giving the determinant formula (\ref{eqn_coherent_gauss_int}).  In more generality, it is well known that the same evaluation is valid provided that merely the Hermitian part of $1+X$ is positive definite.  This requirement is satisfied for spinors $|{z}_{e}\ket$ of sufficiently small norm, and by holomorphicity can be analytically continued to a maximal domain dictated by the poles in (\ref{eqn_coherent_gauss_int}).  This completes the proof.
\end{proof}

\subsection{The Loop Determinant}

We will now show how to evaluate the determinant (\ref{eqn_coherent_gauss_int}) explicitly as a sum of terms depending on the cycle structure of the graph $\Gamma$.  This is due to a special property of the matrix $X(z_e)$ which we define in Definition \ref{def_scalar_loop}.  This allows us to define a quasi-determinant which we call the Loop Determinant in Definition \ref{def_loop_det} which has a nice relation with the usual determinant as given in Proposition \ref{Ldet=det}.

Let us first make precise what we mean by loops and cycles on a graph $\Gamma$.
\begin{figure} 
  \centering
    \includegraphics[width=0.7\textwidth]{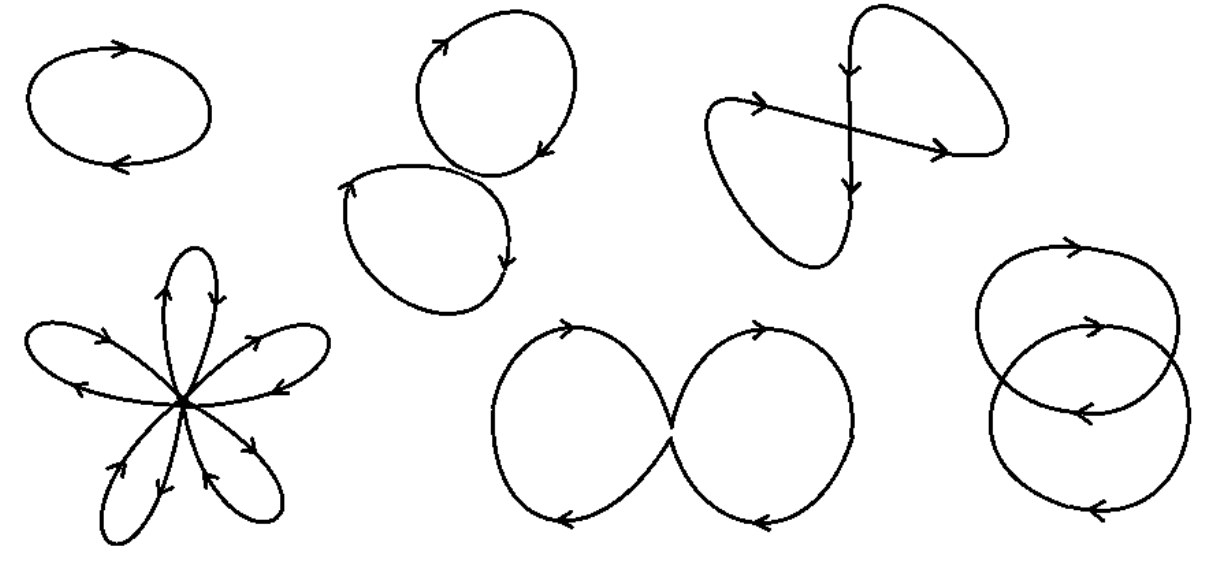}
    \caption{Some examples of paths on a graph which are collections of disjoint simple loops.  There is only one non-trivial cycle located in the top left.  Notice that the middle three diagrams each have an intersection of four edges at one vertex, but but they follow different paths.  These three crossings are analogous to the S, T, and U channels.}  \label{fig_simpleloops}
\end{figure}
\begin{definition}  \label{def_loops_cycles}
A loop of $\Gamma$ is a set of edges $l= e_{1}, \cdots e_{n}$ such that $ t_{e_{i}}= s_{e_{i+1}}$ and $ t_{e_{n}}= s_{e_{1}}$.  
A simple loop of $\Gamma$ is a loop in which $e_{i}\neq e_{j}$ for $i\neq j$, that is each edge enters at most once.
A non trivial cycle $c= (e_{1}, \cdots e_{n} )$ of $\Gamma$ is a simple  loop of $\Gamma$ in which  $s_{e_i} \neq s_{e_j}$ for $i \neq j$, 
i.e. it is a simple loop in which  each vertex is  traversed at most once.
\end{definition}

To define the Loop Determinant, we must first demonstrate how to write the usual determinant as a sum over cycle covers of a graph.  Recall the Laplace expansion of the determinant for a $n \times n$ matrix (of complex numbers)
\be
  \mathrm{det}(A) = \sum_{\pi} \text{sgn}(\pi) a_{1 \pi(1)} a_{2 \pi(2)} \cdots a_{n \pi(n)}.
  \label{eqn_Laplace_det}
\ee
An equivalent definition of the determinant can be given in terms of cycle covers of a complete directed graph on $n$ vertices \cite{mahajan1999determinant}.  On a complete graph we can label a loop by a sequence of vertices since there is only one edge between any two vertices.  A cycle is defined to be a simple loop for which all the vertices are distinct and a cycle cover is defined to be a collection of cycles which cover all the vertices in the graph, i.e. all of $\{1,...,n\}$.  

Notice that every permutation of $(1,...,n)$ corresponds to a unique partition of the set $\{1,...,n\}$ into disjoint cycles.  For example the permutation 
\be
  \pi =
  \bpm 
    1 & 2 & 3 & 4 & 5 & 6 \\
    2 & 4 & 6 & 1 & 5 & 3
  \epm
\ee
corresponds to the cycle cover $\cC = (124)(36)(5)$.  

The weight of a cycle $C=(c_{1}\cdots c_{i})$ is defined to be $W(C) = a_{c_1 c_2} a_{c_2 c_3} ... a_{c_i c_1}$ and the weight of a cycle cover is the product of the weights of its cycles.  The weight of a loop is defined in the same way.
Furthermore, it can be shown that the sign of a permutation is equal to $(-1)^{n+k}$ where $k$ is the number of cycles in its corresponding cover.   Therefore, Eq. (\ref{eqn_Laplace_det}) can be written as
\be
  \mathrm{det}(A) = \sum_{\cC} \mathrm{sgn}(\cC) W(\cC).
  \label{eqn_det_cycles}
\ee
where $\sum_{\cC} = (-1)^{n+k}$.  

Now suppose that the matrix $A$ is composed of elements which are noncommutative such as $2 \times 2$ matrices in the case of Eq. (\ref{eqn_matrix_X}).  In this case we lose many useful relations of the determinant such as the multiplicative property and the behavior with respect to elementary row operations due to the noncommutativity.  Yet for special types of matrices which we call scalar loop matrices we can define a quasi-determinant for which these properties still hold.

\begin{definition}  \label{def_scalar_loop}
A matrix is called a scalar loop matrix if for any loop $L$ the quantity $S(L) = \frac{1}{2}(W(L) + W(L^{-1}))$ is scalar where $L$ and its inverse $L^{-1}$ begin with the same element but the sum is otherwise invariant under cyclic permutations of $L$. 
\end{definition}
\begin{definition} \label{def_loop_det}
Let $A$ be a $n$ by $n$ scalar loop matrix.  The loop determinant of $A$ is defined to be
\be
  \mathrm{Ldet}(A) = \sum_{\cC} \mathrm{sgn}(\cC) S(\cC)
  \label{eqn_sc_det_cycles}
\ee
where the sum is over all cycle covers $\cC = C_1...C_k$ on $\{1,..,n\}$.
\end{definition}
Note that for a commutative matrix Eq. (\ref{eqn_sc_det_cycles}) is equivalent to Eq. (\ref{eqn_det_cycles}) for which the multiplicative property of the determinant was studied in \cite{kovacs1999determinants}.  

The following proposition shows that loop determinant behaves well under elementary row operations, which will be useful when we prove its relation with the usual determinant of a block matrix.
\begin{proposition}
Let $A$ be a scalar loop matrix.  Then the loop determinant behaves as the usual determinant under all the elementary row operations.
In particular the addition of a scalar multiple of one row of $A$ to another row leaves the loop determinant invariant.
\label{thm_sc_det_row_operation}
\end{proposition}
\begin{proof}
Suppose we add a scalar multiple $\lambda$ of row $i$ of $A$ to row $j$.  Then Eq. (\ref{eqn_sc_det_cycles}) is changed by replacing the single factor $A_{i \cdot}$ in each weight by $A_{i \cdot} + \lambda A_{j \cdot}$.  Therefore Eq. (\ref{eqn_sc_det_cycles}) becomes a sum of its original terms plus terms proportional to $\lambda$.  We will now show that all terms proportional to $\lambda$ cancel each other.

Let $\cC$ be a cycle cover of $1,...,n$.  Then there exists two possibilities: $i$ and $j$ are in the same cycle or $i$ and $j$ are in different cycles. Suppose that they are in the same cycle $C$ and let $\cC^\prime$ be the rest of $\cC$.  By cyclic invariance we can assume that $i = c_1$ and call $j=c_j$ where $C = (c_1 ... c_j... c_N)$. Replacing $A_{c_1 c_2}$ with $A_{c_{1} c_{2}} + \lambda A_{c_{j} c_{2}}$ in $W(C)$ we get 
\begin{align}
  W(C) 
  \rightarrow (A_{c_1 c_2} + \lambda A_{c_{j} c_{2}}) A_{c_2 c_3} \cdots A_{c_{j-1}c_{j}} A_{c_{j} c_{j+1}} \cdots A_{c_{N}c_{1}} 
  = W(C) + \lambda W(\widetilde{C}) N(C)
\end{align}
where $\widetilde{C} = (c_j c_{2} c_{3} ... c_{j-1})$ and $N(C) = A_{c_{j} c_{j+1}} A_{c_{j+1} c_{j+2}} \cdots A_{c_{N}c_{1}}$.  Now consider the cycle $\widehat{C} = (c_1 c_{j+1} c_{j+2} ... c_N)$ then
\begin{align}
  W(\widehat{C}) 
  \rightarrow (A_{c_1 c_{j+1}} + \lambda A_{c_{j} c_{j+1}}) A_{c_{j+1} c_{j+2}} \cdots A_{c_{N}c_{1}} 
  = W(\widehat{C}) + \lambda N(C)
\end{align}
and moreover
\be
 W(\widetilde{C}) W(\widehat{C})
  \rightarrow W(\widetilde{C})W(\widehat{C}) + \lambda W(\widetilde{C}) N(C)
\ee
This demonstrates that $W(C)$ and $W(\widetilde{C})W(\widehat{C})$ produce terms proportional to $\lambda$ which are equal but have opposite sign in Eq. (\ref{eqn_sc_det_cycles}) since $\text{sgn}(\widetilde{C} \widehat{C}) = -\text{sgn}(C)$.  We now show exactly how these terms cancel in Eq. (\ref{eqn_sc_det_cycles}), by considering eight cycle covers for which the terms proportional to $\lambda$ all cancel eachother.  Indeed, let 
\begin{align}
C_1 &= (c_{1} c_{2}  ... c_{j-1}), C_2 = (c_{j} c_{j+1} c_{j+2} ... c_{N}), C_3 = (c_{1} c_{N} c_{N-1} ... c_{j+1}), \nonumber \\
C_4 &= (c_{j} c_{j-1} c_{j-2}... c_{2}), C_5 = (c_{1} c_{2} ... c_{j-1} c_{j} c_{j+1} ... c_{N}), C_6 = (c_{1} c_{2} ... c_{j-1} c_{j} c_{N} c_{N-1} ... c_{j+1}), \nonumber \\
C_7 &= (c_{1} c_{j-1} c_{j-2} ... c_{2} c_{j} c_{j+1} ... c_{N}), C_8 = (c_{1} c_{j-1} c_{j-2} ... c_{2} c_{j} c_{N} c_{N-1} ... c_{j+1}) \nonumber
\end{align}
then it is straightforward to show that
\be
  S(C_1)S(C_2) + S(C_3)S(C_4) - S(C_5) - S(C_6) - S(C_7) - S(C_8)
\ee
is invariant after the row operation, i.e. the terms proportional to $\lambda$ cancel.  Conversely, if $c_1$ and $c_j$ are in different cycles we can write them as $C_1$ and $C_2$ in which case we can construct $C_3$,..., $C_8$ which leads to the same cancellation.

It is easy to see from Eq. (\ref{eqn_det_cycles}) that multiplying a row by a scalar produces an overall factor of $\lambda$ and switching two rows produces a minus sign, just like the determinant over a field.  Hence the loop determinant behaves as one would expect under all the elementary row operations.
\end{proof}
The reason we are interested in the loop determinant is because of the following observation.
\begin{proposition}\label{Ldet=det}
Let $A$ be a scalar loop matrix composed of block matrices and denote the ordinary determinant by $|A|$.  Then
\be
  |A| = \left|\mathrm{Ldet}(A)\right|
\ee
\end{proposition}
\begin{proof}
By Theorem \ref{thm_sc_det_row_operation} the loop determinant is unchanged after Gaussian elimination so after eliminating the first column
\be
  \mathrm{Ldet}(A) 
  = \mathrm{Ldet}
  \bpm 
  A_{11} & A_{12} & \ldots & A_{1n} \\
  A_{21} & A_{22} & \ldots & A_{2n} \\
  \vdots & \vdots & \ddots & \vdots \\
  A_{n1} & A_{n2} & \ldots & A_{nn}
  \epm 
  = \mathrm{Ldet}
  \bpm
  A_{11} & A_{12} & \ldots & A_{1n} \\
  0 &    &   &    &  \\
  \vdots &   & B  &  \\
  0 &    &   &    &
  \epm
\ee
where $B$ is a $(n-1) \times (n-1)$ matrix with entries $B_{ij} = A_{ij} - A_{i1} A_{11}^{-1} A_{1j}$.  Note that since $A$ is a scalar loop matrix $A_{11}$ is scalar so $A_{11}^{-1}$ does indeed exist and is also scalar.  Furthermore, if $L=(l_1 l_2 ... l_i)$ is a loop of $\{2,3,...,n\}$ then $W_B(L) = B_{l_1 l_2} B_{l_2 l_3} \cdots B_{l_i l_1}$ can be expressed as
\be
  W_B(L) = W_A(L) + \sum_{\sigma} (-A_{11}^{-1})^{|\sigma|} W_{A}(L(\sigma))
\ee
where $\sigma \subset \{1,2,...,i\}$ and $L(\sigma) = (l_1 ... l_{\sigma_1} 1 l_{\sigma_1 + 1} ... l_{\sigma_2} 1 l_{\sigma_2+1} ... l_{i})$, i.e. it is $L$ with 1 inserted after every element of $\sigma$.  In other words $L(\sigma)$ is a loop of $\{1,2,3,...,n\}$ and so $S_B(L)$ is scalar which shows that $B$ is a scalar loop matrix.  

The hypothesis is clearly true for $n=1$ so now assume it is true for scalar loop matrices of size $(n-1) \times (n-1)$.  Then $|B| = \left|\mathrm{Ldet}(B)\right|$ which then implies
\be
  |A| = |A_{11}| \cdot |B| = |A_{11}| \cdot \left|\mathrm{Ldet}(B)\right| = \left|\mathrm{Ldet}(A)\right|
\ee 
which advances the induction hypothesis.
\end{proof}
The name scalar loop matrix comes from the fact that the collection of indices $L$ corresponds to a loop on the complete directed graph on $n$ vertices.  We will show that the matrix $X(z_e)$ in (\ref{eqn_coherent_gauss_int}), when viewed as a $n \times n$ matrix of $2 \times 2$ blocks defined in Eq. (\ref{eqn_matrix_X}), has precisely this property.  This will allow us to prove our main theorem which is the final expression for the generating functional in Theorem \ref{thm_amp}.
\begin{lemma} \label{lemma_X_scalar_loop}
The matrix $1+X$ in Eq. (\ref{eqn_matrix_X}) is a scalar loop matrix.
\end{lemma}
\begin{proof}
First suppose $\Gamma$ is a complete oriented graph so that we can continue to label loops by pairs of vertices and let $L = (l_1 l_2 \cdots l_i)$ be a loop on $\{1,...,n\}$.  Then $X_{l_j l_k} = |z_{l_j l_k} \ket [z_{l_k l_j} |$ if the edge from $l_j$ to $l_k$ is positively oriented and the negative otherwise.  Suppose that $L$ has $|e|$ edges which are opposite the orientation.  Then
\be
  W(L) = (-1)^{|e|} |z_{l_1 l_2} \ket [z_{l_2 l_1} | z_{l_2 l_3} \ket \cdots [z_{l_{i} l_{i-1}} | z_{l_i l_1} \ket [ z_{l_1 l_i} |
\ee
and 
\be
  W(L^{-1}) = (-1)^{|e|+i} |z_{l_1 l_{i}} \ket [z_{l_{i} l_1} | z_{l_{i} l_{i-1}} \ket \cdots [z_{l_{2} l_{3}} | z_{l_2 l_{1}} \ket [ z_{l_1 l_2} |
\ee
Now using the identity $[z | w \ket = - [w | z \ket$ we have an extra factor of $(-1)^{i-1}$ in the second term and so
\be
  W(L) + W(L^{-1}) = (-1)^{|e|} [z_{l_2 l_1} | z_{l_2 l_3} \ket \cdots [z_{l_{i} l_{i-1}} | z_{l_{i} l_1} \ket \Big( |z_{l_1 l_2} \ket [ z_{l_1 l_i} | - |z_{l_1 l_i} \ket [ z_{l_1 l_2} | \Big)
\ee
now using $|z \ket [w| - |w \ket [z| = -[z|w\ket \one$ we have
\be
  S(L) \equiv \frac{1}{2}\left(W(L) + W(L^{-1})\right) = \frac{(-1)^{|e|}}{2} [z_{l_1 l_i} | z_{l_1 l_2} \ket [z_{l_2 l_1} | z_{l_2 l_3} \ket \cdots [z_{l_{i} l_{i-1}} | z_{l_{i} l_{1}} \ket \one
\label{eqn_weight_cycle}
\ee
By writing $X_{ij}$ as in Eq. (\ref{eqn_matrix_X}) we generalize $\Gamma$ to have any number of edges between pairs of vertices. In that case it is clear that $S(L)$ is equal to the sum of weights of the form on the r.h.s. of Eq. (\ref{eqn_weight_cycle}) over all loops in $\Gamma$ traversing the vertices $(l_1 l_2 \cdots l_i)$ in order.
\end{proof}
Finally we apply the previous lemmas to the matrix $1+X$ in Eq. (\ref{eqn_matrix_X}) to prove our main theorem of this section.  The evaluation of the determinant (\ref{eqn_coherent_gauss_int}) is given by
\begin{theorem} \label{thm_amp}
Let $\Gamma$ be an arbitrary graph with data given as in Lemma \ref{thm_gauss}.  Given a non trivial cycle $c= (e_{1}, \cdots ,e_{n})$ we define the quantity
\be
A_{c}(z_{e}) \equiv -(-1)^{|e|} [\tilde{z}_{e_{1}} | z_{e_{2}}\ket [\tilde{z}_{e_{2}}|z_{e_{3}}\ket \cdots  [\tilde{z}_{e_{n}} | z_{e_{1}}\ket
\ee
where 
$|e|$ is the number of edges of $c$ whose orientation agrees with the chosen orientation of  $\Gamma$,
and $\tilde{z}_{e}\equiv z_{e^{-1}}$.  We define a disjoint cycle union of $\Gamma$ to be a collection $C=\{c_{1},\cdots, c_{k}\}$ of non trivial cycles of $\Gamma$ which are pairwise disjoint (i.e. do not have any common edges or vertices).  Given a disjoint cycle union $C=\{c_{1},\cdots, c_{k}\}$ we define
\be \label{eqn_cycle_union_amp}
A_{C}(z_{e}) = A_{c_{1}}(z_{e})\cdots A_{c_{k}}(z_{e}).
\ee
The fully coherent amplitude is given by
\be \label{eqn_coh_gen_func_form}
{\cal A}_{\Gamma}(z_{e}) = \frac1{\left(1 + \sum_{C} A_{C}(z_{e})\right)^{2}}
\ee
where the sum is over all disjoint cycle unions $C$ of $\Gamma$.
\end{theorem} 
\begin{proof}
By lemma \ref{lemma_X_scalar_loop} the matrix $1+X$ is a scalar loop matrix.  Therefore by proposition \ref{Ldet=det}
\be
  |1+X| = \left|\mathrm{Ldet}(1+X)\right| = \left( \sum_{\cC} \mathrm{sgn}(\cC) S(\cC) \right)^2
\ee
where the sum is over all cycle covers of $V_\Gamma$.  Since the loop determinant is a scalar (proportional to the 2 by 2 identity), its determinant is a perfect square.  The 1-cycles of $1+X$ correspond to the diagonal which all have weight 1.  The cycle cover of all 1-cycles produces the term equal to unity.  The 2-cycles of $1+X$ all vanish since $[z_e|z_e \ket=0$. Therefore the cycle covers consist of disjoint unions of non-trivial cycles with the remaining vertices covered by 1-cycles.  This is enough to see that the weight from the loop determinant formula agrees with the weight in Eq. (\ref{eqn_cycle_union_amp}).  Now the sign of each term is $(-1)^{n+k}$ from the cycle cover and $(-1)^{|e|}$ from the weight formula in Eq. (\ref{eqn_weight_cycle}).  If a cycle cover has $i$ non-trivial cycles covering $n-r$ vertices then there are $k = i+r$ cycles in the cover.  Thus if we assign $(-1)^{|n|+|e|+1}$ to each non-trivial cycle where $|n|$ is the number of vertices in the cycle then $\sum (|n|+1) = (n-r) + i  = n+k-2r$ which agrees with the weight from the cycle cover.
\end{proof}

\subsection{Illustration}
Let us illustrate Theorem \ref{thm_amp} on one of the simplest graphs: the theta graph $\Theta_{n}$.  This graph consists of two vertices with $n$ edges running between them.  
The amplitude for this graph depends on $2n$ spinors denoted $z_{i}$ for the spinors attached to the first vertex and $w_{i}$ for the ones attached to the second vertex. 
We choose the orientation of all the edges to be directed from $z_{i}$ to $w_{i}$ where $i=1,\cdots, n$ labels the edges of $\Theta_{n}$.

For this graph the only cycles which have non-zero amplitudes are of length 2.  Further, since there are only two vertices, each disjoint cycle union
consists of a single nontrivial cycle. The amplitude associated to such a cycle going along the edge $i$ and then $j$ is given by
\be
A_{ij} = [w_{i}|w_{j}\ket[z_{j}|z_{i}\ket.
\ee
Therefore, from our general formula (\ref{eqn_coh_gen_func_form}) we have
\be\label{theta}
{\cal A}_{\Theta_{n}}(z_{i},w_{i}) = \left(1 + \sum_{i<j}[w_{i}|w_{j}\ket[z_{j}|z_{i}\ket \right)^{-2}.
\ee
\begin{figure} 
  \centering
    \includegraphics[width=0.5\textwidth]{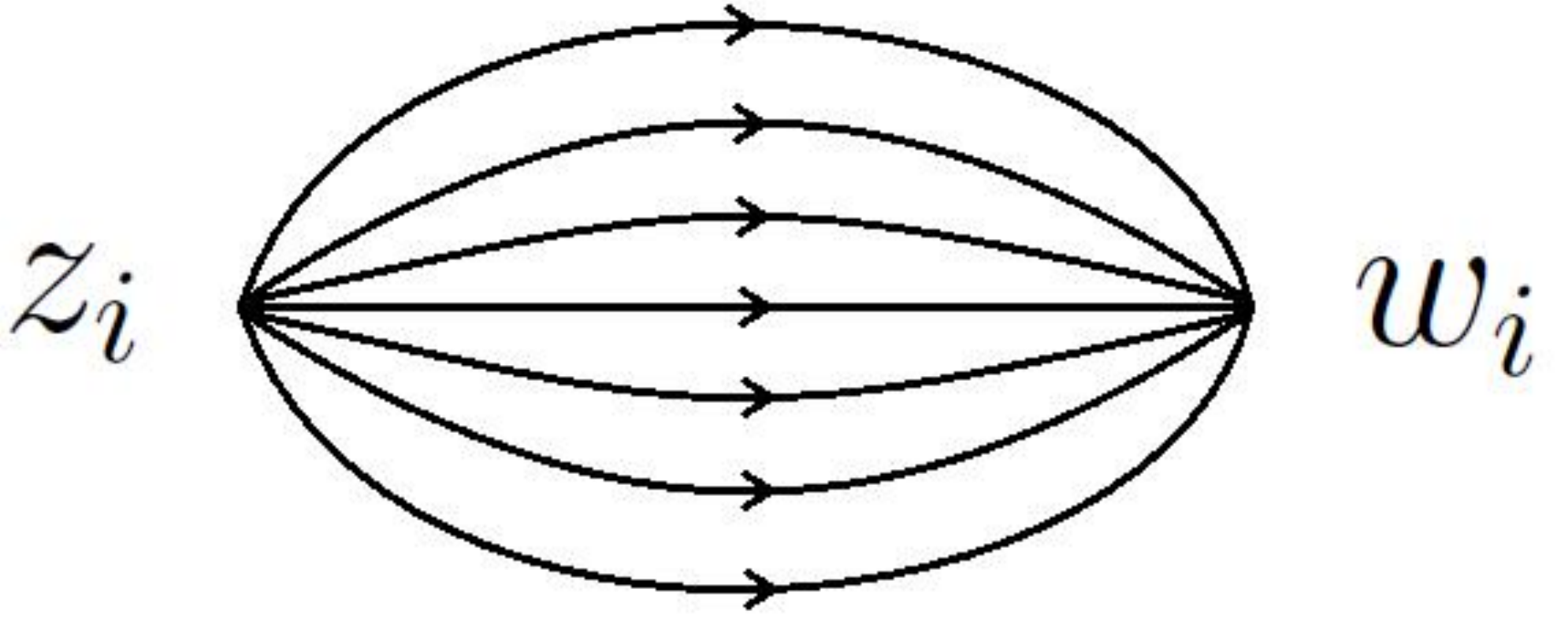}
    \caption{The theta graph with, in general, $n$ legs.  There is a spin $j_i$ and two spinors $\{z_i,w_i\}$ defined for each leg with the orientation from $z$ to $w$.  This amplitude corresponds to the scalar product of coherent intertwiners. }
\end{figure}

The theta graph amplitude gives information about the scalar product of intertwiners. 
In the next section we show how to use this generating function to construct a resolution of identity on the space of $n$-valent intertwiners.

We now illustrate the theorem for cases of the 3-simplex and the 4-simplex. 
In a $n$-simplex there is exactly one oriented edge for any pair of vertices $e=[ij]$ and so we can label cycles by sequences of vertices.  We choose the 
orientation of the simplex to be such that positively oriented edges are given by $e=[ij]$ for $i<j$.
Associated to the oriented edge $e=[ij]$ we assign the spinors $$z_{e} \equiv z^{i}_{j},\qquad \tilde{z}_{e}=
z_{e^{-1}} \equiv z^{j}_{i}.$$ 

Given a non trivial cycle $(1,2, \dots, p)$ of a $n$-simplex we define its amplitude by
\be
A_{12\cdots p}\equiv   [z^{1}_{p}|z_{2}^{1}\ket[z_{1}^{2}|z_{3}^{2}\ket \cdots [z_{p-1}^{p}|z_{1}^{p}\ket.
\ee
For the 3-simplex we have four non-trivial cycles of length $3$ and three non-trivial cycles of length $4$.  Since each of these cycles share a vertex or edge with every other, the only disjoint cycle unions are those which contain one non-trivial cycle.  Therefore, after taking into account the sign convention the 3-simplex amplitude is given by
\be  \label{eqn_6j_gen_func}
{\cal A}_{3S}= \bigg(1 - A_{123} - A_{124} - A_{134} - A_{234} + A_{1234} - A_{1243} - A_{1324} \bigg)^{-2}.
\ee
The sign in front of $A_{123}$ is determined in the following way.  First, there is one $-1$ which comes from the cycle union having one non trivial cycle and two $-1$ because the non trivial cycle $(1,2,3)$ contains the two edges $12$ and $23$ which have a positive orientation.  Thus the sign is negative.  

Expanding the generating functional (\ref{eqn_6j_gen_func}) in power series produces the Racah coefficients, i.e. the 6j symbol.

For the 4-simplex we have ten 3-cycles, fifteen 4 cycles, and twelve 5 cycles and again the disjoint cycle unions consist of only single cycles.
We define the 3-cycle amplitude to be
\be
A_{3}\equiv A_{123} + A_{124}+ A_{134} + A_{234} + A_{125}  +A_{135}+ A_{345} + A_{145} + A_{245} + A_{345},
\ee
the 4-cycle amplitude to be
\be
A_{4}\equiv \hat{A}_{1234} + \hat{A}_{1235} + \hat{A}_{1245} + \hat{A}_{1345} + \hat{A}_{2345},
\quad \mathrm{
with
}\quad 
\hat{A}_{1234}= A_{1234}- A_{1324}- A_{1243}.
\ee
and the 5-cycle amplitude to be
\bea\nonumber
A_{5} &=&  A_{12345} - A_{12435} - A_{23541}- A_{34152} - A_{45213}- A_{51324}\\
& &
- A_{12453} - A_{23514} -A_{34125} - A_{45231}- A_{51342} - A_{13524}.
\eea
Finally, the 4-simplex amplitude is given by
\be
{\cal A}_{4S} = ( 1 - A_{3} + A_{4} - A_{5})^{-2}.
\ee
By expanding this expression for the 4-simplex generating functional we will derive a Racah formula for the 20j symbol in Section \ref{section_racah}.

\subsection{Relating the Coherent and Discrete-Coherent Intertwiners}
\label{section_relate_coh_discoh}

We now would like to understand the relationship between the discrete-coherent basis of intertwiners from Section \ref{section_discrete_coherent} and the Livine-Speziale coherent intertwiners, Section \ref{section_coherent_intertwiners}.  In particular we would like to find the scalar product between these states.
In order to investigate this, let us introduce the normalised intertwiner basis
\be\label{Chat}
\widehat{C}_{[k]}^{(n)}(z_{i}) \equiv  \frac{ \prod_{i<j} [z_{i}|z_{j}\ket^{k_{ij}}}{\sqrt{ (J+1)! \prod_{i<j} k_{ij}!}} 
= \sqrt{\frac{ \prod_{i<j} k_{ij}!}{(J+1)!}} C_{[k]}^{(n)}. 
\ee

Intuitively, the theta graph consists of two $n$-valent intertwiners with pairs of legs identified, i.e. the scalar product.  Indeed, expanding the theta graph amplitude (\ref{theta}) in a power series yields an expression in terms of these intertwiners
\bea
{\cal A}_{\Theta_{n}}(z_{i},w_{i}) &=& \sum_{J} (-1)^J (J+1) \left(\sum_{i<j}  [w_{i}|w_{j}\ket[z_{j}|z_{i}\ket\right)^{J} \\
&=& \sum_{[k]}     {(J+1)!} \frac{
\prod_{i<j}  [w_{i}|w_{j}\ket^{k_{ij}} [z_{i}|z_{j}\ket^{k_{ij}}
}{\prod_{i<j} k_{ij}!}
\\ 
&=& \sum_{j_{i}} \left[(J+1)! \right]^{2}  \sum_{[k] \in K_{j}}\widehat{C}_{[k]}^{(n)}(z_{i})\widehat{C}_{[k]}^{(n)}(w_{i}).
\eea
This shows that ${\cal A}_{\Theta_{n}}(z_{i},w_{i})$ is a generating functional for the $n$-valent intertwiners.
Given the definition (\ref{def2}) of the amplitude ${\cal A}_{\Theta_{n}}(z_{i},w_{i})$ 
this implies that 
\be\label{CC}
  \sum_{[k] \in K_{j}}\widehat{C}_{[k]}^{(n)}(z_{i})\widehat{C}_{[k]}^{(n)}(w_{i}) =\int \rd g  \prod_{i}\frac{[z_{i}|g|w_{i}\ket^{2j_{i}} }{(2j_{i})!}.
\ee
This equation expresses the relation between the scalar product of the coherent intertwiners with the scalar product of the basis (\ref{Chat}).

We now have to understand the normalization properties of  $\widehat{C}_{[k]}^{(n)}$.
In order to do so, it is convenient to introduce another generating functional defined by
\be\label{ACC}
\widehat{\cal A}_{\Theta_{n}}(z_{i},w_{i}) \equiv  \sum_{[k]}   \widehat{C}_{[k]}^{(n)}(z_{i})\widehat{C}_{[k]}^{(n)}(w_{i}).
\ee
The remarkable fact about this generating functional, which follows from (\ref{CC}), is that it can be written as the evaluation of the following integral
\be
\widehat{\cal A}_{\Theta_{n}}(z_{i},w_{i}) =\int_{\SU(2)} \rd g \, e^{\sum_{i}[z_{i}|g|w_{i}\ket} \, .
\ee
We can now compute 
\bea
\int \prod_{i}\rd\mu(w_{i}) \left|\widehat{\cal A}_{\Theta_{n}}(z_{i},w_{i})\right|^{2} &=& \int \rd g \rd h 
\int \prod_{i}\rd\mu(w_{i}) e^{\sum_{i}[z_{i}|g|w_{i}\ket + \sum_{i}\bra w_{i}|h^{-1}|z_{i} ]}\\
&=& \int \rd g \rd h \,e^{\sum_{i}[z_{i}|gh^{-1}|z_{i} ]} 
= \widehat{\cal A}_{\Theta_{n}}(z_{i}, \check{z}_{i})
\eea
where $ |\check{z}_{i}\ket \equiv |z_{i}]$ and in the second line we evaluated the Gaussian integral.  Using (\ref{ACC}) to write this equality in terms of the intertwiner basis we get
\be
\sum_{[k],[k']} \widehat{C}_{[k']}^{(n)}(z_i) \left\bra \widehat{C}_{[k']}^{(n)} \right|\left. \widehat{C}_{[k]}^{(n)} \right\ket  \widehat{C}_{[k]}^{(n)}(\check{z}_{i})
=   \sum_{[k]}   \widehat{C}_{[k]}^{(n)}(z_{i})\widehat{C}_{[k]}^{(n)}(\check{z}_{i})
\ee
where we have used that ${C}_{[k]}^{(n)}(\check{z}_{i})$ is the complex conjugate of $C_{[k]}^{(n)}({z}_{i})$, i.e. $[\check{w}|\check{z}\ket=-\bra w | z ]= \bra z|w]= \overline{[w|z\ket}$, and the scalar product is defined with respect to the measure (\ref{barg_in_prod}).  This shows that the combination 
\be\label{proj}
P _{j}(z_i,z'_i) \equiv  \sum_{[k]\in K_{j} }    \widehat{C}_{[k]}^{(n)}(z_i) \widehat{C}_{[k]}^{(n)}(z'_i)
\ee
is a projector onto the space of SU$(2)$ intertwiners of spin $j_{i}$.  This proves the $n$-valent generalization of Theorem \ref{thm_completeness}.

\section{Generating Functionals}
\label{sec_k_gen_func}

In this section we construct a generating functional to compute the amplitudes of $k$-basis contractions.  We first warm up by constructing a generating functional for the scalar product $\bra k_{ij}| k'_{ij}\ket$.  We find that the scalar product is orthogonal up to terms generated by the Pl\"ucker relations.

We then generalize this construction by defining a generating functional (\ref{defG}) for arbitrary graphs.  We define a matrix $\T$ which makes the Gaussianity of this generating functional explicit.   
In Proposition \ref{prop_gen_func_coh_amp} we show explicitly how this generating functional can be related to the previous one.  In Theorem \ref{thm_gen} we evaluate the generating functional in terms of simple loops on the graph.  Finally Corollary \ref{cor_gen} shows explicitly how the Pl\"ucker relations conspire to eliminate all unions of simple loops that share vertices.  This gives an independent proof of Theorem \ref{thm_amp}.

\subsection{The Scalar Product}

We would like now to provide a  direct evaluation of the scalar product between two discrete-coherent intertwiners.
In order to do so 
we introduce the following generating functional which depends holomorphically on $n$ spinors $|z_{i}\ket$ and $n(n-1)/2$ complex numbers
$\tau_{ij} =-\tau_{ji}$
\be\label{defC}
 {\cal C}_{\tau_{ij}}(z_{i})  \equiv e^{\sum_{i<j} \tau_{ij}[z_{i}|z_{j}\ket } = \sum_{[k] } 
 \prod_{i<j}  C_{[k]}(z_{i}) \tau_{ij}^{k_{ij}}.
\ee
This functional was first consider by Schwinger \cite{schwinger2001angular}.
We now compute the scalar product between two such intertwiners
\bea\label{CC2}
\left\bra {\cal C}_{\tau_{ij}} | {\cal C}_{\tau_{ij}} \right\ket & = & \int \prod_{i}\rd\mu(z_{i})   \left|{\cal C}_{\tau_{ij}}(z_{i})\right|^{2}\\
& = &  \int \prod_{i}\rd\mu(z_{i}) e^{\sum_{i<j} \tau_{ij}[z_{i}|z_{j}\ket + \bar{\tau}_{ij} \bra z_{j}|z_{i}]}.
\eea
If we denote by $\alpha_{i}\in \C$ and $\beta_{i} \in \C$ the two components of the spinor $z_{i}$, 
and use that $[z_{i}|z_{j}\ket = \alpha_{i}\beta_{j} - \alpha_{j} \beta_{i}$ together with the antisymmetry of $\tau_{ij}$, this integral reads 
\be
\int \prod_{i}\rd\mu(\alpha_{i})\rd\mu(\beta_{i}) e^{\sum_{i,j}( \tau_{ij} \alpha_{i}\beta_{j} +\bar{\tau}_{ij}\bar{\alpha}_{i}\bar{\beta}_{j})}
\ee
with $\rd \mu(\alpha)=e^{-| \alpha |^{2}} \rd \alpha/\pi $.
We can easily integrate over $\beta_{j}$, since the integrand is linear in $\beta_{j}$ and we obtain:
\be
\int \prod_{i}\rd\mu(\alpha_{i}) e^{\sum_{i,j,k}\alpha_{i} \tau_{ij} \bar{\tau}_{kj} \bar{\alpha}_{k}}
= \frac{1}{\det(1 + T\overline{T})}
\ee
where $T = (\tau_{ij})$ and $\overline{T} = (\overline{\tau}_{ij})$.  In the case where $n=3$ this determinant can be explicitly evaluated and it is given by
\be
\det(1 + T\overline{T}) = \left(1-\sum_{i<j} |\tau_{ij}|^{2} \right)^{2}
\ee
In the case $n=4$  the explicit evaluation  gives 
\be
\det(1 + T\overline{T}) = \left(1-\sum_{i<j} |\tau_{ij}|^{2} +  |R|^{2}\right)^{2}
\ee
where 
\be \label{eqn_tau_plucker}
R(\tau)= \tau_{12}\tau_{34} +  \tau_{13} \tau_{42}+ \tau_{14}\tau_{23}.
\ee
Note that the Pl\"ucker identity tells us that $R=0$ when $\tau_{ij} =[z_{i}|z_{j}\ket$.  

By expanding the LHS of (\ref{CC2}) for $n=4$
\bea
 \left\bra {\cal C}_{\tau_{ij}} | {\cal C}_{\tau_{ij}} \right\ket&= &\sum_{[k],[k']} \prod_{i<j} \tau_{ij}^{k_{ij}} \bar{\tau}_{ij}^{k_{ij}'} \left\bra C_{[k']} \right|\left. C_{[k]}\right\ket
 \eea
 we see that the generating functional contains information about the scalar products of the new intertwiners.

For general $n$ we notice that
\be
  \det(1 + T\overline{T}) = \det \bpm T & 1 \\ -1 & \overline{T} \epm 
\ee
and since $T$ is $n\times n$ antisymmetric we can express the determinant as the square of a Pfaffian as
\be
  \det(1 + T\overline{T}) = \left( 1 + \sum_{I} (-1)^{\frac{|I|}{2}} \mathrm{pf}(T_I)\mathrm{pf}(\overline{T_I}) \right)^2
\ee
where $I \subset \{1,...,n\}$, $|I| = 2,4,...$ up to $n$, and $T_I$ is the submatrix of $T$ consisting of the rows and columns indexed by $I$.   In particular we have  $\mathrm{pf}(T_{\{i,j\}}) = \tau_{ij}$ and  for $I = \{i,j,k,l\}$ 
\be
  R_{ijkl} \equiv \mathrm{pf}(T_{\{i,j,k,l\}}) = \tau_{ij}\tau_{kl} + \tau_{ik}\tau_{lj} + \tau_{il}\tau_{jk}.
\ee
By the pfaffian expansion formula for $|I| > 4$ $\mathrm{pf}(T_{I})$ consists of terms, all of which contain a factor $R_{ijkl}$ for some $1\leq i<j<k<l \leq n$.  For instance $\mathrm{pf}(T_{\{1,2,3,4,5,6\}}) = \tau_{12}R_{3456} - \tau_{13}R_{2456}+\cdots$.  Therefore if $\tau_{ij} =[z_{i}|z_{j}\ket$ then we have $\binom{n}{4}$ relations $R_{ijkl} = 0$ in which case the scalar product has the form
\be \label{eqn_A_equals_G}
 \left\bra {\cal C}_{[z_{i}|z_{j}\ket}| {\cal C}_{[z_{i}|z_{j}\ket} \right\ket 
 = \left(1-\sum_{i<j} [z_{i}|z_{j}\ket \bra z_{i}|z_{j} ] \right)^{-2} 
 = {\cal A}_{\Theta_{n}}(z_{i},\check{z}_{i})
\ee
where $ |\check{z}_{i}\ket \equiv |z_{i}]$.
This shows that  when $\tau_{ij}=[z_{i}|z_{j}\ket$, we recover the amplitude ${\cal A}$ we computed initially.
This is not a coincidence, this is always true for any graph as we show in Proposition \ref{prop_gen_func_coh_amp}.

\subsection{Discrete-Coherent Amplitude Generating Functionals}

We would now like to define a generating functional for arbitrary $k$-basis contractions.  That is, we want to generalize (\ref{CC2}) from the theta graph, to arbitrary graphs.  This will involve a functional (\ref{defC}) for each vertex with the edges glued by integration with respect to $\rd \mu(z)$.  This is made precise in the following definition:

\begin{definition}
Given an oriented graph $\Gamma$ we define  a generating functional  that depends holomorphically 
on  parameters $\tau_{ee'}^{v}=-\tau_{e'e}^{v}$  associated with a pair of edges $e,e'$ meeting at $v$.
\be\label{defG}
\G(\tau_{ee'}^{v}) \equiv  \int \prod_{e\in E_{\Gamma}} \rd\mu(w_{e}) \prod_{v\in V_{\Gamma}}{\cal C}^{(v)}_{\tau_{ee'}^{v}}(w_{e})
\ee
where the integral is over one spinor per edge of $\Gamma$ and we integrate a product of intertwiners for each vertex $v$.
If $v$ is a $n$-valent vertex with outgoing edges $e_{1},\cdots, e_{k}$ and  incoming edges $e_{k+1},\cdots, e_{n}$ we define
\be
{\cal C}^{(v)}_{\tau_{ee'}^{v}}(w_{e}) \equiv{\cal C}_{\tau_{ee'}^{v}}( w_{e_{1}},\cdots, w_{e_{k}},\check{w}_{e_{k+1}},\cdots, \check{w}_{e_{n}}).
\ee
where ${\cal C}_{\tau}$ is defined in (\ref{defC}).  

The functional $\G(\tau_{ee'}^{v})$ is a Gaussian integral.  To make this explicit, we define a matrix $T^{\Gamma}$ whose entries  are  labeled by oriented edges of $\Gamma$.  The matrix elements of $T^{\Gamma}$ are given by:
\be \label{eqn_T_tau}
 T^{\Gamma}_{e_{1} e_{2}} = \tau_{e_{1}e_{2}}^{v}\quad \mathrm{if} \quad s(e_{1})=s(e_{2})=v,
\ee
while all the other matrix elements vanish.
This matrix is skew-symmetric
\be \label{eqn_T_antisym}
T^{\Gamma}_{e_{1}e_{2}} =-T^{\Gamma}_{e_{2}e_{1}}
\ee
The generating functional can be written as
\begin{align} \label{eqn_gen_func_T_gamma}
\G(\tau_{ee'}^{v}) 
  &= \int \prod_{e\in E_{\Gamma}} \rd\mu(w_{e}) \exp \Big\{ -\frac12 \sum_{e,e'}\big(
  \T_{e^{-1}e'} \bra w_{e}|w_{e'}\ket + \T_{e^{-1}e'^{-1}} \bra w_{e}|w_{e'}]  
  \\ & \hspace{170pt} -  \T_{ee'}[ w_{e} | w_{e'} \ket  - \T_{e e'^{-1}} [ w_{e} | w_{e'} ] \nonumber 
 \big)\Big\}
\end{align}
\end{definition}

Let us now explain how to get (\ref{eqn_gen_func_T_gamma}) from the definitions (\ref{defG}) and (\ref{eqn_T_tau}).  First note that two edges $e$ and $e'$ of $\Gamma$ can either share zero one or two vertices.
When two edges share a vertex there are four possible orientations of the edges at this vertex,
since each edge can be either incoming or outgoing.
 Taking all of these possibilities into account 
 we introduce the coefficients $\T_{ee'}$ which vanishes if $s(e)$ is different from $s(e')$ and is given by
 $\T_{ee'} \equiv \tau_{e e'}^{s_{e}}$ otherwise.
 If two edges meet at one vertex, one of the four coefficients 
 $\T_{ee'},\T_{e^{-1}e'^{-1}},\T_{ee'^{-1}},\T_{e^{-1}e'}$ is not zero.
 If two edges meet at two vertices then two such coefficients do not vanish.

 Finally to express explicitly the amplitude $\G$ as in (\ref{eqn_gen_func_T_gamma}) we need to take into account the orientation of the edges.  Using the convention of $z,\check{z}$ for outgoing,incoming edges and the identities $[\check{w}|w'\ket= -\bra w|w'\ket$, $[w|\check{w}'\ket= [w|w']$ and 
 $[\check{w}|\check{w}'\ket= -\bra w|w']$ the definition (\ref{defG}) translates into (\ref{eqn_gen_func_T_gamma}).  For an example see the generating functional of the scalar product (\ref{CC2}) where $\tau,\overline{\tau}$ correspond to $\tau^{v_1},\tau^{v_2}$ of the two vertices.

For the generating functional of the scalar product (\ref{CC2}) we found that when the variables $\tau$ satisfy the Pl\"ucker identity (\ref{eqn_tau_plucker}) that the generating functional $\G(\tau_{ee'}^{v})$ is equal to the fully coherent amplitude (\ref{eqn_coherent_gauss_int}) of the previous section; see (\ref{eqn_A_equals_G}).   We now prove that this is not a coincidence and hence applies to arbitrary graphs.
\begin{proposition} \label{prop_gen_func_coh_amp}
\be 
\G(\tau_{ee'}^{v}) = {\cal A}_{\Gamma}(z_{e}), \quad \mathrm{if} \quad \tau_{ee'}^{v} = [z_{e}|z_{e'}\ket \quad \mathrm{when}\quad s(e)=s(e')=v
\ee
\end{proposition}
\begin{proof}
The proof is straightforward;
we start from the definition (\ref{defC}) of ${\cal C}_{\tau}$ and notice that when $ \tau_{ee'}^{v} = [z_{e}|z_{e'}\ket$ this expression reads
\bea
{\cal C}_{[z_{e}|z_{e'}\ket} (w_{e}) =\sum_{[k]} (J+1)! \widehat{C}_{[k]}(z_{e}) \widehat{C}_{[k]}(w_{e}) 
= \sum_{j_{e}} \frac{(J+1)!}{(2j_{e})!} \int \rd g [z_{e}|g|w_{e}\ket^{2j_{e}}
\eea
where we have used (\ref{CC}) in the second equality.
Integrating out $w_{e}$ in (\ref{defG}) and using that $$\int \rd\mu(w)[z |g_{s}|{w}\ket^{2j} [z'|g_{t}|\check{w} \ket^{2j'} =
\int \rd\mu(w)[z |g_{s}|{w}\ket^{2j} \bra w |g_{t}^{-1}|z' \ket^{2j'}=  (2j)! \delta_{j,j'}[z|g_{s}g_{t}^{-1}|z'\ket^{2j},$$ we easily obtain that
\be
\G([z_{e}|z_{e'}\ket ) = \sum_{j_{e}} \frac{\prod_{v} (J_{v}+1)!}{\prod_{e}(2j_{e})!} 
\int \prod_{v\in V_{\Gamma}} \rd g_{v} [z_{e}|g_{s_{e}}g^{-1}_{t_{e}}|z_{e^{-1}}\ket^{2j_{e}} = {\cal A}_{\Gamma}(z_{e}).
\ee
\end{proof}

We now formulate our last main result which in analogy with Lemma \ref{thm_gauss} expresses $\G(\tau_{ee'}^{v})$ as an inverse determinant.
\begin{lemma} \label{thm_gen_gauss}
The generating functional $\G$ can be evaluated as an inverse determinant 
\be
\G(\tau_{ee'}^{v}) =\frac1{\det(E-T^{\Gamma})}
\ee
where 
 \be\label{defE}
E \equiv \bpm 0 & 1 \\ -1 & 0  \epm
\ee  
and the antisymmetric matrix $T^{\Gamma}$ is defined in (\ref{eqn_T_tau}).
\end{lemma}
\begin{proof}
Let us begin with (\ref{eqn_gen_func_T_gamma}) and note that the anti-symmetry properties of $\T_{e e'}$  are compatible with the symmetry properties of the spinor products.  
Expressing the spinors of (\ref{eqn_gen_func_T_gamma}) in terms of the two components $w_e = (\alpha_e,\beta_e)^t \in \C^2$ we get 
\begin{align}
 \G (\tau_{ee'}^{v}) 
  &= \int \prod_{e\in E_{\Gamma}} \rd\mu(\alpha_{e}) \rd\mu(\beta_{e}) 
  \exp\Big\{ -\frac12 \sum_{e,e'} 
  \Big( \T_{e^{-1}e' }(\overline{\alpha}_{e} \alpha_{e'} + \overline{\beta}_{e} \beta_{e'})  + \T_{e^{-1}e'^{-1}} ( \overline{\beta}_{e} \overline{\alpha}_{e'} -\overline{\alpha}_{e} \overline{\beta}_{e'})
   \nonumber \\
   &   \qquad \qquad \qquad \qquad\qquad \qquad \qquad \quad-  \T_{ee'}(\alpha_{e} \beta_{e'} - \beta_{e} \alpha_{e'}) - \T_{e e'^{-1}}(\alpha_{e} \overline{\alpha}_{e'} + \beta_{e} \overline{\beta}_{e'}) \Big) \Big\} 
  \nonumber \\ \nonumber
  &= \int \prod_{e\in E_{\Gamma}} \rd\mu(\alpha_{e}) \rd\mu(\beta_{e}) \exp \Big\{ -\sum_{e,e'} \Big( \overline{\alpha}_{e} A_{e e'} \alpha_{e'} + \beta_{e} B_{e e'} \alpha_{e'} + \overline{\alpha}_{e} C_{e e'} \overline{\beta}_{e'} +    {\beta}_{e} D_{e e'} \overline{\beta}_{e'} \Big) \Big\} 
\end{align}
where $\rd \mu(\alpha) = e^{-|\alpha|^2} \rd \alpha/\pi$ and
\begin{align}\nonumber
  A_{e e'} &=\frac12( \T_{e^{-1} e'}-\T_{e' e^{-1}})=  \T_{e^{-1} e'}, \qquad
  D_{e e'} = \frac12(\T_{e'^{-1} e}-\T_{e e'^{-1}})=\T_{e'^{-1} e} = A_{ee'}^{t} \\
  B_{e e'} &=  \frac12(\T_{e' e}-\T_{e e'} )= - \T_{e e'} , \qquad \nonumber\qquad \,\,\,\,\,
  C_{e e'} = \frac12(\T_{e^{-1} e'^{-1}} - \T_{e'^{-1} e^{-1}})=\T_{e^{-1} e'^{-1}}    
  \end{align}
  where $A^{t}$ denotes the transpose of $A$.
Performing the Gaussian integrations first of $\alpha$ and then of $\beta$ we get
\begin{align}
 \G(\tau_{ee'}^{v})
  &= \frac{1}{\mathrm{det}(1+ A)} \int \prod_{e\in E_{\Gamma}} \rd\mu(\beta_{e}) \exp\Big\{ - \sum_{e,e'} \Big({\beta}_{e} A^{t}_{e e'}  \overline{\beta}_{e'} - \beta_{e} (B(1+A)^{-1}C)_{ee'}
  \overline{\beta}_{e'} \Big) \Big\} \\
  &= \mathrm{det}(1+A)^{-1} \mathrm{det}\left( 1 + A^{t} - B(1+A)^{-1} C \right)^{-1} \\
  &= \mathrm{det}\bpm 1+A & 0 \\ B & 1  \epm^{-1} \mathrm{det}\bpm 1 & (1+A)^{-1}C \\ 0 & 1 + A^{t} - B(1+A)^{-1}C \epm^{-1}\\
  &= \mathrm{det}\bpm 1+A & C \\ B & 1 + A^{t} \epm^{-1} 
\end{align}
The  matrix $E$ introduced in (\ref{defE}) 
has a unit determinant; thus the previous determinant is also equal to the determinant of the antisymmetric matrix
\be
\mathrm{det}\left[E \bpm 1+A & B \\ C & 1 + A^{t} \epm  \right]^{-1} =
\mathrm{det}\bpm B  & (1+A^{t}) \\ -(1+A) & -C \epm^{-1} 
=\det(E - \T)^{-1}
\ee
which is what we desired to establish.
\end{proof}
We now are going to evaluate explicitly this determinant in much the same way as Theorem \ref{thm_amp}.  
With these definitions the generating functional is given by
\begin{theorem} \label{thm_gen}
We say two simple loops (see Def. \ref{def_loops_cycles}) are disjoint if they have no edges in common.
Given a simple loop $\ell= \{e_{1}, \cdots ,e_{n}\}$ we define the quantity
\be \label{eqn_A_loop}
A_{\ell}(\tau) = -(-1)^{|e|} \tau_{e_{1}^{-1} e_2}^{s(e_2)} \tau_{e_{2}^{-1} e_3}^{s(e_3)} \cdots \tau_{e_{n}^{-1} e_1}^{s(e_1)}
\ee
where $|e|$ is the number of edges of $l$ whose orientation agrees with the chosen orientation of  $\Gamma$.  
Finally, given a collection of disjoint simple loops $L= l_{1},..., l_{k}$ we define
\be 
A_{L}(\tau) = A_{\ell_{1}}(\tau)\cdots A_{\ell_{k}}(\tau).
\ee
Then the generating functional (\ref{defG}) has the following evaluation
\be \label{eqn_gen_loops}
\G(\tau) = \frac1{\left(1 + \sum_{L} A_{L}(\tau)\right)^{2}}
\ee
where the sum is over all collections of disjoint simple loops of $\Gamma$.
\end{theorem}
\begin{proof}
We now want to evaluate the determinant of $E - \T$. This is a anti-symmetric matrix of size $2N$ by $2N$ 
indexed by $e_1,...,e_N, e_{1}^{-1},...,e_{N}^{-1}$. Therefore this determinant can be evaluated as the 
square of the pfaffian of $E - \T$.
We cannot directly evaluate the Pfaffian of a matrix as a sum over cycles, however it is possible  following \cite{rote2001division} 
to write the product of pfaffians of two $2N$ by $2N$ antisymmetric matrices as
\be
  \mathrm{pf} A \cdot \mathrm{pf} B = \sum_{\text{C}} (-1)^{k} W_{A,B}(C)
\ee
where the sum is over cycle covers $C = c_1, ..., c_k$ of $\{1,...,2N\}$ having $k$ cycles and where each cycle is of even length.  
The weight of a cycle cover is the product of the weights of its cycles and the weight of a single cycle $c = (i_1, ..., i_{n})$ with $i_1 > i_2,...,i_{n}$ is given by 
\be
W_{A,B}(c) =   A_{i_1 i_2} B_{i_2 i_3} A_{i_3 i_4} B_{i_4 i_5} ... A_{i_{n-1} i_{n}} B_{i_n i_1}.
\ee
The specification of $i_1$ as the largest element in the cycle avoids any ambiguity in the definition of the weight.  
If one chooses $B =E$ then $\mathrm{pf}E=(-1)^{N(N-1)/2}$ then we have an expression for $\mathrm{pf}A$ in terms of cycle covers up to an overall sign.  
Let us therefore  set $A =E-T^{\Gamma}$ and let us choose $B=E$.

Lets start by evaluating the weight of a 2-cycle.  Since $E_{ij}$ is non-vanishing only if $j = i \pm N$ the weight must have the form
\be
A_{i_{1}+N, i_{1}} E_{i_{1}, i_{1}+N} = (E- \T)_{e_{1}^{-1}e_{1}} = - (1+\T_{e_{1}^{-1}e_{1}}).
\ee
Note that $\T_{e^{-1}e}\neq 0$ only if $e$ forms a 1-cycle (or bubble) at a vertex of $\Gamma$, i.e. $s(e)=t(e)$.
We have used the correspondence between $i_{1}= e_{1}$  and $i_{1}+N=e^{-1}_{1}$ if $i_{1}<N$.
This shows that $2$-cycles of $\{1,...,2N\}$ correspond to an evaluation in terms of 1-cycles of $\Gamma$.

Lets now consider a 4-cycle of $\{1,...,2N\}$. There are two possibilities depending on whether the second index is $i_{2}$ or $i_{2}+N$.
In the first case we get
\be \label{eqn_2_cycle_1}
 A_{i_1+N,  i_2} E_{i_2, i_2 + N} A_{i_2 + N, i_1} E_{i_1, i_1+N}=  \T_{e_{1}^{-1} e_{2}} \T_{e_{2}^{-1}e_{1}}.
\ee
In the second case we have
\be \label{eqn_2_cycle_2}
 A_{i_1+N,  i_2+N} E_{i_2+N, i_2} A_{i_2,  i_1} E_{i_1, i_1+N}= -\T_{e_{1}^{-1} e_{2}^{-1}} \T_{e_{2}e_{1}}.
\ee
In both cases we have used the fact that since $c$ is a cycle we necessarily have $i_{1} \neq i_{2}$.  Hence (because of the presence of $B_{i_{1},i_{1}+N}$)
we have that $e_{1} \neq e_{2}^{-1}$. This means that we can replace the element $ (E -  \T)_{e_{2}e_{1}}$ by $- \T_{e_{2}e_{1}}$.
One can now see that these weights correspond to 2-cycles of $\Gamma$.  The first case corresponds to the cycle of edges  $(e_{1}e_{2})$ while the second case corresponds to $(e_{1}e_{2}^{-1})$.  Clearly at most one of (\ref{eqn_2_cycle_1}) and (\ref{eqn_2_cycle_2}) is nonvanishing, since at most two of the elements of $\T$ are nonvanishing depending on the orientation.  
The difference in sign comes from $B_{i_2+N, i_2}=-1$ while $B_{i_1, i_1+N}=B_{i_2, i_2+N}=1$.  In effect we obtain a minus sign for each edge that { disagrees} with the orientation of $\Gamma$,
we also get a minus sign for every edge.  

This result generalizes easily now to the case of a $2n$-cycle of $\{1,...,2N\}$.
The same reasoning shows that the weight 
\be
W_{A,B}(c) = A_{i_1+N, i_2} E_{i_2, i_2\pm N} A_{i_2\pm N, i_3} E_{i_3, i_3\pm N}\cdots A_{i_{n-1}\pm N, i_n} E_{i_1, i_1+ N}.
\ee
is non zero if and only if the sequence of edges  $(e_{1},\cdots, e_{n})$ corresponds to a simple loop $\ell$ of $\Gamma$ of length $n$.
In that case 
\be 
W_{A,B}(c) = (-1)^{n- |\bar{e}|} \T_{e_{1}^{-1} e_{2}} \T_{e_{2}^{-1} e_{3}} \cdots \T_{e_{n}^{-1} e_{1}} =  - A_{\ell}(\tau)
\ee
and $|\bar{e}|$ is the number of times $i_j > N$ in which case $B_{i_j, i_j-N} = -1$.  Again this corresponds to traversing the edge $e_{j}$ in the  orientation opposite to the one  of $\Gamma$
 thus $|\bar{e}|$ is the number of edges in $c$ which { disagrees} with the orientation of $\Gamma$.
 We denote by $|e| =n-|\bar{e}|$ the number of edges of $c$ that { \it agrees} with the orientation of $\Gamma$.
This establishes the correspondence between $2n$-cycles $c$ of $\{1,...,2N\}$ and simple loops of $\Gamma$ of length $n$, moreover the weight for a simple cycle 
is precisely minus the amplitude of the loop in $\Gamma$.

A cycle cover $\cC$ on $\{1,...,2N\}$ consists of a disjoint union of 2-cycles and non-trivial (i-e the cycles which are not 2-cycles)  cycles of $\{1,...,2N\}$.
We established that each 2-cycle of $\{1,...,2N\}$ has a weight in the sum given by $(1+T_{e_{i}^{-1}e_{i}})$ where $(e_{i}^{-1}e_{i})$ correspond to a bubble in $\Gamma$.
We also established that each nontrivial cycle on $\{1,...,2N\}$ (with non-zero weight) corresponds to a simple loop of $\Gamma$ with amplitude $A_{\ell}$.
This shows that $\mathrm{pf}(E-\T)$ is (up to an overall sign) equal to
$$\sum_{L} \prod_{v\notin L} \left(\prod_{s(e)=v=t(e)}(1+ T_{e^{-1}e})\right) A_{L}(\tau) $$
where the sum is over disjoint union of simple loops of length at least 2 and the product is over all
vertices not in $L$, with a weight given by the product over the bubbles touching $v$ (and with the convention that the weight is $1$ if there is no bubbles).
Now if $T_{e^{-1}e}$ is non zero this means that $(e^{-1}e)$ is a positively oriented bubble; that is a simple loop of length 1.
Therefore expanding the previous product we get that the pfaffian of $(1+\T)$ is (up to an overall sign) equal to
\be
\sum_{L} A_{L}(\tau) \ee
where the sum is over disjoint union of simple loops of  any length, which is what we desired to establish.
\end{proof}
Note that this result for the generating functional $\G(\tau)$ is very similar to the first theorem \ref{thm_amp} we established in the first section
for the coherent amplitude ${\cal A}_{\Gamma}(z_{e})$  .
The key difference is that the coherent amplitude ${\cal A}_{\Gamma}(z_{e})$ involves a sum over cycles (non intersecting simple loops), while $\G(\tau)$ possesses a sum over the same cycles but also
simple loops that intersect at a vertex. The relation between the two theorems comes from the fact that if the Pl\"ucker  relation is satisfied then the sum of 
loops that meet at this vertex vanish. This is depicted graphically in Fig. \ref{fig_STU_box}
and it is established algebraically in the following Corollary which offers an alternative proof of Theorem \ref{thm_amp}.
\begin{corollary} \label{cor_gen}
If $\tau_{ee'}^{v} = [z_{e}|z_{e'}\ket$ where $s(e)=s(e')=v$ then
\be
{\cal G}_{\Gamma}([z_{e}|z_{e'}\ket) = \frac1{\left(1 + \sum_{C} A_{C}(z)\right)^{2}} = {\cal A}_{\Gamma}(z_{e}), 
\ee
where the sum is over all disjoint cycle unions of $\Gamma$.
\end{corollary}
\begin{proof}
Suppose a simple loop $U = (e_1e_{2} \cdots e_{i-1} e_i\cdots e_{n-1}e_n)$ is such that $s(e_1) = s(e_i) = v$ and $t(e_{i-1}) = t(e_{n})=v$, i.e. it intersects itself at the vertex $v$.  
Then there exists another simple loop $T = (e_1e_{2} \cdots e_{i-1} e_{n}^{-1} e_{n-1}^{-1} \cdots  e_{i}^{-1})$ which also intersects itself at $v$.  Lastly, there exists a pair of simple loops $S = (e_1... e_{i-1})(e_i ... e_n)$ which share the vertex $v$.  The triple $S,T,U$ exhaust the collections of disjoint simple loops which have an intersection at $v$ and contain precisely the set of (unoriented) edges $\{e_1,...,e_n\}$.

Suppose that $p_1$ edges of $\{e_1,...,e_{i-1}\}$ and $p_2$ of $\{e_i,...,e_n\}$ agrees with the orientation of $\Gamma$.  
And lets introduce the amplitudes
\bea
T_{e_{1}\cdots e_{i-1}} \equiv \left(\tau_{e_{1}^{-1} e_{2}}^{s(e_{2})} \cdots \tau_{e_{i-2}^{-1} e_{i-1}}^{s(e_{i-1})}\right)
\eea
Then by the prescription (\ref{eqn_A_loop})
\begin{align}
  A_U = (-1)^{p_1+p_2+1}\, T_{e_{1}\cdots e_{i-1}} \tau_{e_{i-1}^{-1} e_{i}}^{v}
 T_{e_{i}\cdots e_{n}}\tau_{e_{n}^{-1} e_{1}}^{v} \\
  A_T = (-1)^{p_1+p_2+n-i} \,T_{e_{1}\cdots e_{i-1}} \tau_{e_{i-1}^{-1} e_{n}^{-1}}^{v} T_{e_{n}^{-1}\cdots e_{i}^{-1}} \tau_{e_{i} e_{1}}^{v} \\
  A_S = (-1)^{p_1+p_2} \, \tau_{e_{1}\cdots e_{i-1}} \tau_{e_{i-1}^{-1} e_{1}}^{v} T_{e_{i}\cdots e_{n}} \tau_{e_{n}^{-1} e_{i}}^{v} 
\end{align}
Using the antisymmetry property  of $\tau$ shows that $ (-1)^{n-i}T_{e_{n}^{-1}\cdots e_{i}^{-1}} = T_{e_{i}\cdots e_{n}}$ 
Thus
\bea
  A_S + A_T + A_U = (-1)^{p_1 + p_2} \, T_{e_{1}\cdots e_{i-1}} T_{e_{i}\cdots e_{n}}
  \nonumber  \left( \tau_{e_{i-1}^{-1} e_{1}}^{v} \tau_{e_{n}^{-1} e_{i}}^{v} + \tau_{e_{i-1}^{-1} e_{n}^{-1}}^{v} \tau_{e_{i} e_{1}}^{v} - \tau_{e_{i-1}^{-1} e_{i}}^{v} \tau_{e_{n}^{-1} e_{1}}^{v}  \right)
\eea
For clarity let $1 = e_{i-1}^{-1}$, $2 = e_1$, $3 = e_{n}^{-1}$, and $4 = e_i$ then the last factor
\be
  \left( \tau_{12}^{v} \tau_{34}^{v} + \tau_{13}^{v} \tau_{42}^{v} - \tau_{14}^{v} \tau_{32}^{v}  \right)
\ee
is the Pl\"ucker relation and vanishes under the hypothesis.  Hence the only collections of simple loops which survive
the identification $\tau_{ee'}=[z_{e}|z_{e'}\ket$ are ones which are non-intersecting and do not share vertices with other simple loops, i.e. they are  disjoint unions of non-trivial cycles.
\end{proof}

\begin{figure} 
  \centering
    \includegraphics[width=1\textwidth]{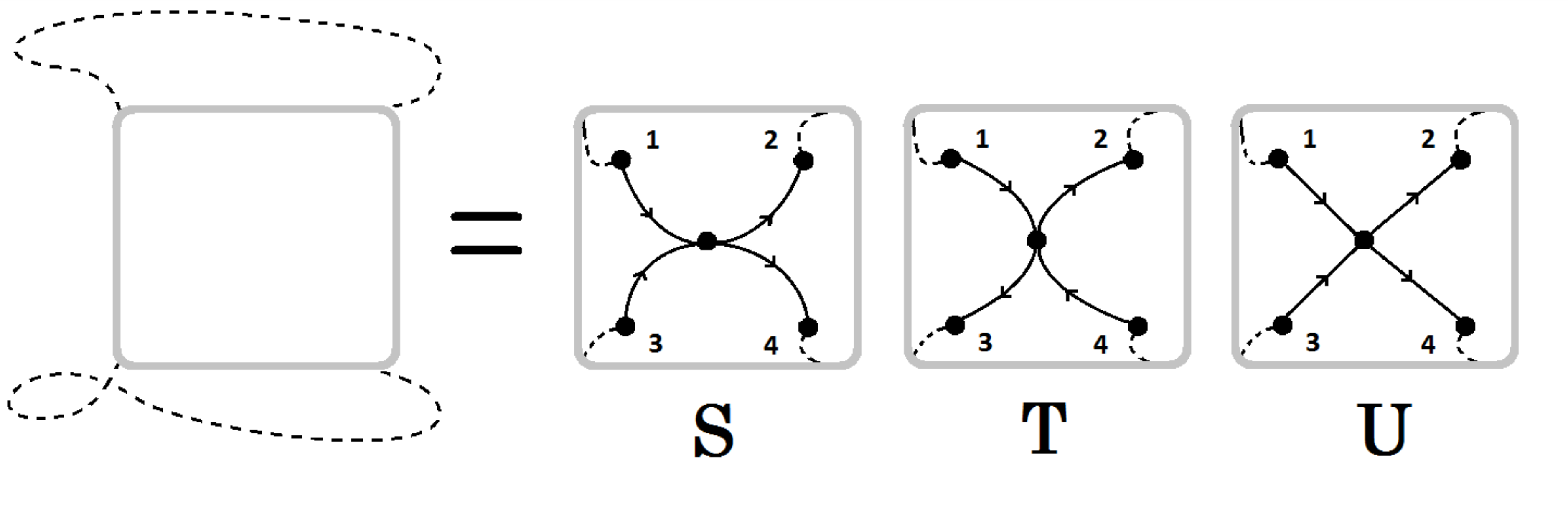}
    \caption{A simple loop depicted by the dashed line intersects itself at a vertex within the box.  In fact there are three possible collections of inequivalent simple loops which intersect at this vertex and have the same unoriented edges in common.  These three collections correspond to the three orientations S,T, and U of the four edges meeting at this vertex.  Note the following identification of vertices: S=(12)(34), T=(13)(42), U=(14)(32) which is an allusion to the Pl\"ucker relation.  An algebraic proof of how the amplitudes of intersecting simple loops arrange into the Pl\"ucker form is given in the proof of Corollary \ref{cor_gen}. }  \label{fig_STU_box}
\end{figure}

\section{Racah Formulae}
\label{section_racah}

Let us now discuss the amplitudes generated by the two generating functionals  ${\cal A}_{\Gamma}(z_{e})$ and ${\cal G}(\tau)$.  In particular, we explain how to extract the amplitude of an arbitrary spin network and how the amplitudes given by the two generating functionals are related by the Pl\"ucker relations.

The amplitude for a general graph, in the discrete basis, will depend on the integers $k_{ee'}^{v}$ associated with each pair of edges meeting at $v$ and is denoted 
$$A_\Gamma(k_{ee'}^{v})\equiv \underset{v\in V_\Gamma}{\corner} |k^{v}_{ee'}\ket.$$  
The fundamental relation (\ref{fundrelation}) between the two bases implies that these two amplitudes are related as follows
\begin{align} \label{eqn_discrete_amp}
  A_\Gamma(j_e,z_e) &= \sum_{k^{v}_{ee'}\in K_j} A_{\Gamma}(k_{ee'}^{v}) \prod_v \frac{(z_{e}| k^{v}_{ee'}\ket}{\|[k^{v}]\|^2} \\
  &=  \sum_{k^{v}_{ee'}\in K_j} A_{\Gamma}(k_{ee'}^{v})
\prod_{v}  \left(\frac{\prod_{(ee')\supset v} [z_{e}|z_{e'}\ket^{k_{ee'}^{v}}}{(J_{v}+1)!}\right). \nonumber
\end{align}



We would now like to evaluate these  amplitudes.  Comparing the coefficients of the same homogeneity in the coherent generating functional 
\be \label{eqn_gen_func}
{\cal A}_{\Gamma}(z_{e}) \equiv \sum_{j_e} \int \prod_{v\in V_\Gamma} \rd g_v (J_v + 1)! \prod_{e \in E_\Gamma} \frac{[z_e|g_{s_e}g_{t_e}^{-1}|z_{e^{-1}}\ket^{2j_e}}{(2j_e)!} = \frac{1}{(1+\sum_{C} A_C(z_e))^2}
\ee 
where $J_v$ is the sum of the spins at the vertex $v$.  

The sum is over collections $C = \{c_1,...,c_k\}$ of non-trivial cycles of the graph which are disjoint, i.e. do not share any edges or vertices with themselves or the other cycles.  The quantities $A_C \equiv A_{c_1} \cdots A_{c_k}$ are defined for each cycle $c_i = (e_1,...,e_n)$ by
\be
  A_{c_i}(z_e) \equiv -(-1)^{|e|} [\tilde{z}_{e_{1}} | z_{e_{2}}\ket [\tilde{z}_{e_2}| z_{e_3}\ket \cdots [\tilde{z}_{e_{n}}|z_{e_1}\ket
\ee
where $\tilde{z}_e \equiv z_{e^{-1}}$ and $|e|$ is the number of edges in the cycle which agrees with the orientation of $\Gamma$.  

Expanding in a power series we obtain 
\be
{\cal A}_{\Gamma}(z_{e}) = \sum_{k_{ee'}^{v}} R_{\Gamma}(k_{ee'}^{v}) \prod_{v} \left(\prod_{(ee')\supset v} [z_{e}|z_{e'}\ket^{k_{ee'}^{v}}\right).
\ee
where $R_{\Gamma}(k_{ee'}^{v})$ are the generalization of the Racah summation for an arbitrary graph
\be \label{eqn_power}
R_{\Gamma}(k_{ee'}^{v}) \equiv  \sum_{[M_C]} (-1)^{N+s} \frac{(N+1)!}{\prod_{C} M_C!}
\ee
where $N = \sum_C M_C$ and the sign $s$ accounts for the ordering of $ee'$ in $[z_{e}|z_{e'}\ket$.  The $M_C$ are positive integers labeled by each disjoint union of cycles $C$ and are summed over. These integers are restricted to depend on the $k_{ee'}^{v}$ by the relation
\be \label{eqn_kee}
  k_{ee'}^{v} = \sum_{C \supset (ee')} M_C,
\ee
where the sum is over all cycle unions $C$ which contain a cycle with the corners $(ee')$ or $(e'e)$.\footnote{Note that the solution of (\ref{eqn_kee}) is not unique since in  the number of cycles is usually greater than the number of independent $k_{ee'}$.  Therefore in general the coefficients $A_\Gamma(k_{ee'}^{v})$ will be given by a sum over arbitrary parameters. This leads  for example  to a summation over  one parameter for the tetrahedral graph, which corresponds to  the Racah expansion of the 6j symbol.  For the 4-simplex this will involve 17 parameters.}

On the other hand the relationship between continuous and discrete bases implies that the generating functional can also be expressed in terms of the discrete intertwiners as
\be
{\cal A}_{\Gamma}(z_{e}) = \sum_{k_{ee'}^{v}} A_{\Gamma}(k_{ee'}^{v}) \prod_{v} \left(\prod_{(ee')\supset v} [z_{e}|z_{e'}\ket^{k_{ee'}^{v}}\right).
\ee
This shows that  
$A_{\Gamma}(k_{ee'}^{v}) \simeq R_{\Gamma}(k_{ee'}^{v})$ 
where $\simeq$ is an equivalence relation on amplitudes $A_{\Gamma}(k^{v}_{ee'})$.
It is defined by  $A_{\Gamma}(k^{v}_{ee'})\simeq 0$ iff $\sum_{k^{v}_{ee'}}A_{\Gamma}(k^{v}_{ee'}) \prod_{v,(ee')} [z_{e}|z_{e'}\ket^{k_{ee'}^{v}}=0$. That is, it vanishes due to the Plucker relations when contracted with and summed over $\prod_{v,(ee')} [z_{e}|z_{e'}\ket^{k_{ee'}^{v}}$.

In order to find the  analog of the Racah formula  for the amplitude $A_{\Gamma}(k_{ee'}^{v})$ we need to use the more general generating functional 
\be
{\cal G}_{\Gamma}(\tau^{v}) \equiv \sum_{k_{ee'}^{v}} A_{\Gamma}(k_{ee'}^{v}) \prod_{v} \left(\prod_{(ee')\supset v} (\tau_{ee'}^{v})^{k_{ee'}^{v}}\right)
\ee
where $\tau^{v}_{ee'}$ are arbitrary complex parameters associated with pairs of edges meeting at $v$.  The expression of this generating functional is similar to (\ref{eqn_gen_func}). The only difference is that the sum is {\it not only} over unions of cycles, but also over unions of simple loops denoted by $L$.

A simple loop is loop of non overlapping edges.
The difference between loops and cycles is that cycles do not have any intersections.  Hence, the unions of cycles are a subset of the unions of loops.  This result implies that the amplitude $A_{\Gamma}(k_{ee'}^{v})$ can be expressed as a Racah sum over loops:
\be 
A_{\Gamma}(k_{ee'}^{v}) \equiv  \sum_{[M_L]} (-1)^{N+s} \frac{(N+1)!}{\prod_{L} M_L!}
\ee
where $M_{L}$ are integers labeled by each disjoint union of simple loops $L$, 
they are summed over with the restriction
$
  k_{ee'}^{v} = \sum_{L \supset (ee')} M_L,
$
while
$N = \sum_L M_L$.

As an application, we will give an explicit expression for the 20j symbol, which is independent of the 15j, as a generalized Racah formula.  

By solving (\ref{eqn_kee}) for $M_C$ in terms of $k^{v}_{ee'}$ we can derive a Racah formula for the amplitude of an arbitarary graph which is given by (\ref{eqn_power}).  Since there are 37 cycles $C$ in the 4-simplex and only 20 independent $k^{v}_{ee'}$ this formula will not be unique and will involve a sum over 17 parameters $p_k$.  The Racah formula is then
\be
  \{20j\}_{S_i,T_i} \simeq \sum_{p_1 \cdots p_{17}} \frac{(-1)^{N+s}(N+1)!}{\prod_{C} M_{C}(j_{ij},S_i,T_i,p_k)!}
\ee
where $N = \sum_C M_C$ and the sign $s = M_{1234}+M_{1235}+M_{1245}+M_{12354}+M_{12435}$ accounts for the edge ordering.  

In appendix \ref{20j_symbol} we give an explicit parameterization of the $M_C$ in terms of the $p_k$ although we note that simpler parameterisations might exist.  Furthermore, using various hypergeometric formulas one may be able to perform some of the summations over the $p_k$ explicitly.

\subsection{Racah Formulae for BF Theory}
\label{sec_BF_racah}

Starting from the expression (\ref{eqn_Z_BF_vertex_amps}) of SU(2) BF theory in terms of vertex amplitudes with the $k$ basis of intertwiners (\ref{C}) we obtain
\be
  Z^{\Delta^\ast}_{BF} = \sum_{j_f} \prod_f (2j_f+1) \sum_{k^{e}_{ff'} \in K_j} \prod_{e}\frac{1}{\|k^{e}_{ff'}\|^2} \prod_v A_{v}(k^{e}_{ff'}).
\ee
where $A_{v}(k^{e}_{ff'}) = \corner_{e \subset v} |k^{e}_{ff'}\ket$ is the contraction of the ``unnormalized'' intertwiners at each of the vertices.  Now since the amplitude vanishes if the spins of contracted intertwiners are different, we are free to sum over $j_{ef}$, that is over each edge of each face.  The reason to do this is that we can then combine the sum over spins and sum over $k^{e}_{ff'} \in K_j$ to just a simple sum over integers $k^{e}_{ff'}$ as
\be 
  Z^{\Delta^\ast}_{BF} = \sum_{k^{e}_{ff'}} \prod_f (2j_f+1) \prod_{e}\frac{1}{\|k^{e}_{ff'}\|^2} \prod_v A_{v}(k^{e}_{ff'}),
\ee
where $j_f = j_{ef}$ for some $e \subset f$.  Let us now use the Racah formula for the vertex amplitude as an expansion in simple loops\footnote{Using the equivalence relation among amplitudes generated by the Pl\"ucker relations, described in the previous section, we can choose to restrict to cycle unions instead of unions of simple loops, which are much fewer in number.}
\be
  A_{v}(k^{ve}_{ff'}) = \sum_{[M_L] \in M_k} (-1)^{\sum_L M_L S_L} \frac{(\sum_L M_L+1)!}{\prod_L M_L !}
\ee
where the sum is over $M_L$ with the restriction $k^{ve}_{ff'} = \sum_{L \supset (veff')} M_L$.  A simple loop $L$ of the boundary spin-network of the vertex is defined to be a closed path for which no link is traversed more than once.  In fact $L$ can be the union of non-overlapping simple loops.  The integer $S_L$ is determined by the orientation of the boundary spin network inherited from the orientation of the 2-complex in a well defined way.  

We can now combine the restricted sum over $[M_L]$ with the sum over $k^{e}_{ff'}$ to just get a simple sum over the integers $M_L$ at each vertex
\be \label{eqn_ZBF_M}
  Z^{\Delta^\ast}_{BF} = \sum_{[M^{v}_{L}]} \prod_f (2j_f+1) \prod_{e} \frac{\Delta_{e}}{\|k^{ve}_{ff'}\| \, \|k^{v'e}_{ff'}\|} \prod_v (-1)^{\sum_L M^{v}_L S^{v}_L} \frac{(\sum_L M^{v}_L+1)!}{\prod_L M^{v}_L!}.
\ee
Thus given a set of integers $[M^{v}_L]$ at each vertex of the 2-complex, a spin foam amplitude is given simply by a ratio of multinomial coefficients on the vertices and edges respectively, a dimension factor for the faces, and a sign from the orientation. 

The simplicity of this expression of $Z^{\Delta^\ast}_{BF}$ is that instead of computing a complicated amplitude for each vertex, such as a 6j, 15j, 20j or more general spin network amplitude, the vertex amplitude becomes a simple multinomial coefficient.  Second of all, instead of specifying spins (non-locally) to the faces of $\Delta^\ast$ we instead assign a set of integers $\{M^{v}_{L}\}$ for each simple loop of the boundary spin network at each the vertex.  

On the other hand, the difficulty with (\ref{eqn_ZBF_M}) is the determination of the signs  $S^{v}_{L}$ at each vertex from the orientation of the faces and ordering of strands in the spin foam edges.  These signs could be easily programmed into a computer, however it would be more desirable to combine the signs in a way that could be more easily determined from the spin foam data.

\chapter{Coarse Graining}
\label{chapter_coarse_graining}

In this chapter we study the operation of coarse graining at the level of the generating functional ${\cal G}(\tau)$.  We first give a non-trivial example by performing the well-known fusion move.  We find that the $\tau$ variables of coarse grained generating functional can be given by a sum over all paths between the boundary edges.  If there are loops in the bulk of the graph then the sum over all paths is infinite.  We then prove this result for an arbitrary graph.

In the last section we show that the generating functional ${\cal G}(\tau)$ on an infinite square lattice with a specific choice of edge orientation and vertex orderings is related to the high temperature loop expansion of the 2d Ising model.  It would be interesting to investigate the relation between the known coarse graining methods of the 2d Ising model with the one found here in terms of sums over paths.

\section{The Fusion Move}

Consider the so called Fusion move depicted in Figure \ref{fig_fusion}.  The result follows by the insertion of the resolution of identity on trivalent intertwiners which is trivially $\|j_1,j_2,j_3\ket\bra j_3,j_2,j_1\|$ since the trivalent intertwiner space is one dimensional; see (\ref{eqn_3j_state}).

\begin{figure} 
  \centering
    \includegraphics[width=0.9\textwidth]{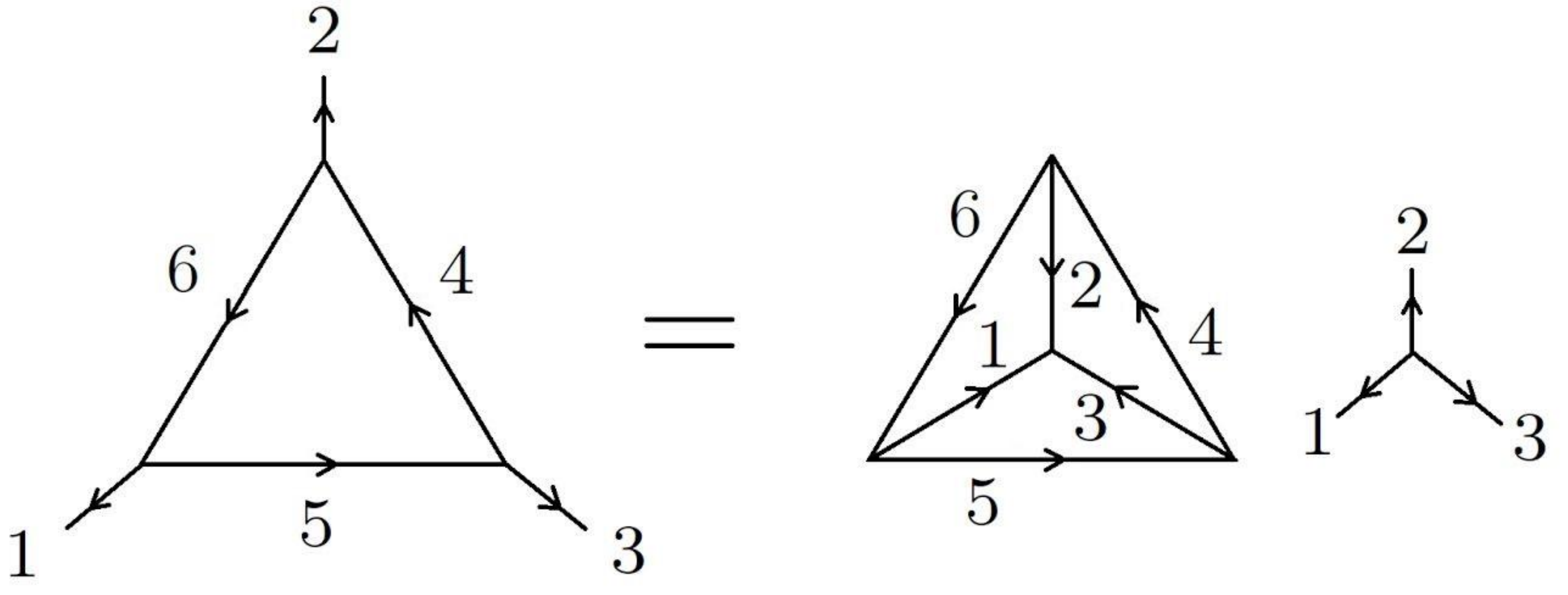}
    \caption{The fusion move.}  \label{fig_fusion}
\end{figure}

Let us now investigate this move in terms of the generating functional (\ref{defG}).  Using the prescription (\ref{eqn_T_tau}) with spinors $x_4,x_5,x_6$ on the internal edges, and spinors $z_1,z_2,z_3$ on the external edges
\be
  \mathcal{G}_{\Gamma}(\tau) = \int \rd\mu(z)\rd\mu(x) e^{S(z,x,\tau)}
\ee
with the initial action 
\be
  S = A^{ij}[\check{x}_{i}|x_{j}\ra + B^{ij}[\check{x}_{i}|z_{j}\ra + C^{ij}[z_{i}|x_{j}\ra
\ee
where 
\bea
A\equiv \left(\begin{array}{ccc}
 0 &0 & \tau_{46} \\
  \tau_{54 }& 0 & 0 \\
  0 & \tau_{65} & 0
  \end{array}\right),\quad
  B \equiv 
  \left(\begin{array}{ccc}
 0 &\tau_{42} & 0 \\
 0& 0 & \tau_{53} \\
  \tau_{61} & 0 & 0
  \end{array}\right), \quad
   C \equiv 
  \left(\begin{array}{ccc}
 0 &\tau_{15} & 0 \\
 0& 0 & \tau_{26} \\
  \tau_{34} & 0 & 0
  \end{array}\right)
\eea

We want to integrate over the internal edges $x$ leaving a generating functional depending on the external spinors $z$.  Including the measure $\sum_{i} [\check{x}_{i}|x_{i}\ra$ the integrals over $x$ are Gaussian and can be computed by evaluating the action on-shell.  The equation of motion for $x$ and $\check{x}$ are
\bea
(1-A)^{ij} |x_{j}\ra - B^{ij}|z_{j}\ra =0 \\
{[}\check{x}_{i}|( 1- A)^{ij} - [z_{i}| C^{ij} =0
\eea
Therefore, after integrating $x$ the action becomes equal to
\be
S' =  [z_{i}|z_{j}\ra (C(1-A)^{-1}B)^{ij}
\ee
The matrix $A$ has the property $A^{3}= \tau \one$ with $\tau \equiv \tau_{45}\tau_{56}\tau_{64} $.  Therefore, assuming the eigenvalues of $A$ are less than one, we can write the inverse as $(1-A)^{-1} = (1-\tau)^{-1}(1+A+A^{2})$ and compute explicitly the product.
 The final result is
\begin{align} \label{eqn_trivalent_action_coarse_grained}
  S' =  (1-\tau)^{-1}\Big( [z_{1}|z_{2}\ra \left(  -\tau_{16}\tau_{62} + \tau_{15}\tau_{54}\tau_{42} \right) + [z_{1}|z_{3}\ra \left(  \tau_{15}\tau_{53} + \tau_{16}\tau_{64}\tau_{43} \right) \hfill \\ \hfill + [z_{2}|z_{3}\ra \left(  -\tau_{24}\tau_{43} + \tau_{26}\tau_{65}\tau_{53} \right)\Big) \nonumber
\end{align}
The prefactor of the amplitude, i.e. the determinant of the inverse covariance matrix, is given by $(1-\tau)^{-2}$.  

Note that the expansion of the factor $(1+\tau)^{-1}$ in (\ref{eqn_trivalent_action_coarse_grained}) associated with the closed loop $(456)$ gives
\be
  (1-\tau)^{-1} = \sum_{l=0}^{\infty} (\tau_{45}\tau_{56}\tau_{64})^l
\ee
which adds powers of loops to the paths in (\ref{eqn_trivalent_action_coarse_grained}).  So in general we expect the coarse grained action to take the form
\be 
 S' = \sum_{i<j} [z_{i}|z_{j}\ra \left( \sum_{P_{ij}} \tau_{P_{ij}}\right) 
\ee
where the sum over path also includes the loops, that is it is allowed to go from $i$ to $j$ looping around several times around a loop and the factor $\tau$ contains an extra factor of $(-1)$ for every loop.
The prefactor is given by the determinant of the matrix corresponding to the subgraph obtained by deleting the external legs, i.e. the generating function of the amputated graph.

Let us now show that the coarse grained action does indeed give the RHS of the fusion move in Figure \ref{fig_fusion}.  Expanding the exponential of the new action $S'$ we find
\begin{align}
  (1-\tau)^{-2}e^{S'} = \sum_{k_{12},k_{13},k_{23}} \left(  \tau_{16}\tau_{62} + \tau_{15}\tau_{54}\tau_{42} \right)^{k_{12}} \left(  \tau_{15}\tau_{53} + \tau_{16}\tau_{64}\tau_{43} \right)^{k_{13}} &\left(  \tau_{24}\tau_{43} + \tau_{26}\tau_{65}\tau_{53} \right)^{k_{23}} \nonumber \\
	 \times (1-\tau)^{J-2} \frac{[z_{1}|z_{2}\ra^{k_{12}}}{k_{12}!} \frac{[z_{1}|z_{3}\ra^{k_{13}}}{k_{13}!} \frac{[z_{2}|z_{3}\ra^{k_{23}}}{k_{23}!} &
\end{align}
where $J = k_{12}+k_{13}+k_{23}$.  Now expanding all the binomials and letting the exponents of the $\tau_{ij}$ be given according to (\ref{eqn_3_k})
\begin{align}
  \prod_{(i,j)} \tau_{ij}^{k_{ij}} \equiv \, &\tau_{12}^{j_1+j_2-j_3}\tau_{13}^{j_1-j_2+j_3}\tau_{23}^{-j_1+j_2+j_3} 
	\tau_{15}^{j_1+j_5-j_6}\tau_{16}^{j_1-j_5+j_6}\tau_{56}^{-j_1+j_5+j_6} \nonumber \\
	&\tau_{24}^{j_2+j_4-j_6}\tau_{26}^{j_2-j_4+j_6}\tau_{46}^{-j_2+j_4+j_6} 
	\tau_{34}^{j_3+j_4-j_5}\tau_{35}^{j_3-j_4+j_5}\tau_{45}^{-j_3+j_4+j_5} \nonumber
\end{align}
we have
\be
  (1-\tau)^{-2}e^{S'} = \sum_{j_1,...j_6} W(j_1,...,j_6) \frac{[z_{1}|z_{2}\ra^{k_{12}}}{k_{12}!} \frac{[z_{1}|z_{3}\ra^{k_{13}}}{k_{13}!} \frac{[z_{2}|z_{3}\ra^{k_{23}}}{k_{23}!} \prod_{(i,j)} \tau_{ij}^{k_{ij}} 
\ee
where the coefficients are
\begin{align}
  W(j_1,...,j_6) &= \sum_{l} (-1)^{l+j_4+j_5+j_6} \binom{l+1}{l-j_1-j_2-j_3} \binom{j_1+j_2-j_3}{j_4+j_5+j_1+j_2-l} \nonumber \\
	  & \hspace{70pt} \times \binom{j_1-j_2+j_3}{j_4+j_6+j_1+j_3-l} \binom{-j_1+j_2+j_3}{j_5+j_6+j_2+j_3-l} \nonumber \\
		&= (-1)^{j_4+j_5+j_6}\frac{\Delta(j_1j_5j_6)\Delta(j_2j_4j_6)\Delta(j_3j_4j_5)}{\Delta(j_1j_2j_3)} \sixj{j_1}{j_2}{j_3}{j_4}{j_5}{j_6}
\end{align}
Hence the coarse grained generating functional does indeed give the 6j symbol coefficients of the fusion move.  The triangle coefficient factors are due to the normalization of the intertwiners in the 6j symbol and the sign is due to vertex orderings.  Now that we have done this example explicitly we can perform the same computation for a general open spin network generating functional.

\section{General Spin Network Coarse Graining}
\label{sec_coarse_graining}
Let us now investigate the coarse graining of a general spin network.
Let $\Gamma$ be a graph with open ends and we denote by $\hat{\Gamma}$ the closed graph obtained by removing from $\Gamma $ the open ends.
The edges of Gamma can be split into internal edges belonging to $\hat{\Gamma}$ and external edges  belonging to $\Gamma \backslash \hat{\Gamma}$.

The generating function associated with an open graph $\Gamma$ depends holomorphically on spinors $|z_{e}\ra$ for each external edge $e \in \Gamma \backslash \hat{\Gamma}$ 
which is {\it outgoing}, and depends antiholomorphically on
$|z_{e}]$ for each external edge $e\in \Gamma \backslash \hat{\Gamma}$ which is {\it ingoing}.
It also depends on a set on complex parameters $\tau^{v}_{ee'}$ associated with pairs of (oriented) edges $e,e'$ of $\Gamma$ meeting at the vertex $v$.
These parameters are antisymmetric in $ee'$, i-e $\tau_{ee'}=-\tau^{v}_{e'e}$
The generating function is then defined to be the integral
\be
\mathcal{G}(z_{e},\tau^{v}) \equiv \int \prod_{e\in \hat{\Gamma}} \rd \mu(x_{e}) \exp \left( - S (z_{e},x_{e}, \tau^{v}_{ee'}) \right)
\ee
where $\rd\mu(x) = \frac1{\pi^{2}} \rd^{2}\omega e^{-\la x | x \ra}$ is a Gaussian measure. 

Just like with the fusion move we can coarse grain the open spin network amplitude $A_\Gamma(z_e)$ into a single vertex by applying a resolution of identity on $n$-valent intertwiners (\ref{proj}). 

The action $S_{\Gamma}$ can be decomposed as a sum of actions $S_{v}$ associated with  each vertex of $\Gamma$, while $ S_{v} = \sum_{e,e' \cap v} S^{v}_{ee'}$ decompose itself as a sum over all pairs of edges meeting at $v$.  Each $S^{v}_{ee'}$ is a quadratic action which depends linearly on $\tau^{v}_{ee'}$ and $z_{e},z_{e'}$. This  action depends on the orientation of the edges meeting at $v$.

 The basic rule, following Definition \ref{defG}, is that it depends holomorphically on $|z_{e}\ra$ for each outgoing  edge $e$  
and depends antiholomorphically on
$|z_{e}]$ for each ingoing edge.
Explicitly we have 
\bea
S^{v}_{ee'}= -\tau^{v}_{ee'} [\check{z}_{e}|z_{e'}\ra,\quad \mathrm{if}\quad t(e)=s(e')=v \\
S^{v}_{ee'}= -\tau^{v}_{ee'} [{z}_{e}|z_{e'}\ra,\quad \mathrm{if}\quad s(e)=s(e')=v \\
S^{v}_{ee'}= -\tau^{v}_{ee'} [{z}_{e}|\check{z}_{e'}\ra,\quad \mathrm{if}\quad s(e)=t(e')=v \\
S^{v}_{ee'}= -\tau^{v}_{ee'} [\check{z}_{e}|\check{z}_{e'}\ra,\quad \mathrm{if}\quad s(e)=t(e')=v 
\eea
where $s(e)$ (reps. $t(e)$) denotes the starting (reps. terminal) vertex of the edge $e$.

In order to write the action in a concise form we introduce the 4-dimensional twistor  and its conjugate
\be
|Z_{e}) \equiv \left(\begin{array}{c}
|z_{e} \ra \\
|z_{e} ] \\
  \end{array}\right) , \quad |\check{Z}_{e}) \equiv \left(\begin{array}{c}
|\check{z}_{e} \ra \\
|\check{z}_{e} ] \\
  \end{array}\right)
\ee
together with a {\it symmetric} pairing
\be
(Z|W) \equiv \la z|w\ra +[z|w].
\ee
It will be important to note that the conjugate twistor is related in a simple manner to the original twistor
\be
|\check{Z}) = E|Z),\quad \mathrm{with}\quad E = \left(\begin{array}{cc}
0&1\\
-1&0
 \end{array}\right) 
\ee
so that $|Z)$ is a null twistor $(Z|\check{Z})=0$.

We also introduce a skew-symmetric  matrix $T_{\Gamma}$ whose entries are half edges of $\Gamma$.
Its matrix elements are such that 
$T_{\Gamma}^{ee'}$ vanishes unless $s(e)=s(e')$. When $s(e)=s(e')$ it is given by
\be
T_{\Gamma}^{ee'}= \tau_{ee'}^{s(e)}.
\ee
We finally introduce the $4$ by $4$ matrix  $\mathbb{T}^{\Gamma}$ as in Lemma \ref{thm_gen_gauss} to be
\be
\mathbb{T}_{\Gamma}^{ee'} \equiv 
\left(\begin{array}{cc}
T^{\bar{e}e'}_{\Gamma} & T^{\bar{e}\bar{e}'}_{\Gamma}\\
-T^{{e}e'}_{\Gamma} & - T^{{e}\bar{e}'}_{\Gamma}
 \end{array}\right)
\ee 
The antisymmetry of $T^{\Gamma}$ translate into  the property 
\be
 \mathbb{T}_{\Gamma}^{t} = - \mathbb{E}_{\Gamma} \mathbb{T}_{\Gamma} \mathbb{E}_{\Gamma},\quad \mathrm{with}\quad \mathbb{E}_{\Gamma}^{ee'} = \delta^{ee'} 
E
\ee
More precisely, this property implies that 
\be
\mathbb{T}_{\Gamma} = \left(\begin{array}{cc}
A_{\Gamma}&B_{\Gamma}\\
C_{\Gamma}&D_{\Gamma}
 \end{array}\right) \quad\mathrm{with}\quad A^{t}_{\Gamma}=D_{\Gamma},\,B^{t}_{\Gamma}=-B_{\Gamma},\,C_{\Gamma}^{t}=-C_{\Gamma}.
\ee
Finally we denote by $\mathbb{T}_{\hat{\Gamma}} $ the square sub-matrix consisting associated with the subgraph $\hat{\Gamma}$ and 
$\mathbb{T}_{\partial\Gamma}$ the rectangular sub-matrix pairing internal edges with external ones. 

We can use this data to rewrite the action associated with the graph $\Gamma$ in a simple form
given by:
\be
S_{\Gamma} = \frac12 (W_{a}|\mathbb{T}^{ab}_{\hat{\Gamma}}|W_{b}) + (W_{a}| \mathbb{T}_{\partial\Gamma}^{ai}|Z_{i}) 
\ee
where $a,b$ labels the internal edges and $i,j $ the external ones. 
In short this reads 
\be
S_{\Gamma} = \frac12 (W|\mathbb{T}_{\hat{\Gamma}}|W) + (W| \mathbb{T}_{\partial\Gamma}|Z)
\ee
 By varying the action with respect to $(W|$ one obtains
the equation 
\be
\left(1+\mathbb{T}_{\hat{\Gamma}}\right) |W) = - \mathbb{T}_{\partial\Gamma} |Z).
\ee
and varying with respect to $|W)$ one gets
 \be
(W| \left(1+ \mathbb{T}_{\hat{\Gamma}}\right)  = - (Z|\hat{\mathbb{T}}_{\partial\Gamma}, \quad \mathrm{with} \quad \hat{\mathbb{T}}\equiv E\mathbb{T}^{t}E^{t}.
\ee
Thus integrating over the twistors $|W) $ is straightforward and we obtain the coarse grained action
\be 
S_{\Gamma\backslash\hat{\Gamma}} =-\frac12 (Z_{i}| \left(\hat{\mathbb{T}}_{\partial\Gamma}(1+\mathbb{T}_{\hat{\Gamma}})^{-1}\mathbb{T}_{\partial\Gamma}\right)^{ij}|Z_{j})
\ee
and the generating function becomes
\be
\mathcal{G}(z_{e},\tau^{v}) = \frac{1}{\det(1+\mathbb{T}_{\hat{\Gamma}})}e^{S_{\Gamma\backslash\hat{\Gamma}}}.
\ee
Note that the inverse matrix can be expressed as an infinite sum $ (1+\mathbb{T}_{\hat{\Gamma}})^{-1}= \sum_{n}(-1)^{n} \mathbb{T}_{\hat{\Gamma}}^{n}$.
Now we can evaluate the matrix elements of $\mathbb{T}^{n}$.
It is straightforward to see that it can be expressed in terms of a sum over paths in $\hat{\Gamma}$,that is
\be
\left(\mathbb{T}_{\hat{\Gamma}}^{n}\right)^{ee'} =(-1)^{n-1} \sum_{P\in {\bf \hat{P}}_{ee'}^{n-1}}  \tau_{P} 
\ee
where ${\bf \hat{P}}_{ee'}^{n}$ denotes the set of paths in $\hat{\Gamma}$  which go from $e$ to $e'$ and are of length $n$ (that is containing $n$ edges besides $e$ and $e'$).
For an oriented  path $P= (ee_{1}e_{2}\cdots e_{n} e')$ with $t(e_{i})=s(e_{i+1})$, the factor $\tau_{P}$ is 
\be
\tau_{P} = (-1)^{a}\tau_{\bar{e}e_{1}} \tau_{\bar{e}_{1}e_{2}} \cdots \tau_{\bar{e}_{n}e'}
\ee
where $a$ is the number of edges of $P$ that agrees with the orientation of $\Gamma$.
Thus 
\be
(1+\mathbb{T}_{\hat{\Gamma}})^{-1}_{ab} = \sum_{P\in \bf{\hat P}_{ab}} \tau_{P},
\ee
where the sum is over all path in $\hat{\Gamma}$ going from $a$ to $b$.
The multiplication by  $ T_{\partial\Gamma}$ extend this sum to a sum over path in $\Gamma$
thus we get 
\be
S_{\Gamma\backslash\hat{\Gamma}} =\frac12 (Z_{i}|Z_{j}) \left(\sum_{P\in \bf{P}_{ij}} \tau_{P}\right) 
\ee
It would be interesting to know whether this general sum over paths can be factorized in terms of simple paths and loops as in (\ref{eqn_trivalent_action_coarse_grained}).

\section{The 2d Ising Model}
\label{sec_ising}

The 2d Ising model on a square 2d lattice describes the possible configurations of spins placed on the lattice sites which can take one of two possible orientations.  The intuition of R. Peierls \cite{peierls1936ising} was that the possible states of this model are given by all possible loops on the dual lattice, which represent the boundary between domains of aligned spins.  The energy associated with the creation of such a domain is given by
$$
  \Delta E = 2JL
$$ 
where $J$ is a coupling constant and $L$ is the number links in the boundary of the domain.  The partition function of the Ising model on a square lattice $\cL_N$ of size $N \times N$at zero magnetic field with one coupling constant $J$ is
\be \label{eqn_ising_part}
  Z_N(v) = \sum_{\{\sigma\}} \exp \left(  \beta J \sum_{(i,j)} \sigma_i \sigma_j \right)
\ee
where for each vertex $i$ the spins are $\sigma_i = \pm 1$, and the sum in the exponent is over nearest neighbors.  Using the identity
$$
  \exp(x \sigma_i \sigma_j) = \cosh (x) ( 1- \sigma_i \sigma_j \tanh (x) )
$$
we get
\be
  Z_N(v) = \cosh^N( J) \sum_{\{\sigma\}} \prod_{(i,j)} (1 - \tanh( \beta J) \sigma_i \sigma_j)
\ee 
Expanding the product and defining $v \equiv \tanh \beta J$ 
\be \label{eqn_ising_loops}
  Z_N(v) = 2^N (1-v^2)^{-N} \left(1+\sum_{P \geq 4} g_P v^P \right)
\ee
where $g_P$ is the number of closed subgraphs $\Gamma^P_{\text{even}}$ of the lattice $\cL$ having a total of $P$ links and with an even number of edges adjacent to each vertex.  In what follows we define $\Gamma_{\text{even}}$ to be the set of such closed, even--valent subgraphs having an arbitrary number of links $P>4$.  Here a closed subgraph can have disconnected components which share neither edges nor vertices.  For more details see \cite{mussardo2010statistical}.

\begin{figure} 
  \centering
    \includegraphics[width=0.45\textwidth]{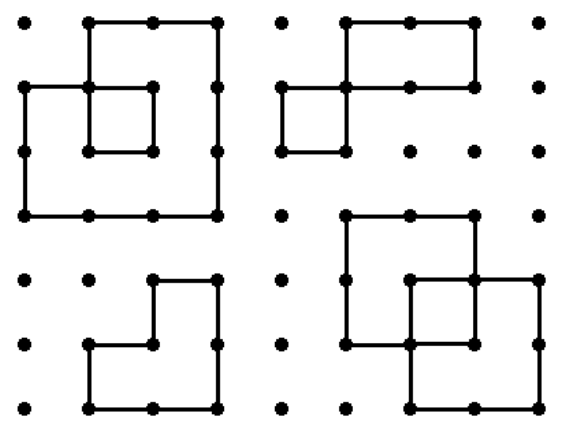}
    \caption{The closed even valent subgraphs of the square lattice correspond to domain boundaries in the 2d Ising model.}  \label{fig_polygon}
\end{figure}

\subsection{Matching of generating function and Ising model} \label{match}

The sum over collections of disjoint simple loops in Theorem \ref{thm_gen} on a square lattice contains all of these configurations of closed subgraphs $\Gamma_{\text{even}}$ in (\ref{eqn_ising_loops}), but also more due to the three possible ways in which two paths can cross at a four--valent vertex\footnote{See the middle three diagrams of Figure \ref{fig_simpleloops}.}.  Another difference is that there are signs in (\ref{eqn_gen_loops}) due to the edge orientation and the vertex ordering.  However, for a particular choice of edge orientation and vertex ordering of the square lattice, and a homogeneous choice of weights $\tau_{ee'} = i \sigma_{ee'} v$, with $\sigma_{ee'}$ being an antisymmetric function, the two sums are equal, as we now show.

\begin{theorem} \label{cor_gen_ising}
Let $\cL$ be the square lattice with edge orientation and vertex ordering as in Figure \ref{fig_lattice}.  Let the vertex weights in (\ref{defG}) be given homogeneously by $\tau_{ee'} = iv$ for $e<e'$ and $\tau_{ee'} = -iv$ for $e>e'$.  Then the spin network generating functional (\ref{defG}) takes the form
\be \label{eqn_G_L}
{\cal G}_{\cL}(iv) = \left(1+\sum_{P} g_P v^P \right)^{-2}
\ee
where the sum is over all even--valent, closed subgraphs of $\cL$ as in (\ref{eqn_ising_loops}).
\end{theorem}

\begin{figure} 
  \centering
    \includegraphics[width=0.7\textwidth]{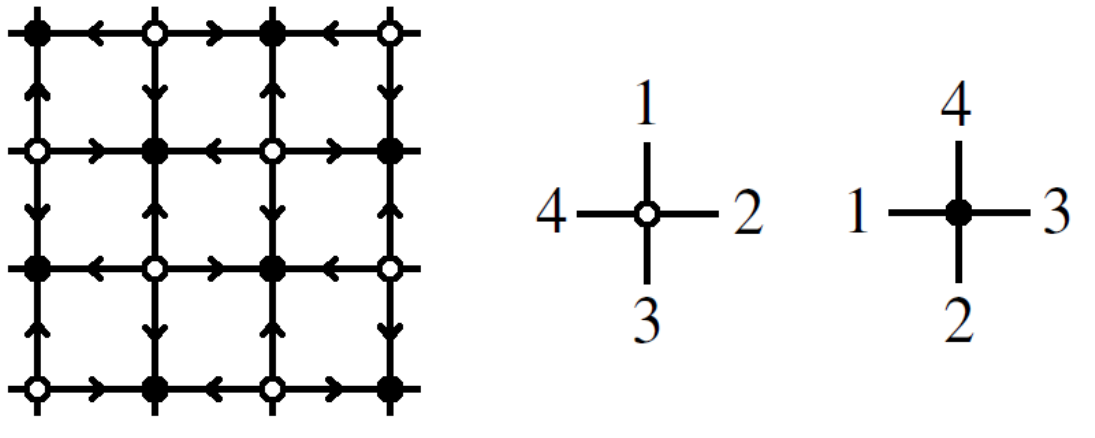}
    \caption{The edge orientation and vertex ordering of a square lattice for which the terms in (\ref{eqn_gen_loops}) all have a positive sign as shown in Theorem \ref{cor_gen_ising}.}  \label{fig_lattice}
\end{figure}

We will prove this theorem by a series of lemmas.  The first step is to control the signs in (\ref{eqn_gen_loops}) which is accomplished by the specific edge direction and vertex ordering in Figure \ref{fig_lattice}.  We say that a vertex $v$ in a loop disagrees with the vertex ordering, if the loop traverses first the edge $e$ and then the edge $e'$ adjacent to $v$ and $e' < e$.  Furthermore a loop without crossing is a loop which may have self intersections (i.e. four edges of the loop meet at one vertex), however the edges are traversed without leading to crossing edge pairs. 

\begin{lemma} \label{lemma_orient}
Let $\cL$ be the lattice in Figure \ref{fig_lattice} with the indicated edge orientation and vertex ordering.  Then 
\begin{enumerate}
\item the number of edges in a loop which agrees with the orientation of $\cL$ is equal to half the number of edges in the loop
\item the number of vertices in a loop without self--crossing, which disagrees with the vertex ordering is odd.
\end{enumerate}
\end{lemma}
\begin{proof}
For the first part, it is easy to see that the edges of every loop in $\cL$ alternates orientation and every loop has an even number of edges so the number of edges that agrees with the orientation is equal to half the number of edges in the loop.



For the second part, we will use induction on the number of plaquettes in the lattice.  To this end we will build up the lattice from the left most lower corner.  One can add squares so that the boundary on the right forms a staircase to reach an infinite lattice in the limit. A finite size lattice can be built row by row. We thus have two cases to consider: adding a square which starts a new row and adding a square to an existing row as is illustrated in Figure \ref{fig_lattice2}.
Notice furthermore that the ordering along a vertex is reversed if the loop is reversed, hence we need just to consider one specific loop orientation. Furthermore exchanging all black vertices with white ones and vice versa we also exchange all orientation induced signs, hence we again just need to consider one choice for the partitioning of the vertices into black and white.

One can check that the loop on a single square has an odd number of vertices which disagrees with the vertex ordering.  Assume that we have a square lattice for which every loop has an odd number of vertices which disagree.  
Consider adding a single square starting a new  row, as in the left panel of  Figure \ref{fig_lattice2}.
By the hypothesis all of the loops which contain $e_1$ have an odd number of vertices which disagree.  Traversing $e_1$ in any direction gives one vertex which disagrees.  On the other hand traversing the three edges in the new square clockwise gives three vertices which disagree (or one vertex that disagrees in the counter-clockwise direction).  Furthermore the new square might lead to a loop with a non--crossing self intersection at the black vertex $v_1$, shared by $e_1$. Here one can also check that for a counter--clockwise orientation of the loop a deformation of the loop to include the new square leads to four additional vertices that disagree.  
Hence all loops of the lattice with the new square also have an odd number of vertices which disagree.  

 Similarly, for adding a square to an existing row, one can check that traversing $e_2$, $e_3$ (or both) contributes the same parity as traversing the new square. Again one can also check that  loops with non--crossing self intersections at the black or white vertex of $e_3$, which include the new square, have an odd number of vertices disagreeing with the ordering. 
 
 Hence by induction the loops in a square lattice of any size will always have an odd number of vertices that disagrees with the ordering.
\end{proof}

\begin{figure} 
  \centering
    \includegraphics[width=0.55\textwidth]{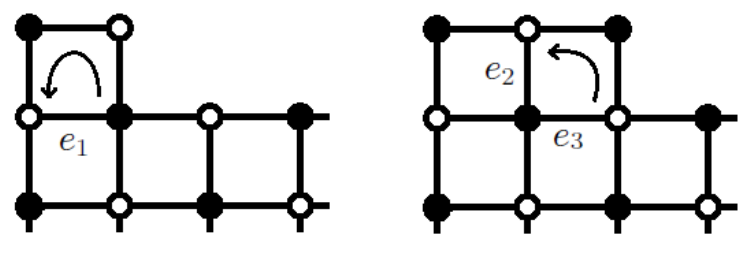}
    \caption{Adding one square to the lattice: Starting a new row and adding to an existing row.  Assuming all the loops in the existing lattice have an odd number of vertices which disagree with the edge ordering, then the loops containing the new square also have an odd number.}  \label{fig_lattice2}
\end{figure}

We have now to discuss the situation that at a given vertex either one loop self--intersects, or two loops touch or even cross each other. A priori all these cases are allowed to appear in the sum for the generating function (\ref{eqn_gen_loops}). This leads to three terms for such a vertex, as there are three possibilities for how two paths meet or cross at a four--valent vertex (see the middle three diagrams of Figure \ref{fig_simpleloops}). In the partition function of the Ising model (\ref{eqn_ising_loops}) only one term for  such a vertex appears. Hence we have to show that always two terms cancel each other, and that the surviving term does not lead to a loop with crossing.


\begin{lemma} \label{lemma_STU}
Consider the lattice $\cL$ and let 
\be
  \tau_{ee'} =\sigma_{ee'}\tau_{e}\tau_{e'}  
\ee
where $\sigma_{ee'} = 1$ if $e<e'$ and $\sigma_{ee'} = -1$ if $e>e'$ according to the vertex ordering.  

 Then the sum over collections of disjoint simple loops in (\ref{eqn_gen_loops}) is reduced to only collections without crossings.  Furthermore, there is a one--to--one matching between these terms and configurations $\Gamma_{\text{even}}$ of closed, even--valent subgraphs.  
\end{lemma}
\begin{proof}

Suppose we have a configuration $A_{L}(\tau)$ of disjoint simple loops, for which all four edges $e_1,\ldots,e_4$ adjacent to a vertex $v$ are shared by either one or two loops. The way the loop or the two loops traverses the four edges, leads to a partition of the  four edges into two pairs of consecutive edges in the loop(s).  There are three such possible pairings. The crossing case  $(1-3,2-4)$ and the two non--crossing cases $(1-2,3-4)$ and $(2-3,4-1)$. (Here the ordering of the edges inside a pair does not matter.)

Hence  there are also two other configurations, which include the same set of edges as $A_{L}(\tau)$, but differ by a certain rearrangement of the edges into loops, so that the other two pairings are obtained. This gives three configurations, which we will name $A_U$ for the crossing case,  $A_S$ for $(1-2,3-4)$  and $A_T$ for $(2-3,4-1)$.

To be concrete consider a black vertex, for white vertices one just has to invert the edges $e_1,\ldots e_4$ everywhere. Note that under a change of orientation of a simple loop we have $A_{\ell}=A^{-1}_{\ell}$ due to the anti--symmetry of the $\tau_{ee'}$ and the definition (\ref{eqn_A_loop}). Furthermore we can choose w.l.o.g.\ the initial vertex in any given loop. Hence we can assume that in the configuration $A_U$ we have a loop $\ell_{U1}$ of the form $\ell_{U1}=(e_3^{-1} P P' e_1)$ where $P$ and $P'$ stand for paths with the source vertex $s(P)$ given by $t(e_3^{-1})$ and the target vertex of $P'$ being $t(P')=s(e_1)$.

We now consider three possibilities for the end point of $P$.
\begin{itemize}\parskip-2mm
\item[(a)]  We have that the target vertex $t(P)=s(e_2)=s(P')$ with $P'=(e_2 e_4^{-1} p')$.
 \item[(b)] We have $t(P)=s(e_4)=s(P')$ with $P'=(e_4 e_2^{-1} p')$. 
\item[(c)] We have $t(P)=s(e_1)$.
In this case $P'$ is empty and there is a second loop $\ell_{U2}$ contributing to $A_U$ whose orientation and starting point we can choose such that $\ell_{U2}=( e_4^{-1} p'e_2)$ with $s(p')=t(e_4^{-1})$ and $t(p')=s(e_2)$. (The two loops intersect also elsewhere for a planar lattice.)
\end{itemize}

Let us  define the corresponding configurations $A_S$ and $A_T$ for the different cases. 
\begin{itemize}\parskip-2mm
\item[(a)] $A_S$ agrees with $A_U$ in all simple loops except for $\ell_{U1}$ which is replaced by $\ell_{S}=(e_2^{-1}P^{-1}e_3e_4^{-1}p'e_1)$. Likewise we replace for $A_T$ the loop  $\ell_{U1}$ by two loops $\ell_{T}=(e_4^{-1}p'e_1)(e_3^{-1}Pe_2)$.
\item[(b)] For $A_S$ we replace  $\ell_{U1}$ by a pair of loops $\ell_S=(e_{2}^{-1}p'e_1)(e_4^{-1}P^{-1}e_3)$ and for $A_T$ by a loop $\ell_T=(e_4^{-1}P^{-1}e_3e_2^{-1}p'e_1)$.
\item[(c)]  For $A_S$ we replace $\ell_{U1}\ell_{U2}$ by a loop $\ell_S=(e_2^{-1} (p')^{-1} e_4e_3^{-1}Pe_1)$ and for $A_T$ by a loop $\ell_T=(e_4^{-1}p'e_2e_3^{-1} Pe_1)$.
\end{itemize}

\begin{figure} 
  \centering
    \includegraphics[width=1\textwidth]{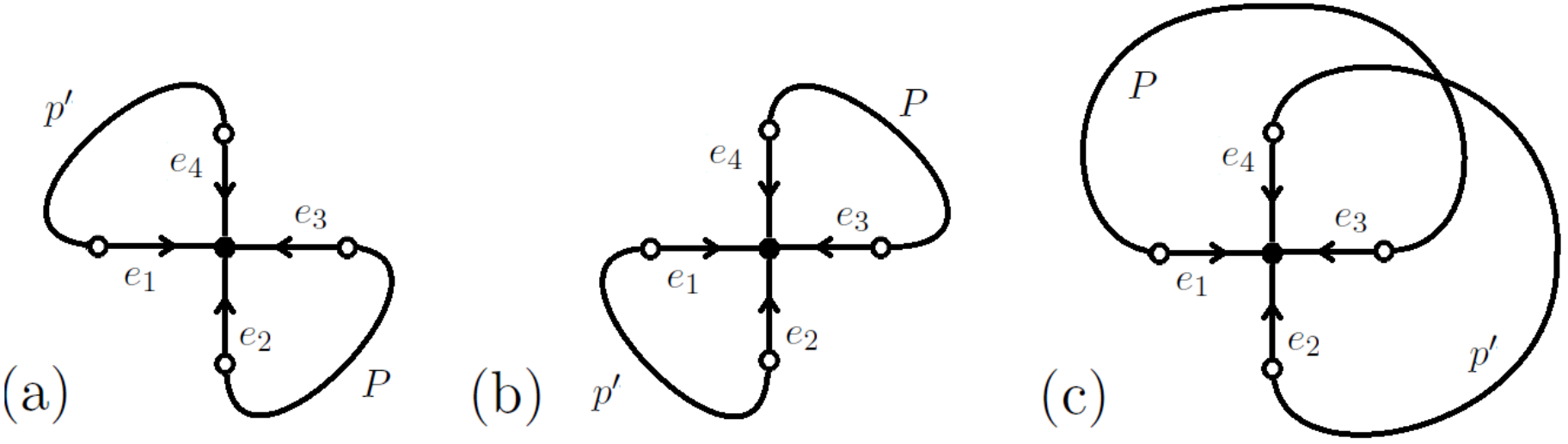}
    \caption{The three possible intersections at a 4-valent vertex.  For each intersection (a), (b), and (c) there are three possible configurations of simple loops $S$, $T$, and $U$.  The paths $P$ and $p'$ are arbitrary.}  \label{fig_abc}
\end{figure}

We have now to compare the corresponding amplitudes as defined in (\ref{eqn_A_loop}). To this end denote by 
\be
A_{{\cal P}} \,= \, (-)^{|{\cal P}|} \prod_{{ \text{bulk}} \,\, v} \tau_v 
\ee
 the contribution from an open path ${\cal P}$, where $|{\cal P}|$ is the number of edges disagreeing with the orientation of the path and $\tau_v$ stands for $\tau_{ee'}$ with $(e,e')$ a pair of edges in ${\cal P}$ adjacent to $v$ and ordered according to the orientation of $|{\cal P}|$. 
 
 Note that under a reversal of the orientation of ${\cal P}$ we have 
 \be
 A_{\cal P} \,=\, (-)A_{{\cal P}^{-1}} \quad .
 \ee
 The reason for this is that the change in sign due to the orientation of edges is given by $(-1)^{\sharp {\cal P}}$ where $\sharp {\cal P}$ is the number of edges in ${\cal P}$. Furthermore the change in sign due to the orientation of the vertices and the antisymmetry of the $\tau_{ee'}$ is given by $(-1)^{\sharp {\cal P}+1}$.
 
 We can now consider all three cases:\\
With $A_{\text{rest}}$ denoting the contribution of all other simple loops in $A_U$  we obtain for the case (a)
 \bea
 A_U&=& A_{\text{rest}}   (-) A_{(e_3^{-1}Pe_2)} \tau_{e_2 e_4^{-1}} A_{(e_4^{-1}p'e_1)} \tau_{e_1e_3^{-1}} \nn\\
 A_S&=&A_{\text{rest}}   (-) A_{(e_3^{-1}Pe_2)^{-1}} \tau_{e_3e_4^{-1}} A_{(e_4^{-1}p'e_1)} \tau_{e_1e_2^{-1}}\,=\,-A_U \nn\\
 A_T &=&A_{\text{rest}}   (-) A_{(e_3^{-1}Pe_2)} \tau_{e_2 e_3^{-1}} (-) A_{(e_4^{-1}p'e_1)} \tau_{e_1e_4^{-1}}\,=\,-A_U \quad .
 \eea
 Here we used the special form of the weights $\tau_{ee'}=\sigma_{ee'}\tau_{e}\tau_{e'}$ to reach $A_S=A_T=-A_U$.
 
 Likewise we also obtain for the other two cases (b) and (c) that $A_S=A_T=-A_U$.

  Thus for cases (a) and (b) we can  cancel in the sum $\sum_L A_L(\tau)$ the term with a crossing $A_U$ such that we remain with the contribution of two simple loops, i.e. for (a) we cancel $A_U$ with $A_S$ and for (b) we cancel $A_U$ with $A_T$. 
  
  In the case (c) we have to cancel $A_U$ with either $A_S$ or $A_T$ and we remain with a loop $\ell_S$ or $\ell_T$ with self--intersection (but non--crossing) at the vertex $v$ under consideration. 
  
  However, in the case of a planar lattice  the two loops $\ell_{U1}$ and $\ell_{U2}$ need to cross at least one other time at one or more other vertices $v',v'',\ldots$. Going to the next vertex, for instance $v'$, we can now resolve this crossing so that the loop is split into two loops. The self--intersection of $\ell_S$ or $\ell_T$ at $v$ then turns into two different loops sharing two vertices. 
  
  Doing this with all vertices we remain with loops which do not self--intersect. Different loops may share vertices. Counting all such configurations would still lead to an over--counting compared to the number of configurations of closed graphs $\Gamma_{\text{even}}$, as can be seen by an example of two loops sharing two vertices\footnote{See the lower right diagram in Figure \ref{fig_polygon}.}, for which there are two (if the loops are not crossing) possibilities involving the same set of edges. But in fact the proof shows that resolving all intersections leads always to just one configuration that remains in the end.  This leads to a matching of (left--over) loops configurations with configurations of closed, even--valent subgraphs $\Gamma_{\text{even}}$ for the Ising model.  
	
Remark: The fact that from the three possible terms $A_S,A_T,A_U$ two terms cancel out generalizes to arbitrary lattices. However to specify the crossing term $A_U$ one needs a planar vertex. Furthermore for i.e. six--valent vertices, three paths might meet at one vertex, in which case one has more terms to consider.	
\end{proof}


Now Theorem \ref{cor_gen_ising} follows from lemmas \ref{lemma_orient} and \ref{lemma_STU}.  Indeed, from Lemma \ref{lemma_STU} the sum in (\ref{eqn_gen_loops}) is reduced to a sum of terms in one--to--one correspondence with the subgraphs $\Gamma_{\text{even}}$ and each term in the sum is a collection of disjoint simple loops having no crossings.  Suppose such a subgraph has $P$ edges then by lemma \ref{lemma_orient} the quantity (\ref{eqn_A_loop}) will have a sign $(-1)^{P/2}$ which is canceled by the factors of $i$ in the weight.


This gives us the following relation between the spin network generating functional and the 2d Ising model partition function.
\be \label{eqn_G_Z}
  \cG_{\cL_N}(iv) = \frac{2^N}{ (1-v^2)^{N} Z_{N}(v)^2}
\ee
where $Z_N(v)$ is the partition function (\ref{eqn_ising_part}) of the 2d Ising model.  In particular, this shows that in the limit $N \rightarrow \infty$ the spin network generating functional $\cG_{\cL_N}(v)$ possesses a second order phase transition at $$v = \sqrt{2}-1$$

Indeed, it is known that the 2d Ising model undergoes a second order phase transition for a particular temperature, namely when $v = \sqrt{2}-1$.  The free energy of $Z_N(v)$ is defined by
\be
  F(T) = -kT \log Z_N(v)
\ee
and is exactly solvable for $N \rightarrow \infty$.  At the critical temperature the logarithm in $F(T)$ becomes singular and since
\be
  \log \cG_{\cL_N}(v) = N\log 2 - N \log(1-v^2) -2 \log Z_N(v)
\ee
it follows that the logarithm of $\cG_{\cL_N}(v)$ is also singular at this point.  Thus we have shown that the spin network generating functional $\cG_{\cL_N}(v)$ will undergo a second order phase transition at the critical value $v = \sqrt{2}-1$.


\chapter{Conclusion}

We have given a self contained overview of the two main bases of SU(2) intertwiners: the orthonormal and the coherent.  We have also defined a new basis which is discrete, coherent, and shares many of the nice properties of the former two bases.

The asymptotic limit of the closed 4-simplex amplitude in this new basis was computed and the semi-classical phase factor was found to give a generalization of the Regge action to Twisted Geometry.
 
We constructed generating functionals for the coherent spin network amplitudes and the amplitudes of the new basis, both for arbitrary graphs.  Using these generating functionals one is able to read off the exact evaluations of arbitrary spin networks in the orthonormal, coherent, and discrete coherent bases.

Finally we showed how coarse graining moves in the new basis could be performed at the level of the generating functional.  Interestingly, we showed how the discrete-coherent generating functional was related to the partition function for the 2d Ising model.  This could potentially be used as a ``Rosetta stone'' for investigations into renormalization of spin network and spin foam models.

We close by discussing some possible directions for future investigations and some open questions.    

\begin{itemize}

\item  In Section \ref{Sec_const_quant} we showed how the $|S,T\ket$ states could be derived by the constraint quantization of an auxillary Hilbert space, the states of which, do not satisfy the Pl\"ucker relations.  

The geometric interpretation of the intertwiner space is known in terms of the Grassmannian space $\text{Gr}(2,\C^{2N})$ of two planes in $\C^N$.  Indeed, in \cite{Freidel:2009ck} it is shown that the full intertwiner space on $N$ legs can be represented in terms of the Grassmannian as
\be
  \mathcal{H}_N = \bigoplus_{\{j_i\}} \mathcal{H}_{j_1,...,j_n} = L^2\left(\text{Gr}(2,\C^N) \right)
\ee 
where $\mathcal{H}_{j_1,...,j_n} \equiv \text{Inv}_{\text{SU(2)}}(V^{j_1} \otimes \cdots V^{j_n})$ was first defined in (\ref{eqn_inter_space}).  Furthermore, the Grassmannian can be embedded into complex projective space via the well known Pl\"ucker embedding 
\be
  \text{Gr}(2,\C^N) \hookrightarrow \mathbb{P}(\C^N \wedge \C^N) \cong \C\mathbb{P}^{\binom{N}{2}-1}
\ee
which maps 2d subspaces of $\C^N$ with coordinates $|z_i\ket$ to the points $x$ in projective space 
\be
  x = \sum_{i<j} [z_i|z_j\ket e_i \wedge e_j
\ee
where the coordinates $[z_i|z_j\ket$ are known as  Pl\"ucker coordinates.  As we saw many times in this thesis, the number of Pl\"ucker coordinates (for $N>3$) is greater than the dimension of $\text{Gr}(2,\C^N)$.  Indeed, they satisfy the Pl\"ucker relations.  

The geometric interpretation of these relations is that an element $x \in \mathbb{P}(\C^N \wedge \C^N)$ is in $\text{Gr}(2,\C^N)$ iff it is decomposable, which for $N=4$ takes the form  
\be
  x \wedge x = \Big( [z_1|z_2\ket[z_3|z_4\ket - [z_1|z_3\ket[z_2|z_4\ket + [z_1|z_4\ket[z_2|z_3\ket \Big) e_1 \wedge e_2 \wedge e_3 \wedge e_4 = 0
\ee
Hence the Pl\"ucker relations offer an algebraic criterion for the simplicity of a complex bivector in 4d.

Moreover, the Grassmannian can be defined by the projective variety defined by the Ideal generated by the Pl\"ucker relations.  Therefore, it seems that our auxillary Hilbert space should be the quantization of $\mathbb{P}(\C^4 \wedge \C^4)$ and we can project down to the intertwiner space by modding by the Pl\"ucker relation.  

This could provide an opportunity of building models based on algebraic properties, rather than the usual geometric critera such as closure, simplicity, etc.  At the very least, the guiding geometric principles could be translated into an algebraic formalism which might be better suited for solving some problems.

For example, Cluster Algebras are a new class of algebras which might be useful in this regard \cite{scott2006grassmannians}.  In the case of the Grassmannian, these algebras consist of monomials of Pl\"ucker coordinates and they ``mutate'' between each other via the Pl\"ucler relations.  These tools could help illuminate the role of Pl\"ucker relations in the amplitudes, and perhaps even reveal hidden discrete symmetries.

\item Regarding the Spin Foam program there seems to be two interesting options for using the $k$-basis.  The first is to impose simplicity constraints on spin(4) representations like in the usual spin foam models.  

From the semiclassical relation (\ref{eqn_sc_k}) between the $k$ basis and the 3d dihedral angles it is easy to see that these simplicity constraints should be of the form
\be
  k^{L}_{ij} = \rho^2 k^{R}_{ij} 
\ee
in order to imposing the usual simplicity constraints $\bra J^{L}_i \cdot J^{L}_j \ket = \rho^2 \bra J^{R}_i \cdot J^{R}_j \ket$.  This formulation was actually considered in \cite{Dupuis:2010iq} however this possibility was abandoned since constraints on the observables $(1-\cos\theta_{ij})/2 = \sin^2 (\theta_{ij}/2)$ lacked an interesting interpretation.  We now have a clear interpretation of these variables: they are precisely the $k$ variables.

The second option is to use the 20j symbol as a vertex amplitude and to break the topological invariance by imposing the shape matching constraints (\ref{eqn_shape_match}) on the boundary $k$ values.  This is in the spirit of the area-angle formulation of Regge Calculus (\ref{eqn_Regge_3d_angles}) proposed by Dittrich and Speziale \cite{Dittrich:2008va}.  Furthermore, the semiclassical limit would give the Regge action and a 4d interpretation almost by construction.

The $k$'s are related to the 3d dihedral angles in the asymptotic limit via (\ref{eqn_sc_k}).  The 3d dihedral angles are related to the 2d angles by (\ref{eqn_2d_3d_dihedral}) which are used in the shape matching constraints.  Hence the shape matching constraints give a set of constraints on the $k$'s.  

In \cite{Dittrich:2008va} it is claimed that the closure constraints and the shape matching constraints reduce the 10 areas and 30 dihedral angles to 10 independent variables in the 4-simplex.  The closure constraints are automatically imposed by the relation $\sum_{j} k_{ij} = A_i$ since
\be
  A_i = \sum_j k_{ij} = \frac{A_i}{A} \sum_{j} A_j \left( 1- \cos \theta_{ij} \right) = A_i - \frac{A_i}{A} \sum_j A_j \cos \theta_{ij}
\ee
which implies the closure relation
\be
  \sum_j A_j \cos \theta_{ij} = 0
\ee
There are 20 independent $k$'s in the 20j symbol with the conditions $\sum_{j} k_{ij} = A_i$ imposed.  Therefore the shape matching constraints should presumably give ten independent constraints on the twenty independent $k$'s leaving ten degrees of freedom.  This might give a new interesting form for the simplicity constraints.

In Section \ref{sec_BF_racah} we gave a formulation of the BF partition function which involved a sums over integers at each of the vertices.  Each integer corresponds to a cycle union on the spin network of the vertex amplitude.  Furthermore the vertex and edge amplitudes were shown to be given simply by multinomial coefficients.  The shape matching constraints would then be imposed by Kronecker deltas on these integers at each vertex.  Such a model would at least be very simple in its formulation.

\item  The renormalization of the Ising model is well studied analytically.  It would be interesting to investigate if this renormalization could be formulated by the sum over paths developed in (\ref{sec_coarse_graining}).  Furthermore, this begs the question of whether we can associate a physical interpretation of the $\tau$ variables in the generating functional by analogy with the Ising model.

The sum over paths in the spin network coarse graining is very reminiscent of the derivation of the Lieb-Robinson bound.  It would be interesting to calculate the correlation of observables at large graph distances and see how this translates from the Ising model to the spin network picture.

There is also the question of using the exactly soluble techniques to compute Pachner moves in non-topological spin foam models.  The degree of divergence of these amplitudes could in principle be extracted from the exact evaluation.

\end{itemize}

\appendix

\chapter*{APPENDICES}
\addcontentsline{toc}{chapter}{APPENDICES}

\chapter{Gaussian Integration} \label{Gauss_int_app}

Using the standard formula for one complex variable
\be
  \int_\C \frac{\rd^2 \alpha}{\pi} e^{-|\alpha|^2 + a \alpha + b \overline{\alpha}} = e^{ab}
\ee
and comparing coefficients we have the formula
\be \label{eqn_gaussian_int_one_var}
  \int_\C \frac{\rd^2 \alpha}{\pi} e^{-|\alpha|^2} \overline{\alpha}^k \alpha^{k'} = \delta_{k,k'} k!.
\ee
We can then combine use this twice to compute the spinor integral
\be
  \int_{\C^2} \rd \mu(z) \bra a | z \ket^k \bra z | b \ket^{k'} = \delta_{k,k'} k! \bra a | b \ket  
\ee
The delta function for holomorphic functions is also given by Gaussian integration as in
\be
  \int_{\C^2} \rd \mu(z) f(z) e^{\bra z | w \ket} = f(w)
\ee

Now consider
\be
  \int_{\C^n} \prod_{i=1}^{n} \frac{\rd^2 \alpha_i}{\pi} \, e^{-\sum{i,j} \bar{\alpha}_i A_{ij} \alpha_j} = \frac{1}{\det(A)} 
\ee
for $n$ spinors this is
\be
  \int_{\C^{2n}} \prod_{i=1}^{n} \rd \mu(z_i) \, e^{\sum{i,j} \bra z_i | A_{ij} | z_j \ket} = \frac{1}{\det(1-A)} 
\ee
where the identity matrix in $\det(1-A)$ comes from the measure $\rd \mu(z)$.

\chapter{Proofs}

\section{Proof of Lemma \ref{product2}}
\label{proj_kern_proof}
\begin{lemma}
The projector onto the kernel of $H$ is explicitly given by
\be 
\Pi_{j_{i}} = 1 +\sum_{N=1}^{\mathrm{min}(2j_{i})}\frac{(-1)^{N}}{N!} \frac{(J-N+1)!(J-2N+1)!}{(J+1)!^{2}} \, \hat{R}^{N}(\hat{R}^{\dagger})^{N}.
\ee
\end{lemma}
\begin{proof}
To prove this we will use start from the computation of the scalar product of the generating functionals (\ref{eqn_scalar_expanded}) and its evaluation (\ref{det4}) which reads
\be \label{eqn_ST_loops}
  \sum_{j_i,S,T} \sum_{j'_i,S',T'} \left\bra S,T \right|\left. S',T' \right\ket \prod_{i<j} \bar{\tau}_{ij}^{k_{ij}(j_i,S,T)} \tau_{ij}^{k_{ij}(j_i,S',T')} = \left( 1 - \sum_{i < j} |\tau_{ij}|^2 + |R(\tau)|^2 \right)^{-2},
\ee
where $R(\tau) = \tau_{12}\tau_{34} + \tau_{13}\tau_{42} + \tau_{14}\tau_{23}$.  Expanding the RHS of (\ref{eqn_ST_loops}) gives
\be
 \sum_{[k],N}  \frac{(-1)^N(J+N+1)! }{N !}  |R(\tau)|^{2N} \prod_{i<j} \frac{|\tau_{ij}|^{2k_{ij}}}{k_{ij}!}.
\ee
and by shifting $j_i \rightarrow j_i - N/2$ and using the relations (\ref{rel}) this becomes
\be
  \sum_{j_i,s,t,N}  \frac{(-1)^N(J-N+1)!}{N!} \sum_{S,T} \sum_{S',T'} R^{(s,t)}_{(S,T)}(N) R^{(s,t)}_{(S',T')}(N) \prod_{i<j} \frac{\tau_{ij}^{k_{ij}(j_i,S,T)}{\bar{\tau}}_{ij}^{k_{ij}(j_i,S',T')}}{k_{ij}(j_i-N/2,s,t)!} 
\ee
Now comparing coefficients of this with the LHS of (\ref{eqn_ST_loops}) gives the desired result, that is 
\be
\left\bra S,T \right|\left. S',T' \right\ket=
\sum_{j_i,s,t,N}  \frac{(-1)^N}{N!} \frac{(J-N+1)!}{\prod_{i<j} k_{ij}(j_i-N/2,s,t)!} R^{(s,t)}_{(S,T)}(N) R^{(s,t)}_{(S',T')}(N).
\ee
This can also be written as 
\be
\left\bra S,T \right|\left. S',T' \right\ket=
\sum_{j_i,s,t,N}  \frac{(-1)^N}{N!}||s,t||_{j_{i}-N/2}^{2 }  \frac{(J-N+1)!}{(J-2N +1)!} R^{(s,t)}_{(S,T)}(N) R^{(s,t)}_{(S',T')}(N).
\ee
 Using the expression (\ref{matrixel}) for the matrix elements of $\hat{R}^{N}$:
 \be
{}_{j_{i}} (S,T| \hat{R}^{N}  | s,t)_{j_{i}-N/2} = ||s,t||_{j_{i}-N/2}^{2 }  \frac{(J+1)!}{(J-2N +1)!} 
  R^{(s,t)}_{(S,T)}(N),
\ee
this expression reads
\bea
&&\sum_{N,s,t}\frac{(-1)^{N}}{N!} \frac{(J-N+1)!(J-2N+1)!}{(J+1)!^{2}} \, 
\frac{( S,T |\hat{R}^{N}|s,t)_{j_{i}-N/2}(s,t|(\hat{R}^{\dagger})^{N} |S',T')}{||s,t||_{j_{i}-N/2}^{2}}\nonumber \\
&=& \sum_{N=0}^{\mathrm{min}(2j_{i})}\frac{(-1)^{N}}{N!} \frac{(J-N+1)!(J-2N+1)!}{(J+1)!^{2}} \,\nonumber 
( S,T |\hat{R}^{N}(\hat{R}^{\dagger})^{N} |S',T')\\
&=& ( S,T | \Pi_{j_{i}} |  S',T' ).\nonumber
\eea
where we have used the decomposition of the identity in the second line.
\end{proof}

\section{Proof of Lemma \ref{Lemma_sum_T}}
\label{Lemma_sum_T_proof}

\begin{lemma}
\be 
  \sum_{T} R^{(s,t)}_{(S,T)}(N) = \delta_{s,S}. 
\ee
\end{lemma}
\begin{proof}
Writing $a = N+s-S$ and $b=N+t-T$ with $a,b\geq 0$ and $a+b\leq N$ we get
\bea
\sum_{T} R^{(s,t)}_{(S,T)}(N) = \sum_{b=0}^{N-a} \frac{(-1)^{b}N!}{a!b!(N-a-b)!} = \sum_{b=a}^{N} (-1)^{b-a} \binom{N}{b} \binom{b}{a} = \delta_{N,a} = \delta_{S,s},
\eea
where we used a standard binomial identity in the last step.  Similarly we use the same identity to show that
\bea
\sum_{S} (-1)^{k_{23}(j_i,S,T)} R^{(s,t)}_{(S,T)}(N) = (-1)^{k_{23}(j_i,s,t)} \sum_{a=0}^{N-b} \frac{(-1)^{a}N!}{a!b!(N-a-b)!} = (-1)^{k_{23}(j_i,s,t)} \delta_{T,t}.
\eea
\end{proof}

\section{Proof of Proposition \ref{ST_orthog}}
\label{ST_orthog_proof}
  
\begin{proposition}
\be
  \left\bra S\right|\left.S' \right\ket  = \frac{\delta_{S,S'}}{2S+1}  \Delta^{2}(j_{1}j_{2}S) \Delta^{2}(j_{3}j_{4}S), \hspace{12pt} \left\bra T\right|\left.T' \right\ket  = \frac{\delta_{T,T'}}{2T+1}  \Delta^{2}(j_{1}j_{3}T) \Delta^{2}(j_{2}j_{4}T),
\ee
where the triangle coefficients were given in (\ref{eqn_tri_coeff}).
\end{proposition}
\begin{proof}
We wish to perform the summation
\be
  \left\bra S\right|\left.S' \right\ket = \delta_{S,S'} \sum_{t,N} \frac{(-1)^N (J-N+1)!}{N!\prod_{i<j} k_{ij}(j_i-N/2,s,t)!}.
\ee
We can first evaluate the sum over $t$ by noticing that
\be
  \sum_{T} \frac{1}{\prod_{i<j} k_{ij}!} = \frac{1}{k_{12}! k_{34}! (k_{13}+k_{14})!(k_{23}+k_{24})!} \sum_{T} \binom{k_{13}+k_{14}}{k_{13}}\binom{k_{23}+k_{24}}{k_{23}}
\ee
and using Vandermonde's identity\footnote{For proof compare the coefficients in the expansion of $(1+x)^p(1+x)^q = (1+x)^{p+q}$.}
\be \label{vandermonde}
  \sum_{k} \binom{p}{k} \binom{q}{j-k} = \binom{p+q}{j}
\ee
we have
\be
  \sum_{T} \frac{1}{k_{ij}!} = \frac{(2S)!}{k_{12}!k_{34}!k_{1}!k_{2}!k_{3}!k_{4}!}.
\ee
where $k_{1} = k_{13} + k_{14} = j_1 - j_2 + S$, $k_{2} = k_{23} + k_{24} = j_2 - j_1 + S$, $k_{3} = k_{13} + k_{23} = j_3 - j_4 + S$, and $k_{4} = k_{14} + k_{24} = j_4 - j_3 + S$.  Therefore after changing the  variable $j_i$ to $j_{i}-N/2$ we have
\be
   \left\bra S\right|\left.S' \right\ket = \delta_{S,S'} \sum_{N} \frac{(-1)^N (J-N+1)!(2S)!}{N!(k_{12}-N)!(k_{34}-N)!k_{1}!k_{2}!k_{3}!k_{4}!}.
\ee
We can now perform the sum over $N$
\be
  \sum_N (-1)^{N} \frac{(J-N+1)!}{N!(k_{12}-N)!(k_{34}-N)!} = \frac{(j_1+j_2+S+1)!}{k_{12}!}\sum_N (-1)^{N} \binom{k_{12}}{N} \binom{J-N+1}{j_1+j_2+S+1}
\ee
using the identity\footnote{For proof compare the coefficients in the expansion of $(1+x)^p/(1+x)^{q+1} = (1+x)^{p-q -1}$ with $p < q$.}
\be \label{binom_id}
  \sum_{k} (-1)^{k} \binom{p}{k} \binom{j+q- k}{q} = \binom{j+q-p}{j}
\ee
this becomes
\be
  \sum_N (-1)^{N} \frac{(J-N+1)!}{N!(k_{12}-N)!(k_{34}-N)!} = \frac{(j_1+j_2+S+1)!(j_3+j_4+S+1)!}{k_{12}! k_{34}!(2S+1)!}
\ee
and finally
 \be
  \bra S | S' \ket
  = \delta_{S,S'} \frac{(j_1+j_2+S+1)!(j_3+j_4+S+1)!}{(2S+1)k_{12}!k_{34}!k_{1}!k_{2}!k_{3}!k_{4}!}.
\ee
\end{proof}

\chapter{20j Racah formula}
\label{20j_symbol}

In this section we give an explicit parameterization of the 37 $M_C$ in terms of the 17 parameters $p_k$.  We label the 37 cycles $C$ of the 4-simplex by ordered sets of vertices 1,...,5.  Choosing the following parameters: $p_{1} = M_{1324}, p_{2} = M_{1325}, p_{3} = M_{1345}, p_{4} = M_{1354}, p_{5} = M_{1435}, p_{6} = M_{1425}, p_{7} = M_{2345}, p_{8} = M_{2354}, p_{9} = M_{2435}, p_{10} = M_{12345}, p_{11} = M_{12543}, p_{12} = M_{13245}, p_{13} = M_{13254}, p_{14} = M_{13425}, p_{15} = M_{13524}, p_{16} = M_{14235}, p_{17} = M_{14325}$ we can solve for the $M_C$ in (\ref{eqn_kee}) to be
\begin{align}
  M_{123} &= k^{3}_{12}-p_{1}-p_{2}-p_{12}-p_{13}, \hspace{12pt} M_{145} = k^{1}_{45}-p_{6}-p_{5}-p_{16}-p_{17}, \nonumber \\
  M_{124} &= k^{4}_{12}-p_{1}-p_{6}-p_{16}-p_{15}, \hspace{12pt} M_{234} = k^{2}_{34}-p_{1}-p_{8}-p_{12}-p_{16}, \nonumber \\
  M_{125} &= k^{5}_{12}-p_{2}-p_{6}-p_{14}-p_{17}, \hspace{12pt} M_{235} = k^{2}_{35}-p_{2}-p_{7}-p_{13}-p_{17}, \nonumber \\
  M_{134} &= k^{1}_{34}-p_{1}-p_{4}-p_{13}-p_{15}, \hspace{12pt} M_{245} = k^{2}_{45}-p_{9}-p_{6}-p_{14}-p_{15}, \nonumber \\
  M_{135} &= k^{1}_{35}-p_{2}-p_{3}-p_{12}-p_{14}, \hspace{12pt} M_{345} = k^{4}_{35}-p_{7}-p_{3}-p_{10}-p_{11}, \nonumber \\
  M_{1234} &= k^{3}_{24}-k^{2}_{34}+p_{1}+p_{8}+p_{12}+p_{16}-p_{7}-p_{10}-p_{17}, \nonumber \\
  M_{1243} &= k^{3}_{14}-k^{1}_{34}+p_{1}+p_{4}+p_{13}+p_{15}-p_{3}-p_{14}-p_{11}, \nonumber \\
  M_{1245} &= k^{5}_{14}-k^{1}_{45}+p_{6}+p_{5}+p_{16}+p_{17}-p_{3}-p_{10}-p_{12}, \nonumber \\           
  M_{1254} &= k^{5}_{24}-k^{2}_{45}+p_{9}+p_{6}+p_{14}+p_{15}-p_{7}-p_{11}-p_{13}, \nonumber \\
  M_{12354} &= k^{4}_{15}-k^{1}_{45}+k^{2}_{45}-k^{5}_{24}+p_{7}+p_{5}+p_{11}+p_{16}+p_{17}-p_{4}-p_{14}-p_{15}-p_{9}, \nonumber \\
  M_{12453} &= k^{4}_{25}-k^{2}_{45}+k^{1}_{45}-k^{5}_{14}+p_{9}+p_{3}+p_{10}+p_{14}+p_{15}-p_{5}-p_{8}-p_{16}-p_{17}, \nonumber \\
  M_{12435} &= k^{4}_{23}-k^{2}_{34}+k^{1}_{34}-k^{3}_{14}+p_{8}+p_{3}+p_{12}+p_{16}+p_{11}-p_{4}-p_{9}-p_{13}-p_{15}, \nonumber \\
  M_{12534} &= k^{4}_{13}-k^{1}_{34}+k^{2}_{34}-k^{3}_{24}+p_{4}+p_{7}+p_{13}+p_{15}+p_{10}-p_{5}-p_{8}-p_{12}-p_{16}, \nonumber \\
  M_{1235} &= k^{3}_{25}-k^{2}_{35}+k^{1}_{45}-k^{4}_{15}+k^{5}_{24}-k^{2}_{45}+p_{2}+p_{13}+p_{9}+p_{4}+p_{14}+p_{15}-p_{8}-p_{5}-p_{11}-2p_{16}, \nonumber \\
  M_{1253} &= k^{5}_{23}-k^{2}_{35}+k^{1}_{34}-k^{4}_{13}+k^{3}_{24}-k^{2}_{34}+p_{8}+p_{5}+p_{12}+p_{16}+p_{2}+p_{17}-p_{9}-p_{4}-p_{10}-2p_{15}. \nonumber 
\end{align}        
where $k^{i}_{jk}$ are parameterized in terms of $S_i$ and $T_i$ as in (\ref{int1}) and (\ref{int2}) by the relations $2j_{ij} = \sum_{k} k^{j}_{ik}$ and $2j_{jk} = \sum_{i} k^{j}_{ik}$.



\bibliographystyle{plain}
\cleardoublepage 
\phantomsection  
\renewcommand*{\bibname}{References}

\addcontentsline{toc}{chapter}{\textbf{References}}

\bibliography{ref}

\nocite{*}

\end{document}